\documentclass[12pt]{article}
\usepackage{amsmath}
\usepackage{amssymb,amscd}
\usepackage{bm}
\usepackage{color}
\usepackage{wrapfig}

\def\thefootnote{\fnsymbol{footnote}}

\newlength{\minitwocolumn}
\catcode`\@=11
\long\def\@makefntext#1{
\protect\noindent \hbox to 3.2pt {\hskip-.9pt  
$^{{\eightrm\@thefnmark}}$\hfil}#1\hfill}               

\def\thefootnote{\fnsymbol{footnote}}
\def\@makefnmark{\hbox to 0pt{$^{\@thefnmark}$\hss}}    
        
\def\ps@myheadings{\let\@mkboth\@gobbletwo
\def\@oddhead{\hbox{}
\rightmark\hfil\eightrm\thepage}   
\def\@oddfoot{}\def\@evenhead{\eightrm\thepage\hfil
\leftmark\hbox{}}\def\@evenfoot{}
\def\sectionmark##1{}\def\subsectionmark##1{}}

\font\eightrm=cmr8

\newtheorem{theorem}{Theorem}[section]

\newtheorem{definition}[theorem]{Definition}

\newtheorem{lemma}[theorem]{Lemma}
\newtheorem{proposition}[theorem]{Proposition}

\newenvironment{proof}[1][Proof]{\begin{trivlist}
\item[\hskip \labelsep {\bfseries #1}]}{\end{trivlist}}
\newcommand{\qed}{\nobreak \ifvmode \relax \else
      \ifdim\lastskip<1.5em \hskip-\lastskip
      \hskip1.5em plus0em minus0.5em \fi \nobreak
      \vrule height0.75em width0.5em depth0.25em\fi}

\def\delzero{\delta_0}
\newcommand\hQ{\mbox{\boldmath $Q$}}

\newcommand\gh{{\rm gh}}

\newcommand{\rd}{\overleftarrow{\partial}} 
\newcommand{\ld}{\overrightarrow{\partial}} 
\newcommand{\sbv}[2]{{\{{{#1},{#2}}\}}}
\newcommand{\ssbv}[2]{{\{{{#1},{#2}}\}}}

\newcommand{\cfun}[1]{\langle #1 \rangle}
\newcommand{\bracket}[2]{\langle #1\,,#2\rangle}

\newcommand{\floor}[1]{{\lfloor #1 \rfloor}}

\def\ba{\mbox{\boldmath $A$}}

\def\bq{\mbox{\boldmath $q$}}
\def\bp{\mbox{\boldmath $p$}}

\def\bbx{\mbox{\boldmath $x$}}
\def\bbxi{\mbox{\boldmath $\xi$}}
\def\bbq{\mbox{\boldmath $q$}}
\def\bbp{\mbox{\boldmath $p$}}
\def\bbe{\mbox{\boldmath $e$}}

\def\bbd{\mbox{\boldmath $d$}}
\def\bbF{\mbox{\boldmath $F$}}

\def\bomega{\mbox{\boldmath $\omega$}}

\def\brho{\mbox{\boldmath $\rho$}}

\def\bR{\mbox{\boldmath $R$}}
\def\bC{\mbox{\boldmath $C$}}
\def\bZ{\mbox{\boldmath $Z$}}
\def\bPhi{\mbox{\boldmath $\Phi$}}

\def\bphi{\mbox{\boldmath $\phi$}}
\def\bvarphi{\mbox{\boldmath $\varphi$}}

\newcommand{\calB}{{\cal B}}

\newcommand{\calD}{{\cal D}}

\newcommand{\calL}{{\cal L}}
\newcommand{\calM}{{\cal M}}

\newcommand{\calO}{{\cal O}}

\newcommand{\calX}{{\cal X}}

\newcommand{\calMC}{{\cal MC}}

\newcommand{\Hom}{{\rm Hom}}

\newcommand{\Q}{{\kern.24em\vrule width.04em height1.4ex%
                 depth-.05ex\kern-.26em\mathsf Q}}
\newcommand{\C}{{\kern.24em\vrule width.04em height1.4ex%
                 depth-.05ex\kern-.26em\mathsf C}}

\newcommand{\Map}{{\rm Map}}
\newcommand{\ev}{{\rm ev}}

\def\bTheta{\mbox{$\widehat{\Theta}$}}
\def\bvartheta{\mbox{$\widehat{\vartheta}$}}

\newcommand{\beq}{\begin{equation}}
\newcommand{\eeq}{\end{equation}}
\newcommand{\bea}{\begin{eqnarray*}}
\newcommand{\eea}{\end{eqnarray*}}
\newcommand{\beqa}{\begin{eqnarray}}
\newcommand{\eeqa}{\end{eqnarray}}

\newcommand{\be}{\boldsymbol{e}}

\newcommand\eeqref{Eq.~\eqref}
\newcommand\qone{{\eta}}
\def\bqone{\mbox{\boldmath $\eta$}}

\newcommand\proj{{pr}}

\def\rc#1{\textcolor{red}{#1}}

\newcommand{\citenum}{\cite}

\begin{document}
\renewcommand{\thefootnote}{\alph{footnote}}


\baselineskip 0.7cm

\begin{titlepage}
\begin{flushright}
\end{flushright}

\vskip 1.35cm
\begin{center}
{\Large \bf
Lectures on AKSZ Sigma Models for Physicists 
}
\vskip 1.2cm
Noriaki IKEDA
\footnote{E-mail:\ 
nikeda@se.ritsumei.ac.jp
}
\vskip 0.4cm
{\it 
Department of Mathematical Sciences,
Ritsumeikan University \\
Kusatsu, Shiga 525-8577, Japan 
}

\vskip 0.4cm

\today

\vskip 1.5cm

\begin{abstract}
This is an introductory review of topological field theories (TFTs)
called AKSZ 
sigma models.
The AKSZ construction is a mathematical formulation for the construction and 
analysis of a large class of TFTs, inspired by the Batalin-Vilkovisky formalism of gauge theories.
We begin by considering a simple two-dimensional topological field theory and explain the ideas of the AKSZ sigma models.
This construction is then generalized and leads to a mathematical formulation of a general topological sigma model.
We review the mathematical objects, such as algebroids and supergeometry, that are used in the analysis of general gauge structures.
The quantization of the Poisson sigma model is presented 
as an example of a quantization of an AKSZ sigma model.
\end{abstract}
\end{center}
\end{titlepage}

\tableofcontents

\setcounter{page}{2}


\rm

\section{Introduction}
\noindent
This lecture note will present basics of so-called \textit{AKSZ 
(Alexandrov-Kontsevich-Schwarz-Zaboronsky)
sigma models}.
Though there are several reviews 
which present mathematical aspects of AKSZ construction and AKSZ sigma models
\cite{Kotov:2010wr, Cattaneo:2010re, Qiu:2011qr},
in this lecture, we will introduce these theories by 
the physics language and explain mathematical foundations
gently.
Thus, mathematical rigor will sometimes be sacrificed.

An AKSZ sigma model is a type of topological field theory (TFT).
TFT was proposed by Witten
\cite{Witten:1988ze, Witten:1988xj} as a special version of a 
quantum field theory.
After that, a mathematical definition has been provided
\cite{Atiyah:1989vu}.
Apart from it, this theory has been formulated by the 
(BRST and)
Batalin-Vilkovisky (BV) formalism 
\cite{Batalin:jr, Schwarz:1992nx, Schwarz:1992gs} 
of gauge theories.
The \textsl{AKSZ construction} 
\cite{Alexandrov:1995kv, Cattaneo:2001ys}
is a reformulation of a TFT in this direction.
Purpose of the latter formulation is
to analyze classical and quantum aspects
of topological field theories by 
the action principle and the physical quantization technique,
which is fundamental to the formulation of a gauge theory,
and to apply them to various physical and mathematical problems.

The AKSZ construction is a powerful formulation
since a large class of TFTs are constructed and unified by 
this construction.
These include known TFTs, such as the 
A-model, the B-model \cite{Witten:1991zz},
BF theory \cite{Horowitz:1989ng},
Chern-Simons theory \cite{Witten:1988ze},
topological Yang-Mills theory \cite{Witten:1988ze},
Rozansky-Witten theory \cite{Rozansky:1996bq},
the Poisson sigma model \cite{Ikeda:1993aj, Ikeda:1993fh, Schaller:1994uj},
the Courant sigma model \cite{Ikeda:2002wh, Hofman:2002rv, Roytenberg:2006qz},
and Schwarz-type TFTs \cite{Schwarz:1992nx, Schwarz:1992gs}.
Moreover, we find that the AKSZ sigma models contain more 
TFTs, which, for instance, have the structure of Lie algebroids, 
Courant algebroids, homotopy Lie algebras, or their higher generalizations.

We start this lecture note by explaining the simplest example
to introduce idea of the AKSZ construction,
which is the two-dimensional abelian BF theory,
First, we express this theory using the BV formalism.
Next, deformation theory is used to find
the most general consistent interaction term that
satisfies physical properties.
As a result, we obtain the Poisson sigma model, an important nontrivial two-dimensional topological sigma model of AKSZ type.
As another example, we also consider the BV formalism of an abelian BF theory 
in higher dimensions.
From the analysis of these models,
we identify the mathematical components of the AKSZ construction,
a QP-manifold.

In the next section, we explain the basic mathematical notion, 
a \textsl{QP-manifold}, a differential graded symplectic manifold.
It is a triple consisting of a graded manifold, a graded Poisson 
structure, and a coboundary operator called homological vector field.



Based on the QP-manifold structure, we construct
a sigma model as a map between two graded manifolds, from $\calX$ to $\calM$,
which is the AKSZ construction.
We discuss that structures of the target space and gauge symmetries
of AKSZ sigma models are derived from this QP-manifold.
We analyze the gauge symmetries of general forms of AKSZ sigma models,
which are deformations of abelian BF theories, and 
we will find that
the infinitesimal gauge symmetry algebras of these models are 
not Lie algebras.
This analysis leads us to the introduction of Lie algebroids 
and their generalizations as gauge symmetries of AKSZ sigma models.
The finite versions 
of these gauge symmetries
corresponding to Lie groups
are groupoids.
These mathematical objects which are not so familiar to
physicists are explained by using local coordinate expressions.

In the last part of this lecture note, two important applications 
of AKSZ sigma models are discussed.
One is the derivation of topological strings. 
The A-model and the B-model are derived by gauge-fixing 
AKSZ sigma models in two dimensions
\cite{Alexandrov:1995kv}.
The other application 
is the deformation quantization on a Poisson manifold.
The quantization of the Poisson sigma model on a disc provides
a star product formula on the target space
\cite{Kontsevich:1997vb, Cattaneo:1999fm}.
The second application is also 
a prototype of the quantization of AKSZ sigma models; although
such quantizations
have been successfully carried out in only a few cases,
this example is one such case.

This lecture note is organized as follows. 
In Section 2, the BV formalism of 
an abelian BF theory in two dimensions is considered 
and an interaction term is determined by deformation theory.
This theory is reconstructed by the superfield formalism.
In Section 3, an abelian BF theory in higher dimensions 
is constructed by the BV formalism and reformulated by 
the superfield formalism.
In Section 4, a QP-manifold, which is the mathematical object for the AKSZ 
construction is defined.
In Section 5,
important examples are listed.
In Section 6, the AKSZ construction is defined and explained.
In Section 7, we use deformation theory to obtain general consistent 
interaction terms for general AKSZ sigma models.
In Section 8, we express an AKSZ sigma model in local coordinates.
In Section 9, we provide some examples of AKSZ sigma models.
In Section 10, we analyze AKSZ sigma models on an open manifold.
In Sections 11 and 12, we discuss two important applications, review the derivation of the A-model and the B-model and present a deformation 
quantization on a Poisson manifold 
from the quantization of the Poisson sigma model.
Section 13 is devoted to discussing related works and areas of future investigation.

\section{Topological Field Theory in Two Dimensions}\label{sectoin2PSM}
\noindent
We begin by explaining the concept of the AKSZ construction by providing a simple example.
We consider an abelian BF theory in two dimensions 
and discuss its Batalin-Vilkovisky formalism.
A consistent interaction term is introduced by using deformation theory.
Finally, we present a mathematical construction of 
its interacting theory by using the AKSZ construction.

\subsection{Two-Dimensional Abelian BF Theory}
\noindent
The simplest topological field theory is a
two-dimensional abelian BF theory.
Let $\Sigma$ be a manifold in two dimensions
with a local coordinate $\sigma^{\mu}$ ($\mu = 0, 1$) and 
suppose that $\Sigma$ has no boundary.
Here, we will take the Euclidean signature.

Let $A_{\mu i}(\sigma)$ be a gauge field and let
$\phi^i(\sigma)$ be a scalar field,
where $i = 1, 2, \cdots, d$ is an index
on $d$-dimensional target space.
The action is as follows:
\begin{eqnarray}
S_A
= - \frac{1}{2} \int_{\Sigma} d^2 \sigma \ \epsilon^{\mu\nu} 
F_{0 \mu\nu i} \phi^i
= \int_{\Sigma} d^2 \sigma \ \epsilon^{\mu\nu} A_{\mu i} 
\partial_{\nu} \phi^i,
\nonumber
\end{eqnarray}
where $F_{0 \mu\nu i} = \partial_{\mu} A_{\nu i} - 
\partial_{\nu} A_{\mu i}$ is the field strength.
Note that the boundary integral vanishes.
The gauge symmetry of this theory is $U(1)$:
\begin{eqnarray}
\delta_0 A_{\mu i} &=& \partial_{\mu} \epsilon_i,
\qquad
\delta_0 \phi^i = 0,
\nonumber
\end{eqnarray}
where $\epsilon_i(\sigma)$ is a gauge parameter.

Let us consider the following problem.
We add terms to $S_A$ 
and deform the gauge symmetry $\delta_0$ as follows:
\begin{eqnarray}
&& S = S_A + S_I, 
\nonumber \\
&& \delta = \delta_0 + \delta_1.
\nonumber
\end{eqnarray}
We search for the consistent $S$ and $\delta$.
The new action $S$ and
the new modified gauge symmetry $\delta$ 
must satisfy the following 
two consistency conditions:
The action is gauge invariant, that is, $\delta S=0$;
and the gauge symmetry algebra is closed, at least under 
the equations of motion,
$[\delta_{\epsilon}, \delta_{\epsilon^{\prime}} ]
\approx
\delta_{[{\epsilon}, \epsilon^{\prime}]}
$.
Note that $\delta S=0$ must be satisfied without the equations of motion, 
but it is sufficient to satisfy the closedness condition for the gauge algebra,
$[\delta_{\epsilon}, \delta_{\epsilon^{\prime}} ]=\delta_
{[{\epsilon}, \epsilon^{\prime}]}
$
along the orbit 
of the equations of motion.

In order to construct a consistent field theory, 
physical conditions are imposed on $S$: It
is required to be 
diffeomorphism invariant, local and unitary.
Two actions are equivalent if they become
classically the same action when there is local replacement of the fundamental fields.
That is,
if two actions coincide, $\tilde{S}(\tilde{\Phi}) = S(\Phi)$,
under a local redefinition of the fields, $\tilde{\Phi} = f(\Phi)$,
then they are equivalent.
Moreover, we regard 
two theories as equivalent if they have the same gauge symmetry, 
i.e. ${\delta}_1 = 0$.
As required by a local field theory, 
we have a Lagrangian $\calL$ such that
$S = \int_{\Sigma} d^2 \sigma \calL$,
where $\calL$ is a function of the local fields.
We assume that $\calL$ is at most a polynomial 
with respect to a gauge field $A_{\mu i}$.
%

The problem is to determine the most general $S_I$ under the 
assumptions discussed above.
In order to unify the conditions 
$\delta S=0$ and 
$[\delta_{\epsilon}, \delta_{\epsilon^{\prime}} ]=\delta_
{[{\epsilon}, \epsilon^{\prime}]} + \mbox{(equations of motion)}$,
we use the BV formalism to formulate the theory.
This is the most general method for obtaining a consistent gauge theory.

Let us apply the BV formalism 
to this abelian BF theory
\cite{Henneaux:1992ig, Gomis:1994he}.
First, a gauge parameter $\epsilon_i$ is replaced 
by the Faddeev-Popov (FP) ghost $c_i$,
which is a Grassmann-odd scalar field.
\footnote{This is the Faddeev-Popov method
of the quantization of a gauge theory.}
The \textsl{ghost numbers} of the fields
$\Phi \in \{A_{\mu i}, \phi^i, c_i \}$,
$\gh \, \Phi$, are defined 
as $\gh \, A_{\mu i}=\gh \, \phi^i=0$ and $\gh \, c_i=1$.
The gauge transformation $\delta_0$ is changed to
a BRST transformation such that $\delta_0^2=0$
by replacement of the gauge parameter with the FP ghost. 
This condition imposes
$\delta_0 c_i=0$.

For each of the fields $\Phi$,
we introduce an antifield $\Phi^* \in \{A^{*\mu i}, \phi^*_i, c^{*i}\}$.
Compared to the corresponding field, the antifield has the opposite Grassmann properties but the same spin.
The ghost numbers of the antifields are defined by the equation
$\gh \, \Phi + \gh \, \Phi^* = -1$. For ghost number $-1$,
$A^{*\mu i}$ is a vector and $\phi_i^*$ is a scalar field.
$c^{*i}$ is a scalar field of ghost number $-2$.



\begin{table}[h]
\begin{center}

         		\caption{Ghost number and form degree of fields
and antifields}
		\label{tab:load}
         		\begin{tabular}{|l|rrrr|}
			\hline
			form degree \textbackslash ghost number 
& $-2$ & $-1$ & $0$ & $1$ \\
			\hline
			$0$	& $c^{*i}$ & $\phi_{i}^{*}$ 
& $\phi^i$ & $c_i$	\\
			$1$	&  & $A^{*\mu i}$ & $A_{\mu i}$ & 	\\
			\hline
		\end{tabular}
\end{center}
         \end{table}

Next, 
an odd Poisson bracket, 
called the \textsl{antibracket}, is introduced
as 
$\sbv{\Phi(\sigma)}{\Phi^*(\sigma^{\prime})} 
= - \sbv{\Phi^*(\sigma^{\prime})}{\Phi(\sigma)}
= \delta^2(\sigma-\sigma^{\prime})
$. 
It is written as
\begin{equation}
\sbv{F}{G}  \equiv 
\sum_{\Phi} \int_{\Sigma} d^2 \sigma 
\left(F \frac{\rd}{\partial \Phi(\sigma)}
\frac{\ld }{\partial \Phi^*(\sigma^{\prime})}  G
- 
F \frac{\rd}{\partial \Phi^*(\sigma)}
\frac{\ld}{\partial \Phi(\sigma^{\prime})} G
\right)
\delta^2(\sigma-\sigma^{\prime}),
\label{2dantibracket}
\end{equation}
where the differentiation is the functional differentiation, 
and 
$F \frac{\rd}{\partial \Phi(\sigma)} = 
(-1)^{(\gh F - \gh \Phi)(\gh \Phi)} 
\frac{\partial F}{\partial \Phi(\sigma)}$
 denotes right derivative and 
$\frac{\ld }{\partial \Phi^*(\sigma^{\prime})} F = 
\frac{\partial F}{\partial \Phi^*(\sigma^{\prime})}$
denotes left derivative.
The antibracket is graded symmetric and it satisfies the
graded Leibniz rule
and the graded Jacobi identity:
\begin{eqnarray}
&& \sbv{F}{G} = -(-1)^{({\gh F} + 1)({\gh G} + 1)} \sbv{G}{F},
\nonumber \\
&& \sbv{F}{G  H} = \sbv{F}{G} H
+ (-1)^{({\gh F} + 1){\gh G}} G \sbv{F}{H},
\nonumber \\
&& \sbv{F  G}{H} = F \sbv{G}{H}
+ (-1)^{{\gh G}({\gh H} + 1)} \sbv{F}{H} G,
\nonumber \\
&& (-1)^{({\gh F} + 1)({\gh H} + 1)} \sbv{F}{\sbv{G}{H}}
+ {\rm cyclic \ permutations} 
= 0,
\nonumber
\end{eqnarray}
where $F, G$, and $H$ are functions of $\Phi$ and $\Phi^*$.

Finally, the BV action $S^{(0)}$ is constructed as follows:
\begin{eqnarray}
S^{(0)} = S_{A} + (-1)^{\gh \Phi} \int_{\Sigma} \Phi^* \delta_0 \Phi
+ O(\Phi^{*2}),
\nonumber
\end{eqnarray}
where $O(\Phi^{*2})$ is determined order by order 
to satisfy $\sbv{S^{(0)}}{S^{(0)}}=0$, which
is called the classical master equation.
In the abelian BF theory, the BV action is defined by adding ghost terms
as follows:
\begin{eqnarray}
S^{(0)} = \int_{\Sigma} d^2 \sigma \epsilon^{\mu\nu} A_{\mu i} 
\partial_{\nu} \phi^i
+ \int_{\Sigma} d^2 \sigma A^{* \nu i} 
\partial_{\nu} c_{i}, 
\nonumber
\end{eqnarray}
and $O(\Phi^{*2})=0$.
It is easily confirmed that $S^{(0)}$ satisfies
the classical master equation.

The BRST transformation in the BV formalism 
is
\begin{eqnarray}
\delta_0 F[\Phi, \Phi^*] = \sbv{S^{(0)}}{F[\Phi, \Phi^*]},
\nonumber
\end{eqnarray}
which coincides with the gauge transformation on fields $\Phi$.
The explicit BRST transformations are 
\begin{eqnarray}
&& \delta_0 A_{\mu i} = \partial_{\mu} c_i,
\qquad 
\delta_0 A^{* \mu i } =\epsilon^{\mu\nu} \partial_{\nu} \phi^i,
\nonumber \\
&& \delta_0 \phi^{*}{}_i = \epsilon^{\mu\nu} \partial_{\mu} A_{\nu i},
\qquad \delta_0 c^{* i} = - \partial_{\mu} A^{* \mu i},
\label{2brst}
\end{eqnarray}
and zero for all other fields.
The classical master equation, $\sbv{S^{(0)}}{S^{(0)}}=0$, guarantees two 
consistency conditions:
gauge invariance of the action and 
closure of the gauge algebra.
Gauge invariance of the action is proved as 
$\delta_0 S^{(0)}=\sbv{S^{(0)}}{S^{(0)}}=0$.
Closure of the gauge symmetry algebra is proved as
$\delta_0^2 F =\sbv{S^{(0)}}{\sbv{S^{(0)}}{F}}  
= \frac{1}{2} \sbv{\sbv{S^{(0)}}{S^{(0)}}}{F} = 0$
by using the Jacobi identity.

\subsection{Deformation of Two-Dimensional Abelian BF Theory }
\noindent
The deformation theory of a gauge theory is a systematic method
for obtaining a new gauge theory from a known one 
\cite{Barnich:1993vg, Barnich:1994db, Henneaux:1997bm}.
Deformation theory within the BV formalism locally determines 
all possible nontrivial consistent interaction terms $S_I$.

We consider the deformation of 
a BV action $S^{(0)}$ to $S$ as follows:
\begin{eqnarray}
S = S^{(0)} + g S^{(1)} + g^2 S^{(2)} + \cdots, 
\label{deformexpansion}
\end{eqnarray}
under the fixed antibracket $\sbv{-}{-}$, where 
$g$ is a deformation parameter.
Consistency requires the classical master equation, $\sbv{S}{S}=0$, on 
the resulting action $S$.
Moreover, we require an equivalence relation, that is,
$S^{\prime}$ is equivalent to $S$
if and only if
$S^{\prime} = S + \sbv{S}{T}$,
where $T$ is the integral of a local term in the fields and antifields.
This condition corresponds to the physical equivalence
discussed in the previous subsection.
$S^{(n)}$ $(n=1, 2, \cdots)$ is determined order by order
by solving the expansions of the classical master equation 
with respect to $g^n$.
Invariance, locality and unitarity (the physical conditions discussed in the previous subsection) are required in order for the resulting action
to be physically consistent.
From these requirements, $S$ is diffeomorphism invariant on $\Sigma$,
it is the integral of a local function (Lagrangian) 
$\calL$ on $\Sigma$,
and it has ghost number $0$.

We substitute equation (\ref{deformexpansion}) into
the classical master equation $\sbv{S}{S}=0$.
At order $g^0$, we obtain $\sbv{S^{(0)}}{S^{(0)}}=0$.
This equation is already satisfied,
since it is the classical master equation 
of the abelian BF theory.

At order $g^1$, we obtain 
\begin{eqnarray}
\sbv{S^{(0)}}{S^{(1)}}=\delta_0 S^{(1)}=0.
\label{1storderofg}
\end{eqnarray}
From the assumption of locality, $S^{(1)}$ is an integral of 
a $2$-form $\calL^{(1)}$ such that
$S^{(1)} = \int_{\Sigma} \calL^{(1)}$.
Thus, equation (\ref{1storderofg}) requires that 
$\delta_0 \calL^{(1)}$ be a total derivative.
Then, the following equations 
are obtained by 
repeating the same arguments for the descent terms:
\begin{eqnarray}
&& \delta_0 {\calL}^{(1)} + d a_1 = 0,
\nonumber \\
&& \delta_0 a_1 + d a_0 = 0,
\nonumber \\
&& \delta_0 a_0 = 0,
\nonumber 
\end{eqnarray}
where $a_1$ is a $1$-form of ghost number $1$,
$a_0$ is a $0$-form of ghost number $2$. 
$a_0$ can be determined as
\begin{eqnarray}
a_0 = - \frac{1}{2} f^{ij}(\phi) c_i c_j,
\nonumber
\end{eqnarray}
up to $\delta_0$ exact terms.
Here, $f^{ij}(\phi)$ is an arbitrary function of $\phi$ such that
$f^{ij}(\phi)=-f^{ji}(\phi)$.
Note that terms including the metric on $\Sigma$
and terms including differentials $\partial_{\mu}$
can be dropped,
since those terms are $\delta_0$ exact up to total derivatives.
If we solve the descent equation, then
\begin{eqnarray}
a_1 = f^{ij} A_i c_j
- \frac{1}{2} \frac{\partial f^{ij}}{\partial \phi^k}
A^{+k} c_i c_j,
\nonumber
\end{eqnarray}
up to BRST exact terms, and finally
${\cal L}^{(1)}$ is uniquely determined as
\begin{eqnarray}
{\cal L}^{(1)} 
&=&  
\frac{1}{2} f^{ij}
(A_i A_j - 2 \phi^{+}_i c_j)
+ \frac{\partial f^{ij}}{\partial \phi^k}
\left(\frac{1}{2} c^{+k} c_i c_j + A^{+k} A_i c_j  \right)
\nonumber \\
&& - \frac{1}{4} \frac{\partial^2 f^{ij}}{\partial \phi^k \partial \phi^l}
A^{+k} A^{+l} c_i c_j
\label{BV1Poisson}
\end{eqnarray}
up to BRST exact terms \cite{Izawa:1999ib}.
Here,
$A_i \equiv d \sigma^{\mu} A_{\mu i}$, 
$A^{+i} \equiv d \sigma^{\mu} \epsilon_{\mu\nu} A^{*\nu i}$,
$\phi^{+}_i \equiv * \phi^{*}_{i} $, and
$c^{+i} \equiv * c^{*i}$, where
$*$ is the Hodge star on $\Sigma$.
From the definition of the BRST transformations, we have
\footnote{In the Lorentzian signature, 
the transformations of $A^{+}$ and $c^{+}$ 
have opposite sign.}
\begin{eqnarray}
&& \delta_0 A_{i} = d c_i,
\qquad 
\delta_0 \phi^{+}{}_i = d A_{i},
\nonumber \\
&& \delta_0 A^{+ i } = - d \phi^i,
\qquad \delta_0 c^{+ i} = d A^{+ i}.
\nonumber
\end{eqnarray}

At order $g^2$, the master equation is
$\sbv{S^{(1)}}{S^{(1)}} + 2 \sbv{S^{(0)}}{S^{(2)}}=0$.
From the assumption of locality, 
$S^{(2)}$ is an integral of a local 
function $\calL^{(2)}$ of fields and antifields.
Since $\delta_0 (\Psi) \propto \partial_{\mu}(*)$ 
for all the fields and antifields up to $\delta_0$ exact terms, 
$\sbv{S^{(0)}}{S^{(2)}}= \int d \calL^{(2)} =0$
if there is no boundary term.
The condition $\sbv{S^{(0)}}{S^{(2)}} =0$ for $S^{(2)}$ is the same
as the condition for $S^{(1)}$.
This means that if $S^{(1)}$ is redefined as 
$S^{(1)\prime} = S^{(1)}+ g S^{(2)}$,
$S^{(2)}$ can be absorbed into $S^{(1)}$.
\footnote{This is because $S_0$ is the action of 
the abelian BF theory.
This equation will not be satisfied for a different $S_0$.}

Continuing this procedure order by order, 
we obtain all the consistency conditions:
\begin{eqnarray}
&& \sbv{S^{(1)}}{S^{(1)}} =0,
\nonumber \\
&& 
S^{(n)}=0, \qquad  \ (n=2, 3, \cdots).
\end{eqnarray}
Substituting equation (\ref{BV1Poisson}) into
$\sbv{S^{(1)}}{S^{(1)}} =0$, 
we obtain the following condition on $f^{ij}(\phi)$:
\begin{eqnarray}
\frac{\partial f^{ij}}{\partial \phi^m}(\phi) f^{mk}(\phi)
+ \frac{\partial f^{jk}}{\partial \phi^m}(\phi) f^{mi}(\phi)
+ \frac{\partial f^{ki}}{\partial \phi^m}(\phi) f^{mj}(\phi) =0.
\label{Jacobi}
\end{eqnarray}

We have found the general solution for the deformation of 
the two-dimensional abelian BF theory \cite{Izawa:1999ib}.
The complete BV action is as follows:
\begin{eqnarray}
S &=& S^{(0)} + g S^{(1)}
\nonumber \\
\!\!\! &=& \!\!\! \int_{\Sigma}\Bigg( A_i d \phi^i
+ A^{+ i} d c_{i}
+ g \biggl(
\frac{1}{2} f^{ij}
(A_i A_j - 2 \phi^{+}_i c_j)
\nonumber \\
&& 
+ \frac{\partial f^{ij}}{\partial \phi^k}
\left(\frac{1}{2} c^{+k} c_i c_j + A^{+k} A_i c_j \right)
- \frac{1}{4} \frac{\partial^2 f^{ij}}{\partial \phi^k \partial \phi^l}
A^{+k} A^{+l} c_i c_j
\biggr)\Bigg).
\label{BVPSM}
\end{eqnarray}
Here, $f^{ij}(\phi)$ satisfies identity
(\ref{Jacobi}).
If we set $\Phi^*=0$, 
we have the following non-BV action:
\begin{eqnarray}
S 
&=& 
\int_{\Sigma} d^2 \sigma \left(
\epsilon^{\mu\nu} A_{\mu i} 
\partial_{\nu} \phi^i
+ 
\frac{1}{2} \epsilon^{\mu\nu} f^{ij}(\phi) A_{\mu i} A_{\nu j}
\right)
\nonumber \\
%
&=& \int_{\Sigma}\left( A_{i} d \phi^i
+  \frac{1}{2} f^{ij}(\phi) A_i A_j\right),
\label{PSM}
\end{eqnarray}
where $g$ is absorbed by redefinition of $f$.
This action is called the Poisson sigma model or nonlinear gauge theory in two dimensions.
\cite{Ikeda:1993aj, Ikeda:1993fh, Schaller:1994es, Schaller:1994uj}

\begin{theorem} 
The deformation of a two-dimensional abelian BF theory
is the Poisson sigma model. \cite{Izawa:1999ib}
\end{theorem} 
This model is considered to be the simplest nontrivial AKSZ sigma model.

\subsection{Poisson Sigma Model}
\noindent
In this subsection, we list the properties of
the Poisson sigma model (\ref{PSM}).

In  special cases, the theory reduces to well-known theories.
If $f^{ij}(\phi)=0$, then
the theory reduces to the abelian BF theory:
\begin{eqnarray}
S_A
= \int_{\Sigma} d^2 \sigma \epsilon^{\mu\nu} A_{\mu i} 
\partial_{\nu} \phi^i
= \frac{1}{2} \int_{\Sigma} d^2 \sigma 
\epsilon^{\mu\nu} \phi^i \ F_{0 \mu\nu i}.
\nonumber
\end{eqnarray}
%
If $f^{ij}(\phi)$ is a linear function,
$f^{ij}(\phi)= f^{ij}{}_k \phi^k$, 
equation (\ref{Jacobi}) is equivalent to the Jacobi identity 
of the structure constants $f^{ij}{}_k$
of a Lie algebra.
The resulting theory is a nonabelian BF theory:
\begin{eqnarray}
S_{NA} = \int_{\Sigma} d^2 \sigma 
\left( \epsilon^{\mu\nu} A_{\mu i} 
\partial_{\nu} \phi^i
+ \frac{1}{2} \epsilon^{\mu\nu} 
f^{ij}{}_k \phi^k A_{\mu i} A_{\nu j} \right)
= \int_{\Sigma} d^2 \sigma \epsilon^{\mu\nu} \phi^i
F_{\mu\nu i},
\nonumber
\end{eqnarray}
where 
$F_{\mu\nu i} = \partial_\mu A_{\nu i} - \partial_\nu A_{\mu i}
+ f^{jk}{}_i A_{\mu j} A_{\nu k}$, 
and this action has the following gauge symmetry:
\begin{eqnarray}
&& \delta \phi^i = - f^{ij}{}_k \phi^k \epsilon_j,
\qquad
\delta A_{\mu i} = \partial_{\mu} \epsilon_i
+ f^{jk}{}_i A_{\mu j} \epsilon_k.
\nonumber
\end{eqnarray}

Next, we analyze the symmetry of the Poisson sigma model.
The Poisson sigma model has the following gauge symmetry:
\begin{eqnarray}
&& \delta \phi^i = - f^{ij}(\phi) \epsilon_j,
\nonumber \\
&& \delta A_{\mu i} = \partial_{\mu} \epsilon_i
+ 
\frac{\partial f^{jk}(\phi)}{\partial \phi^i}
A_{\mu j} \epsilon_k,
\label{PSMgaugesym}
\end{eqnarray}
under the condition given by equation (\ref{Jacobi}). In fact, we can directly prove that
the requirement $\delta S=0$ under the gauge transformation (\ref{PSMgaugesym})
is equivalent to equation (\ref{Jacobi}).
In the Hamiltonian formalism, the constraints are
\begin{eqnarray}
&& G^i = \partial_1 \phi^i + f^{ij}(\phi) A_{1 j},
\nonumber
\end{eqnarray}
which satisfy the algebra
defined by the following Poisson bracket:
\begin{eqnarray}
&& \{G^i(\sigma), G^j(\sigma^{\prime}) \}_{PB}
= - \frac{\partial f^{ij}}{\partial \phi^k}
G^k(\sigma) \delta(\sigma - \sigma^{\prime}).
\nonumber
\end{eqnarray}
We can also derive the gauge transformation \eqref{PSMgaugesym}
generated by the charge constructed from the constraints $G^i(\sigma)$.
The gauge algebra has the following form:
\begin{eqnarray}
&& [\delta(\epsilon_1), \delta(\epsilon_2)]
\phi^i = \delta(\epsilon_3) \phi^i,
\nonumber \\
&& [\delta(\epsilon_1), \delta(\epsilon_2)]
A_{\mu i} = \delta(\epsilon_3) A_{\mu i} 
+ \epsilon_{1j} \epsilon_{2k} 
\frac{\partial f^{jk}}{\partial \phi^i \partial \phi^l}(\phi)
\epsilon_{\mu\nu} \frac{\delta S}{\delta A_{\nu l}},
\label{gaugealgebra}
\end{eqnarray}
where $\epsilon_1$ and $\epsilon_2$ are gauge parameters,
and $\epsilon_{3i} = 
\frac{\partial f^{jk}}{\partial \phi^i}(\phi)
\epsilon_{1j} \epsilon_{2k}$.
Equation (\ref{gaugealgebra})
for $A_{\mu i}$ shows that the gauge algebra is open.
Therefore, this theory cannot be quantized by 
the BRST formalism and it requires the BV formalism.

This model is a sigma model from a two-dimensional manifold 
$\Sigma$ to a target space $M$, based on a map
$\phi:\Sigma \longrightarrow M$.
If equation (\ref{Jacobi}) is satisfied on $f^{ij}(\phi)$, then
$\{ F(\phi), G(\phi) \}_{PB} \equiv 
f^{ij}(\phi) \frac{\partial F}{\partial \phi^i} 
\frac{\partial G}{\partial \phi^j}$ defines 
a Poisson bracket on a target space $M$, since
equation (\ref{Jacobi}) is 
the Jacobi identity of this Poisson bracket.
\footnote{In the notation used in this paper, $\sbv{-}{-}$ is the BV antibracket,
and $\sbv{-}{-}_{PB}$ is the usual Poisson bracket.}

Conversely, assume that the Poisson bracket on $M$ is given by
$\{ F(\phi), G(\phi) \}_{PB}
=
f^{ij}(\phi) \frac{\partial F}{\partial \phi^i} 
\frac{\partial G}{\partial \phi^j}$. Then, equation (\ref{Jacobi}) is derived from 
the Jacobi identity and
the action given in equation (\ref{PSM}), which is constructed by this $f^{ij}(\phi)$,
is gauge invariant.
From this property, the action $S$ is 
called the \textsl{Poisson} sigma model.

The algebraic structure of the gauge algebra is 
not a Lie algebra but a 
\textsl{Lie algebroid} over the cotangent bundle 
$T^*M$. \cite{Levin:2000fk}

\begin{definition}\label{LieAlgebroid}
A \textsl{Lie algebroid} over a manifold $M$ is a vector bundle 
$E \rightarrow M$ with a Lie algebra structure on the 
space of the sections $\Gamma(E)$ defined by the 
bracket $[e_1, e_2]$, for $e_1, e_2 \in \Gamma(E)$
and a bundle map (the anchor)
$\rho: E \rightarrow TM$ satisfying the following properties:
\begin{eqnarray}
&& 1, \ [\rho(e_1), \rho(e_2)] = \rho([e_1, e_2]),
\label{liealgdef1}
\\
&& 2, \
[e_1, F e_2] = F [e_1, e_2] + (\rho(e_1) F) e_2,
\label{liealgdef2}
\end{eqnarray}
where 
$e_1, e_2 \in \Gamma(E)$, 
$F \in C^{\infty}(M)$
and
the bracket $[-,-]$ on
the r.h.s. of equation (\ref{liealgdef1}) is the 
Lie bracket on the vector fields.
\end{definition}

Let us consider the expressions of a Lie algebroid in local coordinates. 
Let $x^i$ be a local coordinate 
on a base manifold $M$, and let $e_a$ be a local basis on the fiber of $E$.
The two operations of a Lie algebroid
are expressed as
\begin{eqnarray}
&& \rho(e_a) F(x) = \rho{}^{i}{}_a(x) \frac{\partial F(x)}{\partial x^i}, 
\qquad
[e_a, e_b] = f{}^c{}_{ab}(x) e_c,
\nonumber 
\end{eqnarray}
where 
$i, j, \cdots$ are indices on $M$, 
$a, b, \cdots$ are indices of the fiber of the vector bundle $E$, and
$\rho{}^{i}{}_a(x)$ and $f{}^c{}_{ab}(x)$ are local functions.
Then, equations (\ref{liealgdef1}) and (\ref{liealgdef2})
are written as
\begin{eqnarray}
&&
\rho{}^{m}{}_a \frac{\partial \rho{}^{i}{}_b}{\partial \phi^m} 
- \rho{}^{m}{}_b \frac{\partial \rho{}^{i}{}_a}{\partial \phi^m} 
+ \rho{}^{i}{}_c f{}^c{}_{ab} = 0,
\label{Liealgebroididentity1}
\\
&&
\rho{}^{m}{}_{[a} \frac{\partial f{}^d{}_{bc]}}{\partial \phi^m} 
+ f{}^d{}_{e[a} f{}^e{}_{bc]} = 0.
\label{Liealgebroididentity2}
\end{eqnarray}
Here, we use the notation
$f{}^d{}_{e[a} f{}^e{}_{bc]} 
= f{}^d{}_{ea} f{}^e{}_{bc} 
+ f{}^d{}_{eb} f{}^e{}_{ca} 
+ f{}^d{}_{ec} f{}^e{}_{ab} 
$.
For the cotangent bundle $E = T^*M$,
the indices on the fiber $a, b, \cdots$ 
run over the same range as the indices $i, j, \cdots$.
We can take $\rho{}^{ij}(\phi) = f^{ij}(\phi)$ and
$f{}_i{}^{jk}(\phi) = \frac{\partial f^{jk}}{\partial \phi^i}(\phi)$.
Substituting these equations into equation (\ref{Liealgebroididentity2}),
we obtain the Jacobi identity $(\ref{Jacobi})$.
This special Lie algebroid is called the Poisson Lie algebroid.

The action given by equation (\ref{PSM}) is unitary, and the fields have no physical degrees of freedom,
which can be shown by analyzing it using the constraints 
in the Hamiltonian analysis or by
counting the gauge symmetries in the Lagrangian analysis.
%
The partition function does not depend on the metrics on $\Sigma$ and $M$.
That is, the Poisson sigma model is a topological field theory.

In the remaining part of this subsection, we list 
known applications of the Poisson sigma model.

\paragraph{1.}We consider two-dimensional gravity theory as a nontrivial example of a Poisson sigma model \cite{Ikeda:1993aj, Ikeda:1993fh, Schaller:1994uj}.
Consider a target manifold $M$ in three dimensions.
Let the target space indices be $i=0,1,2$ and $\bar{i}=0,1$.
Let us denote $A_{\mu i} = (e_{\mu \bar{i}}, \omega_{\mu})$
and $\phi^i = (\phi^{\bar{i}}, \varphi)$.
We can take $f^{ij}(\phi)$ as
\begin{eqnarray}
f{}^{\bar{i}\bar{j}}(\phi^i) 
= - \epsilon^{\bar{i}\bar{j}} V(\varphi),  \quad
f{}^{2\bar{i}}(\phi^i) = - f{}^{\bar{i}2} = 
\epsilon^{\bar{i}\bar{j}} \phi_{\bar{j}}, \quad
f{}^{22}(\phi^i) = 0.
\label{2dgravity}
\end{eqnarray}
Equation (\ref{2dgravity}) satisfies equation (\ref{Jacobi}), and
the action given by equation $(\ref{PSM})$ reduces to 
\begin{eqnarray}
S &=& \int_{\Sigma} \sqrt{-g} \left(\frac{1}{2} \varphi R
- V(\varphi)\right) - \phi_{\bar{i}} T^{\bar{i}},
\nonumber
\end{eqnarray}
where $g$ is the determinant of the metric $g_{\mu\nu}
= \eta^{\bar{i}\bar{j}} e_{\mu \bar{i}}e_{\mu \bar{j}}$
on $\Sigma$, $R$ is the 
scalar curvature, and $T^{\bar{i}}$ is the torsion.
Here, $e_{\mu \bar{i}}$ is identified with the zweibein, and 
$\omega_{\mu}^{\bar{i}\bar{j}} 
= \omega_{\mu} \epsilon^{\bar{i}\bar{j}}$ is the spin connection.
This action is the 
gauge theoretic formalism of a gravitational theory
with a dilaton scalar field $\varphi$.

\paragraph{2.} 
Let $G$ be a Lie group.
The Poisson sigma model on the target space $T^*G$ 
reduces to the $G/G$ gauged Wess-Zumino-Witten (WZW) model,
when $A_{\mu i}$ is properly gauge fixed. \cite{Alekseev:1995py}

\paragraph{3.}If $f^{ij}$ is invertible as an antisymmetric matrix, then 
$f^{-1}_{ij}$ defines a symplectic form on $M$.
Then, $A_{\mu i}$ can be integrated out, and
the action \eqref{PSM} becomes the so-called A-model,
\begin{eqnarray}
S = \frac{1}{2} \int_{\Sigma} d^2 \sigma \epsilon^{\mu\nu} 
f^{-1}{}_{ij}(\phi) \partial_{\mu} \phi^i
\partial_{\nu} \phi^j,
\nonumber
\end{eqnarray}
in which the integrand is the pullback of the symplectic structure on 
$M$.
If $M$ is a complex manifold, the 
B-model can also be derived from the Poisson sigma model. 
\cite{Alexandrov:1995kv}

\paragraph{4.}A Poisson structure can be constructed from a classical r-matrix.
A sigma model in two dimensions with a classical r-matrix 
can be constructed as a special case of the Poisson sigma model
\cite{Falceto:2001eh, Calvo:2003kv}
which has a Poisson-Lie structure.

\paragraph{5.}The Poisson sigma model is generalized by introducing the 
Wess-Zumino term
$\int_{X_3} \frac{1}{3!} H_{ijk}(\phi) 
d \phi^i \wedge d \phi^j \wedge d \phi^k$:
\begin{eqnarray}
S 
&=& \int_{\Sigma} A_{i} d \phi^i
+  \frac{1}{2} f^{ij}(\phi) A_i A_j +
\int_{X_3} \frac{1}{3!} H_{ijk}(\phi) 
d \phi^i \wedge d \phi^j \wedge d \phi^k,
\label{WZPSM}
\end{eqnarray}
where $X_3$ is a manifold in three dimensions such that 
$\partial X_3 = \Sigma$, and $H(\phi)
= \frac{1}{3!} H_{ijk}(\phi) 
d \phi^i \wedge d \phi^j \wedge d \phi^k$ is the
pullback of a closed $3$-form on $M$.
This action is called the WZ-Poisson sigma model 
or the twisted Poisson sigma model. \cite{Klimcik:2001vg}


\paragraph{6.}Quantization of the Poisson sigma model derives
a deformation quantization on a target Poisson manifold.
The open string tree amplitudes 
of the boundary observables of the Poisson sigma model
on a disc 
coincide with the deformation quantization formulas
on the Poisson manifold $M$ obtained by Kontsevich.
\cite{Cattaneo:1999fm}
This corresponds to the large B-field limit 
in open string theory.
\cite{Seiberg:1999vs}

\subsection{Superfield Formalism}\label{superfieldformalismPSM}
\noindent
From this point onward, we set $g=1$ or equivalently, we
absorb $g$ into $f^{ij}(\phi)$.
The BV action of the Poisson sigma model (\ref{BVPSM})
is simplified by introducing supercoordinates.
\cite{Cattaneo:1999fm}
Let us introduce a Grassmann-odd 
supercoordinate $\theta^{\mu} \ (\mu =0, 1)$.
It is not a spinor but a vector
and carries a ghost number 
of $1$.

Superfields are introduced by combining
fields and antifields with $\theta^{\mu}$, as follows:
\begin{eqnarray}
&& \bphi^i (\sigma, \theta)
\equiv \phi^i + \theta^{\mu} A_{\mu}^{+i} 
+ \frac{1}{2} \theta^{\mu} \theta^{\nu} c_{\mu\nu}^{+i}
= \phi^i + A^{+i} + c^{+i},
\nonumber \\
&& \ba_i (\sigma, \theta) 
\equiv 
- c_i + \theta^{\mu} A_{\mu i} 
+ \frac{1}{2} \theta^{\mu} \theta^{\nu} \phi_{\mu\nu i}^{+}
= - c_i + A_{i} 
+ \phi_{i}^{+},
\label{superfield2d}
\end{eqnarray}
where each term in the superfield has the same ghost number
\footnote{Note that $d \sigma^{\mu}$ is commutative with
a Grassmann-odd component field in the nonsuperfield BV
formalism, whereas 
$\theta^{\mu}$ is anticommutative with
a Grassmann-odd component field in the superfield formalism.
}. 
Note that in this subsection, the component superfields are assigned the same notation as in the nonsuperfield formalism and
$d \sigma^{\mu}$ in the differential form expression of 
each field is replaced by $\theta^{\mu}$ in 
equation (\ref{superfield2d}).
The ghost number is called the \textsl{degree}, $|\bPhi|$, in the AKSZ formalism
\footnote{Precisely, the notation $|\Phi|$ represents the total degree, 
the sum of the ghost number plus the super form degree of $\Phi$, 
if it is a graded differential form on a graded manifold.
See Appendix.}.
The degree of $\bphi$ is zero, and that of $\ba$ is one.
The original fields $\phi^i$ and $A_{\mu i}$ appear
in $|\bphi|$-th order of $\theta$ and 
$|\ba|$-th order of $\theta$ components in the superfields, respectively.

With this notation, the BV action of equation $(\ref{BVPSM})$ is summarized 
as the superintegral 
of superfields as
\begin{eqnarray}
S &=&  \int_{T[1]\Sigma}
d^2 \sigma d^2 \theta \ \left(
\ba_i \bbd \bphi^i +  \frac{1}{2} f^{ij}(\bphi) \ba_i \ba_j
\right),
\label{superBVaction2} 
\end{eqnarray}
where $\bbd \equiv \theta^{\mu} \partial_{\mu}$ is 
the superderivative and 
$T[1]\Sigma$ is a supermanifold, which has
local coordinates $(\sigma^{\mu}, \theta^{\mu})$.
The degree of $S$ is zero, $|S|=0$.
If we integrate by $d^2 \theta$,
then equation $(\ref{superBVaction2})$ reduces to equation
$(\ref{BVPSM})$.

The antibrackets of component fields given in (\ref{2dantibracket})
are combined into a compact form by using the superantibracket as
\begin{eqnarray}
\sbv{F}{G} \equiv \int_{T[1]\Sigma} d^2 \sigma d^2 \theta
\left( F \frac{\rd}{\partial \bphi^i}
\frac{\ld }{\partial \ba_i}  G
- 
F \frac{\rd}{\partial \ba_i}
\frac{\ld}{\partial \bphi^i} G
\right)
\delta^2(\sigma-\sigma^{\prime})
\delta^2(\theta-\theta^{\prime}),
\nonumber
\end{eqnarray}
where $F$ and $G$ are functionals of superfields.
The classical master equations can be replaced by the super-classical 
master equation, $\sbv{S}{S}=0$, where the bracket is the super-antibracket.
The BRST transformation on a superfield
$\bPhi = 
\Phi^{(0)}
+ \theta^{\mu} \Phi_{\mu}^{(1)}
+ \frac{1}{2} \theta^{\mu} \theta^{\nu} \Phi_{\mu\nu}^{(2)}$
is 
\begin{eqnarray}
&& \delta \bPhi = \sbv{S}{\bPhi} = 
\delta \Phi^{(0)}
- \theta^{\mu} \delta \Phi_{\mu}^{(1)}
+ \frac{1}{2} \theta^{\mu} \theta^{\nu} \delta \Phi_{\mu\nu}^{(2)},
\nonumber
\end{eqnarray}
and the BRST transformation $\delta$ has degree 1.
The explicit form of the BRST transformation of each superfield is
\begin{eqnarray}
&& \delta \bphi^i = \sbv{S}{\bphi^i} = \bbd \bphi^i + f^{ij}(\bphi) \ba_j, 
\nonumber \\
&& \delta \ba_i = \sbv{S}{\ba_i} 
= \bbd \ba_i + \frac{1}{2} \frac{\partial f^{jk}}{\partial \bphi^i}
(\bphi) \ba_j \ba_k.
\nonumber
\end{eqnarray}
The (pullback on the) Poisson bracket on a target space is constructed by 
the double bracket of the super-antibracket:
$$
\{F(\phi), G(\phi)\}_{PB}
= f^{ij}(\bphi) \frac{\partial F(\bphi)}{\partial \bphi^i}
\frac{\partial G(\bphi)}{\partial \bphi^j}\Big|_{\bphi=\phi}
= - \sbv{\sbv{F(\bphi)}{S}}{G(\bphi)}\Big|_{\bphi=\phi}.
$$
This double bracket is called 
a \textsl{derived bracket} \cite{Kosmann-Schwarzbach:2003en}.

This superfield description
leads to the AKSZ construction
of a topological field theory.
In the AKSZ construction, objects in the BV formalism 
are interpreted 
as follows:
a superfield is a graded manifold;
a BV antibracket is a graded symplectic form; and
a BV action and the classical master equation 
are a coboundary operator (homological vector field) $Q$ 
with $Q^2=0$ and its realization by a Hamiltonian function, respectively.

\section{Abelian BF Theories for $i$-Form Gauge Fields in Higher Dimensions}
\label{ABFiform}
\subsection{Action}
\noindent
The superfield constructions discussed in the previous section 
can be applied to a wide class of TFTs.
An abelian BF theory in $n+1$ dimensions is considered
as a simple example
to show the formulation of the AKSZ construction.

Let us take an $n+1$-dimensional manifold $X_{n+1}$,
and let the local coordinates on $X_{n+1}$ be $\sigma^{\mu}$.
We consider $i$-form gauge fields with internal 
index $a_i$,
\begin{eqnarray}
e^{a(i)} \equiv e^{(i) a_i} = \frac{1}{i!}
d \sigma^{\mu_1} \wedge \cdots \wedge d \sigma^{\mu_i}
e^{a(i)}_{\mu_1\cdots \mu_i}(\sigma),
\end{eqnarray}
for $0 \leq i \leq n$,
where we choose the abbreviated notation $e^{a(i)}$.
$a(i)$ denotes
an internal index for an $i$-form gauge field.
For convenience, we divide the $e^{a(i)}$'s into two types:
$(q^{a(i)}, p_{a(n-i)})$, 
where $q^{a(i)} = e^{a(i)}$ if $0 \leq i \leq \floor{n/2}$;
and
$p_{a(n-i)} = e^{a(i)}$ if $\floor{(n+1)/2} \leq i \leq n$;
where $\floor{m}$ is the floor function, 
which takes the value of the largest integer less than 
or equal to $m$.
If $n$ is even, 
$q^{a(\floor{n/2})}$ and $p_{a(n-\floor{(n+1)/2})}= p_{a(n/2)}$
are both $n/2$-form gauge fields. 
Therefore, we introduce a metric $k_{a(n/2)b(n/2)}$ on 
the internal space of ${n/2}$-forms, and 
we can take $p_{a(n/2)} = k_{a(n/2)b(n/2)} q^{b({n/2})}$.
%
We denote a $0$-form by $x^{a(0)} (= q^{a(0)} = e^{a(0)})$ and 
an $n$-form by $\xi_{a(0)} (= p_{a(0)} = e^{a(n)})$.

The action $S_A$ of an abelian BF theory 
is the integral of a Lagrangian as $e \wedge d e^{\prime}$.
The integral is nonzero only for 
($n+1$)-form terms of $e \wedge d e^{\prime}$,
since $X_{n+1}$ is in $n+1$ dimensions.
Therefore, the action has the following form.
If $n = 2m +1$ 
is odd, 
\begin{align}
S_A
&=
\int_{X_{n+1}} 
\sum_{0 \leq i \leq (n-1)/2, a(i)}
(-1)^{n+1-i} p_{a(i)} d q^{a(i)}
\nonumber \\
&=
\int_{X_{n+1}} 
\Bigg( (-1)^{n+1} \xi_{a(0)} d x^{a(0)}
+ 
\sum_{1 \leq i \leq (n-1)/2, a(i)}
(-1)^{n+1-i}  p_{a(i)} d q^{a(i)}
\Bigg),
\end{align}
and if $n$ is even,
\begin{align}
S_A
&=
\int_{X_{n+1}} 
\Bigg( 
\sum_{0 \leq i \leq (n-2)/2, a(i)} 
(-1)^{n+1-i} p_{a(i)} d q^{a(i)}
+ (-1)^{\frac{n+1}{2}} k_{a(n/2)b(n/2)} q^{a(n/2)} d q^{b(n/2)}
\Bigg)
\nonumber \\
&=
\int_{X_{n+1}} 
\Bigg( 
(-1)^{n+1} \xi_{a(0)} d x^{a(0)}
+ 
\sum_{1 \leq i \leq (n-2)/2, a(i)} 
(-1)^{n+1-i} p_{a(i)} d q^{a(i)}
\nonumber \\ &
\qquad 
+ (-1)^{\frac{n+1}{2}}
k_{a(n/2)b(n/2)} q^{a(n/2)} d q^{b(n/2)}\Bigg).
\end{align}
The sign factors are introduced for later convenience.
If we define
$p_{a(n/2)} = k_{a(n/2)b(n/2)} q^{a(n/2)}$, 
then $S_A$ has the same expression for $n$ even or odd:
\begin{eqnarray*}
S_A
&=&
\sum_{0 \leq i \leq \floor{n/2}, a(i)} 
\int_{X_{n+1}} 
(-1)^{n+1-i} p_{a(i)} d q^{a(i)}.
\label{habelianBF}
\end{eqnarray*}
This action has the following abelian gauge symmetries:
\begin{eqnarray}
\delta q^{a(i)} 
= d q^{(i-1), a(i)},
\qquad
\delta p_{a(i)}
= d p^{(n-i-1)}{}_{a(i)},
\label{abrs}
\end{eqnarray}
where
$q^{(i-1), a(i)}$ is an ($i-1$)-form gauge parameter, and
$p^{(n-i-1)}{}_{a(i)}$ is an ($n-i-1$)-form gauge parameter.
These equations are summarized as
$\delta e^{a(i)} = d e^{(i-1), a(i)}$,
where $e^{(i-1), a(i)}
= (q^{(i-1), a(i)}, p^{(i-1)}{}_{a(n-i)})$ 
is an ($i-1$)-form gauge parameter.

If the $i$-forms are expanded by local fields as 
$e^{a(i)}(\sigma) 
= \sum_{k, \mu_k} \frac{1}{k!}
d \sigma^{\mu_1} \wedge \cdots \wedge d \sigma^{\mu_k}
e^{a(k)}{}_{\mu_1\cdots \mu_k}(\sigma)
$, 
the action becomes
\begin{eqnarray*}
S_A
&=&
\sum_{\substack{0 \leq i \leq \floor{n/2} \\ a(i), \mu_i}} 
\!\!\!\! \pm \frac{1}{i!(n-i)!}
\int_{\calX} \!\!
d^{n+1} \sigma 
\ 
(-1)^{n+1-i} 
\epsilon^{\mu_0\cdots\mu_n}
p_{a(i) \mu_{i+1}\cdots \mu_n}
\partial_{\mu_i}
q^{a(i)}{}_{\mu_0\cdots \mu_{i-1}}.
\label{abelianBF1}
\end{eqnarray*}

\subsection{BV Formalism }
\noindent
In the BV formalism,
the ghosts,
ghosts for ghosts, and antifields are
introduced
for each $i$-form gauge field $e^{a(i)}$.
First, the gauge parameter $e^{(i-1), a(i)}$ is regarded as 
the FP ghost of ghost number $1$.
Moreover, we need the following towers of ghosts for ghosts,
because the gauge symmetry is reducible:
\begin{eqnarray}
&& \delzero e^{a(i)} = d e^{(i-1), a(i)}, 
\nonumber \\ 
&& \delzero e^{(i-1), a(i)} = d e^{(i-2), a(i)},
\nonumber \\ 
&& \vdots  \nonumber \\
&& \delzero e^{(1), a(i)} = d e^{(0), a(i)}, 
\nonumber \\
&& \delzero e^{(0), a(i)} = 0, 
\label{agauge}
\end{eqnarray}
where $e^{(k), a(i)}$ is a $k$-form ghost for ghosts,
$(k = 0, \cdots, i-1)$, of
%
ghost number $i-k$.
As usual, 
these fields are Grassmann-odd (even)
if the ghost number is odd (even).
We denote the original field by $e^{(i), a(i)} = e^{a(i)}$.

Next, antifields $e^{*(k)}{}_{a(i)}$ are introduced
for all fields and ghosts $e^{(k), a(i)}$ above.
An antifield $e^{*(k)}{}_{a(i)}$ has the same $k$-form as that of
the corresponding field $e^{(k), a(i)}$.
Note that ${\rm gh}(\Phi) + {\rm gh}(\Phi^*) = -1$ requires
that the antifield has ghost number $k-i-1$.
%
It is convenient to introduce 
the Hodge dual of an antifield,
$e^{+(n+1-k)}{}_{a(i)} = * e^{(k)}{}_{a(i)}$,
which is an $(n+1-k)$-form of ghost number $k-i-1$.
The \textsl{antibracket} is defined as
\footnote{Here, we use simple notation for the functional superderivative, but it will be defined later with more mathematical rigor.}
\begin{eqnarray}
&& \sbv{F}{G} \equiv 
\sum_{i, k} 
\int_{X_{n+1}} \!\!\!\! 
d^{n+1} \sigma 
\left(F \frac{\rd}{\partial e^{(k), a(i)}(\sigma)}
\frac{\ld }{\partial e^{+(n+1-k)}{}_{a(i)}(\sigma^{\prime})}  G
\right. 
\nonumber \\
&& 
- (-1)^{i(n+1-i)}
\left. 
F \frac{\rd}{\partial e^{+(n+1-k)}{}_{a(i)}(\sigma)}
\frac{\ld}{\partial e^{(k), a(i)}(\sigma^{\prime})} G
\right)
\delta^{n+1}(\sigma-\sigma^{\prime}).
\label{nantibracket}
\end{eqnarray}


%
%
The BV action is as follows:
\begin{eqnarray}
S^{(0)} &=& S_A + 
\sum_{\Phi} (-1)^{\gh \Phi} \int_{X_{n+1}} 
d^{n+1}\sigma~ \Phi^* \delta_0 \Phi
\nonumber \\
&=&
\sum_{0 \leq i \leq \floor{n/2}, 1 \leq k \leq i} 
\int_{X_{n+1}} 
\left( 
(-1)^{n+1-i} p_{a(i)} d q^{a(i)}
+ (-1)^{i-k} e^{+(n+1-k)}{}_{a(i)} d e^{(k-1), a(i)}
\right)
\nonumber \\
&=& 
\sum_{\substack{0 \leq i \leq \floor{n/2} \\ 1 \leq k \leq i}}
\int_{X_{n+1}} 
\left( 
(-1)^{n+1-i} p_{a(i)} d q^{a(i)}
+ 
(-1)^{i-k} q^{+(n+1-k)}{}_{a(i)} d q^{(k-1), a(i)}
\right.
\nonumber \\
&&~~~~~~~~~~~~~~~~~~~~
\left.
+ 
(-1)^{i+k-n} p^{+(k+1), a(i)} d p^{(n-k-1)}{}_{a(i)}
\right).
\label{iformBV}
\end{eqnarray}

\subsection{Superfield Formalism}
\noindent
Let us introduce a supercoordinate $\theta^{\mu}$
of ghost number $1$, i.e. of degree $1$.
The base $d \sigma^{\mu}$ 
is replaced by the supercoordinates $\theta^{\mu}$, thus
$e^{(k), a(i)}$ and $e^{+(n+1-k)}{}_{a(i)}$ are replaced by the 
superfield monomials,
\begin{eqnarray}
\be^{(k), a(i)} &=& (\pm) \frac{1}{k!}
\theta^{\mu_1} \cdots \theta^{\mu_k}
e^{(k), a(i)}_{\mu_1\cdots \mu_k}(\sigma),
\\
\be^{+(n+1-k)}{}_{a(i)} &=& (\pm) \frac{1}{k!}
\theta^{\mu_1} \cdots \theta^{\mu_{n+1-k}}
e^{+(n+1-k)}_{a(i), \mu_1\cdots \mu_{n+1-k}}(\sigma)
\end{eqnarray}
of degree $i$
and 
of degree $n-i$, respectively.
Although 
sign factors appear 
in the equations relating the original ghosts and antifields 
with the superfield components,
we do not write them explicitly.
Since the relation is one-to-one, 
we can identify the original fields and ghosts
by ghost number and form degree.
We define a superfield of degree $i$,
$\be^{a(i)}$,
where 
fields and ghosts for an $i$-form gauge field
and the antifields for an ($n-i$)-form gauge field
are combined
\cite{Cattaneo:2000rt, Cattaneo:2000mc}.
By combining $\be^{(k), a(i)}$ and $\be^{+(n+1-k)}{}_{a(i)}$ 
of degree $i$, 
we obtain
\begin{eqnarray}
\be^{a(i)} &=& 
\be^{(0), a(i)}
+ \be^{(1), a(i)}
+ \cdots 
+ \be^{(i), a(i)}
+ \be^{+(i+1), a(n-i)}
+ \be^{+(i+2), a(n-i)}
+ \cdots
+ \be^{+(n), a(n-i)}
%
\nonumber \\ 
&=& 
\sum_{k=0}^i \be^{(k), a(i)}
+ \sum_{k=i+1}^n \be^{+(k), a(n-i)},
\label{component}
\end{eqnarray}
where $0 \leq i \leq n$.
Note that the
internal indices $a(i)$ and $a(n-i)$ are equivalent, 
since we are considering a BF theory.

Let us denote the super-antibracket conjugate pair by
$(\be^{a(i)}, \be^{a(n-i)}) = (\bbq^{a(i)}, \bbp_{a(i)})$.
Then, the superfields can be written as follows:
\begin{eqnarray}
\bbq^{a(i)} &=& 
\sum_{k=0}^i \bbq^{(k), a(i)}
+ \sum_{k=i+1}^n \bbp^{(k), a(n-i)},
\nonumber \\
\bbp_{a(i)} &=& 
\sum_{k=0}^i \bbp^{(k)}{}_{a(i)}
+ \sum_{k=i+1}^n \bbq^{(k)}{}_{a(n-i)}.
\end{eqnarray}
If $n$ is even, the $n/2$-form part has a special relation,
 $\bbp_{a(n/2)} = k_{a(n/2)b(n/2)} \bbq^{b({n/2})}$.
Therefore, $\bbq^{a({n/2})}$ contains 
both ghosts and antifields for 
an (${n}/{2}$)-form gauge field $\bbq^{(n/2), a(n/2)}$:
\begin{eqnarray}
\bbq^{a(n/2)} &=& 
\sum_{k=0}^{n/2} \bbq^{(k), a(n/2)}
+ \sum_{k=n/2+1}^n k^{a(n/2)b(n/2)} \bbq_{(k), b(n/2)}.
\end{eqnarray}

If we use superfields, the antibrackets 
and the BV action are simplified.
The antibracket (\ref{nantibracket})
can be rewritten 
using superfields (\ref{component}) as follows:
\begin{eqnarray}
&& \sbv{F}{G} \equiv 
\int_{X_{n+1}} \!\!\!\! 
d^{n+1} \sigma d^{n+1} \theta
\left(F \frac{\rd}{\partial \bbq^{a(i)}(\sigma, \theta)}
\frac{\ld }{\partial \bbp_{a(i)}(\sigma^{\prime}, 
\theta^{\prime})} G 
\right.
\nonumber \\
&&
\left.
\qquad\qquad - (-1)^{i(n-i)}
F \frac{\rd}{\partial \bbp_{a(i)}(\sigma, \theta)}
\frac{\ld }{\partial \bbq^{a(i)}(\sigma^{\prime}, 
\theta^{\prime})} G
\right)
\delta^{n+1}(\sigma-\sigma^{\prime})
\delta^{n+1}(\theta-\theta^{\prime})
\nonumber
\\
&&= 
\int_{X_{n+1}} \!\!\!\! 
d^{n+1} \sigma d^{n+1} \theta
\left(F \frac{\rd}{\partial \be^{a(i)}(\sigma)}
\bomega^{a(i)b(j)}
\frac{\ld }{\partial \be^{b(j)}(\sigma^{\prime})} G
\right)
\delta^{n+1}(\sigma-\sigma^{\prime})
\delta^{n+1}(\theta-\theta^{\prime}).
\label{abelianBFnDantibracket}
\end{eqnarray}
Note that $\bomega^{a(i)b(j)}$ is the inverse of the 
graded symplectic structure 
on superfields.
The complicated BV action (\ref{iformBV})
can be simplified as the BV superaction as follows:
\begin{eqnarray}
S^{(0)} 
&=&
\sum_{0 \leq i \leq \floor{n/2}} 
\int
d^{n+1} \sigma d^{n+1} \theta \
(-1)^{n+1-i} 
\bbp_{a(i)}
\bbd \bbq^{a(i)}
\nonumber \\
&=&
\sum_{0 \leq i \leq n} 
\int
\mu \
\frac{1}{2} 
\be^{a(i)}
\bomega_{a(i)b(j)}
d \be^{b(j)},
\nonumber
\end{eqnarray}
where $\mu$ is the Berezin measure on the supermanifold.

As in the previous section, we apply deformation theory to the BV action $S^{(0)}$ and obtain all 
possible consistent terms of the BV action 
$S_I$ in BF theory.
Deformation theory in the superfield formalism 
yields the same result as in the nonsuperfield BV formalism,
in the case of a topological field theory.
\cite{Ikeda:2000yq, Ikeda:2001fq}
Therefore, below we will compute only in the superfield formalism.

The topological field theories constructed in Sections 
2 and 3 have the same structures:
superfields, antibrackets and BV actions.  These are
formulated in a unified way by
\textsl{QP-manifolds} and the structure becomes more transparent.
\section{QP-manifolds}
\subsection{Definition}
\noindent
A QP-manifold, which is also called 
a differential graded symplectic manifold,
is a key structure
for the AKSZ construction of a topological field theory.
This section and the next are 
devoted to providing the fundamentals of the formulation.
For further reading, we refer to 
Refs.~\citenum{Cattaneo:2001ys, Roytenberg:2006qz, Qiu:2011qr, Cattaneo:2010re}.

A \textsl{graded manifold} is the
mathematical counterpart to a superfield,
which is defined as
a ringed space with a structure sheaf 
of a graded commutative 
algebra over an ordinary smooth manifold $M$. 
It is defined locally using even and odd coordinates.
This grading is compatible with supermanifold grading,
that is, a variable of even degree is commutative, and 
one of odd degree is anticommutative.
The grading is called the \textsl{degree}.
$\calM$ is locally isomorphic to 
$C^{\infty}(U) \otimes S^{\cdot}(V)$,
where 
$U$ is a local chart on $M$, 
$V$ is a graded vector space, and $S^{\cdot}(V)$ is a free 
graded commutative algebra on $V$.
We refer to Refs.~\citenum{Carmeli, Manin, Varadarajan}
for a rigorous definition and a discussion of the properties
of a supermanifold.
The formulas for the graded differential calculus
are summarized in Appendix A.

The grading is assumed to be nonnegative in this lecture\footnote{
Though we do not consider a grading with negative 
degree in this article,
there exist sigma models on  
target graded manifolds with negative degree.
\cite{Ikeda:2004gp, Zucchini:2007ie}}
and a graded manifold with a nonnegative grading is called 
an \textit{N-manifold}.

The mathematical structure corresponding to the antibracket is 
a \textsl{P-structure}.
Thus, an N-manifold equipped with 
a graded symplectic structure 
$\omega$ of degree $n$ is called a P-manifold of degree $n$,
 $({\cal M},\omega)$, and
$\omega$ is a P-structure.
The graded Poisson bracket on $C^\infty ({\cal M})$ is defined 
from the graded symplectic structure $\omega$ on ${\cal M}$ as 
$$    
\{f,g\} 
= (-1)^{|f|+n} \iota_{X_f} \delta g
= (-1)^{|f|+n+1} \iota_{X_f} \iota_{X_g} \omega,
$$
for $f, g \in C^\infty({\cal M})$,
where the Hamiltonian vector field $X_f$ is defined by the equation 
$\iota_{X_f} \omega = - \delta f$. 

Finally, a \textsl{Q-structure}
corresponding to a BV action is introduced.
Let $({\cal M}, \omega)$ be a $P$-manifold of degree $n$.
We require that there is a differential $Q$ of degree $+1$ with $Q^2=0$ on $\calM$.
This $Q$ is called a Q-structure.
\begin{definition}
The triple $(\calM, \omega, Q)$ is called
a \textsl{QP-manifold} of degree $n$, 
and its structure is called a \textsl{QP-structure},
if $\omega$ and $Q$ are compatible,
that is, $L_Q \omega =0$. {\rm \cite{Schwarz:1992nx, Schwarz:1992gs}}
\end{definition}
$Q$ is also called a homological vector field. 
In fact, $Q$ is a Grassmann-odd vector field on $\calM$.
We take a generator $\Theta\in C^{\infty}(\calM)$ 
of $Q$ with respect to the graded Poisson bracket, $\{-,-\}$,
satisfying
\beq
Q=\{\Theta,-\}.
\eeq
$\Theta$ has degree $n+1$ and is called homological function, or 
Q-structure function.
$\Theta$ is also called Hamiltonian.\footnote{In fact, if the degree of a QP-manifold is \textsl{positive},
there always exists a generator $\Theta$ for the Q-structure 
differential $Q$ \cite{Roytenberg:2006qz}.
}
The differential condition, $Q^2=0$, implies that
$\Theta$ is a solution of the \textsl{classical master equation},
\begin{equation}
\{\Theta,\Theta\}=0.
\label{cmaseq}
\end{equation}

\subsection{Notation}
\noindent 
We will now introduce the notation for graded manifolds.
Let $V$ be an ordinary vector space. Then $V[n]$ is
a vector space in which the degree is shifted by $n$.
More generally, 
if $V_m$ is a graded vector space of degree $m$,
the elements of $V_m[n]$ are of degree $m+n$
(this is also denoted by $V_{m+n} = V_{m}[n]$).
If $V$ has degree $n$,
the dual space $V^*$ has degree $-n$.
The product of $u \in V_m$ and $v \in V_n$ is graded commutative,
$uv = (-1)^{mn} vu$.

Let $M$ be an ordinary smooth manifold. 
Given a vector bundle $E \longrightarrow M$, 
$E[n]$ is a graded manifold
assigning 
degree $n$ to the fiber variables,
i.e., a base variable has degree $0$, and
a fiber variable has degree $n$.
%
%
If the degree of the fiber is shifted by $n$, graded tangent and cotangent bundles are denoted by $T[n]M$  and $T^*[n]M$, respectively.

This notation is generalized to the case that 
both a smooth manifold $M$ and its fiber are graded.
$E[n]$ means that the degree of the fiber is shifted by $n$.
Note that 
$TM[1]$ is a tangent bundle 
for which the base and fiber degrees are $1$ and $1$, which 
is denoted by $(1,1)$.
Considering the duality of $V$ and $V^*$, we then have that
$T^*M[1]$ is a cotangent bundle 
for which the base and fiber degrees are $(1,-1)$.
Therefore, $T^*[n]M[1]$ is a cotangent bundle of degrees $(1, n-1)$.

Let us consider a typical example: a double vector bundle $T^*E$, which
is the cotangent bundle of a vector bundle.
We take local coordinates on $E$, $(x^i, q^a)$,
where $x^i$ is a coordinate on $M$, and $q^a$ is a coordinate on the fiber.
We also take dual coordinates $(\xi_i, p_a)$ on the cotangent space.
If we consider the graded bundle $T^*[n]E[1]$,
the coordinates $(x^i, q^a)$ have degrees  $(0,1)$ and
$(\xi_i, p_a)$ have degrees $(0+n,-1+n) = (n, n-1)$.
\footnote{
For notation $[n]$, we consider degree by $\bZ$-grading.
On the other hand, we can regard a graded manifold as a supermanifold
by considering the degree modulo $2$.
In this case, the shifting of odd and even degrees is denoted by $\Pi$.
For example, $\Pi TM$ is a tangent bundle in which the 
degree of the fiber is odd. There is a natural isomorphism,
$\Pi TM \simeq T[1]M$.
}

We can see that $C^{\infty}(E[1])$, the space of functions on $E[1]$,
is equivalent to the space of sections of the exterior algebra, 
$\wedge^{\bullet} E$,
$C^{\infty}(E[1]) = \Gamma(\wedge^{\bullet} E)$,
if we identify the local coordinates of degree 1
with the basis of the exterior algebra.
Let $e^a$ be a local basis of the sections of $E$.
Then, a function 
\beqa
\frac{1}{s!} f_{a_1 \cdots a_s}(x) q^{a_1} \cdots q^{a_s} 
\in C^{\infty}(E[1])
\eeqa
can be identified with 
\beqa
\frac{1}{s!} f_{a_1 \cdots a_s}(x) e^{a_1} \wedge \cdots \wedge e^{a_s} 
\in \Gamma(\wedge^{\bullet} E).
\eeqa
\section{Examples of QP-Manifolds}\label{examplessection}
\noindent 
Typical examples of QP-manifolds are listed below.
\subsection{Lie Algebra and Lie Algebroid as QP-manifold of degree $n$}
\subsubsection{Lie Algebra}\label{LieAlgebra}
Let $n \geq 1$.
For an arbitrary $n$,
a Lie algebra becomes a QP-manifold of degree $n$
on a point $M=\{pt \}$.

Let $\mathfrak{g}$ be a Lie algebra with a Lie bracket $[-,-]$.
Then, $T^*[n]\mathfrak{g}[1]
\simeq \mathfrak{g}[1] \oplus \mathfrak{g^*}[n-1]$ 
is a P-manifold of degree $n$ 
with graded symplectic structure 
induced by a canonical symplectic structure on
$T^*\mathfrak{g}$.
We take local coordinates as follows: 
$q^a \in \mathfrak{g}[1]$ of degree $1$, 
and $p_a \in \mathfrak{g}^*[n-1]$ of degree $n-1$.
A P-structure $\omega = (-1)^{n|q|} \delta q^a \wedge \delta p_a$ 
is of degree $n$, and it is induced by 
the canonical symplectic structure 
on $T^*\mathfrak{g} \simeq \mathfrak{g} \oplus \mathfrak{g^*}$
by shifting the degree of the coordinates.
Taking a Cartan form $\Theta = \frac{1}{2} \langle p, [q, q] \rangle
= \frac{1}{2} f^a{}_{bc} p_a q^b q^c$,
where $\langle -, - \rangle$ is the canonical 
pairing of $\mathfrak{g}$ and $\mathfrak{g}^*$,
$f^a{}_{bc}$ is the structure constant, then,
$\Theta$ defines a Q-structure,
since it satisfies $\sbv{\Theta}{\Theta}=0$ 
due to the Lie algebra structure.

\subsubsection{Lie Algebroid}\label{LalgebroidQP}
A Lie algebroid has been defined in Definition \ref{LieAlgebroid}.
A Lie algebroid has a realization by a QP-manifold 
of degree $n$ for every $n$.

Let $n\ge 2$.
Let $E$ be a vector bundle over $M$, and let $\calM = T^{*}[n]E[1]$
be a graded manifold of degree $n$.
We take local coordinates $(x^i, q^a, p_a, \xi_i)$
of degrees $(0, 1, n-1, n)$.
The P-structure $\omega$ is a graded differential form of degree $n$
and is locally written as
\begin{eqnarray}
\omega = \delta x^i \wedge \delta \xi_i 
+ 
(-1)^{n|q|}
\delta q^a \wedge \delta p_a.
\end{eqnarray}
The Q-structure function is of degree $n+1$, and we have
\begin{eqnarray}\label{deftheta2}
\Theta = f{}_1{}^i{}_{a} (x) \xi_i q^a 
+\frac{1}{2} f_2{}^a{}_{bc}(x) p_a q^b q^c,
\end{eqnarray}
where the $f_i$'s are functions of $x$.
The Q-structure condition $\sbv{\Theta}{\Theta}=0$
imposes the following relations:
\begin{eqnarray}
&&
f{}_1{}^k{}_{b} \frac{\partial f{}_1{}^i{}_{a}}{\partial x^k} 
- f{}_1{}^k{}_{a} \frac{\partial f{}_1{}^i{}_{b}}{\partial x^k} 
+ f{}_1{}^i{}_{c} f_2{}^c{}_{ab} = 0,
\label{lafcn1} \\
&& 
f{}_1{}^k{}_{[d} \frac{\partial f_2{}^a{}_{bc]}}{\partial x^k} 
- f_2{}^a{}_{e[b} f_2{}^e{}_{cd]} = 0.
\label{lafcn2}
\end{eqnarray}
\eqref{lafcn1} and \eqref{lafcn2} are the same conditions as
for a Lie algebroid, 
\eqref{Liealgebroididentity1} and \eqref{Liealgebroididentity2},
where $f_1{}^i{}_a = \rho{}^i{}_a$ and $f_2{}^a{}_{bc} = - f^a{}_{bc}$.

For $n=1$, we need a slightly different realization,
which appeared in Ref.~\citenum{Bonechi:2005sj}.

\subsection{
$n=1$}
\noindent
In general, a QP-manifold of $n=1$ defines a Poisson structure.
We can also realize a complex structure using $n=1$.
Here, we give their constructions.
\subsubsection{Poisson Structure}\label{pmanifold}
A P-manifold $\calM$ of $n=1$ has the two degrees $(0, 1)$, and
it is canonically isomorphic to the cotangent bundle
$
\calM=T^{*}[1]M,
$
over the smooth manifold $M$.

On $T^*[1]M$, we take local coordinates $(x^i, \xi_i)$
of degrees $(0, 1)$; here, 
$x^i$ is a coordinate 
of the base manifold $M$, and $\xi_i$ is a coordinate of the fiber.
Note that $\xi_i$ is an odd element: $\xi_i \xi_j = - \xi_j \xi_i$.
The P-structure is $\omega = \delta x^i \wedge \delta \xi_i$. 
For $n=1$, the graded Poisson bracket $\sbv{-}{-}$ 
is isomorphic to the Schouten-Nijenhuis bracket.
Since the Q-structure function $\Theta$ has 
degree two, the general form is
$\Theta = \frac{1}{2} f^{ij}(x) \xi_i \xi_j$,
where $f^{ij}(x)$ is an arbitrary function of $x$.
The classical master equation, $\sbv{\Theta}{\Theta}=0$,
imposes the following condition on $f^{ij}(x)$:
\beqa
\frac{\partial f^{ij}(x)}{\partial x^{l}}
f^{lk}(x) + (ijk \ \mbox{cyclic})=0.
\label{PoissonJacobi}
\eeqa
The Q-structure $\Theta$ with Equation \eqref{PoissonJacobi}
is called a Poisson bivector field.

If $f^{ij}$ satisfies equation (\ref{PoissonJacobi}),
then the derived bracket defines a Poisson bracket on $M$:
\beqa
\sbv{F}{G}_{PB} = f^{ij}(x) 
\frac{\partial F}{\partial x^i} \frac{\partial G}{\partial x^j}
= - \sbv{\sbv{F}{\Theta}}{G}.
\label{Poissonb}
\eeqa
Equation (\ref{PoissonJacobi}) corresponds to 
the Jacobi identity of this Poisson bracket.

Conversely, assume 
a Poisson bracket $\sbv{F}{G}_{PB}$ on $M$.
The Poisson bracket can be
locally written as
$f^{ij}(x) 
\frac{\partial F}{\partial x^i} \frac{\partial G}{\partial x^j}$.
Then, 
$\Theta = \frac{1}{2} f^{ij}(x) \xi_i \xi_j$
satisfies the classical master equation and 
is a Q-structure.

Thus, a QP-manifold of degree 1, $T^*[1]M$, 
defines a Poisson structure on $M$.
This QP-manifold of degree 1 is also regarded as a Lie algebroid on $T^*M$,
according to Definition \ref{LieAlgebroid}.


\subsubsection{Complex Structure}\label{ComplexStructure}
Let $M$ be a complex manifold of real dimension $d$.
A linear transformation $J: TM \longrightarrow TM$ is called 
a \textsl{complex structure} if the following two conditions
are satisfied:
\\
1) $J^2 = -1$
\\
2) For $X, Y \in TM$, 
$\proj_{\mp} [\proj_{\pm}X,  \proj_{\pm}Y] = 0$,
\qquad (integrability condition)
\\
where $\proj_{\pm}$ is the projection onto the $\pm \sqrt{-1}$ eigenbundles
in $TM$, and $[-, -]$ is the Lie bracket of vector fields.
We take a local coordinate expression of $J$,
$J^i{}_j(x)$, which is a rank $(1, 1)$ tensor.

In order to formulate a complex structure as a QP-manifold
we take the graded manifold
$\calM = T^*[1] T[1]M$.
This double vector bundle is locally 
isomorphic to 
$U \times \bR^{d}[1] \times \bR^d[1] \times \bR^{d}[0]$,
where $U$ is a local chart on $M$.
Let us take local coordinates on the local chart as $(x^i, \xi_i, q^i, p_i)$
of degree $(0, 1, 1, 0)$.
%
%
The P-structure is defined as
\begin{eqnarray*}
\omega &=& 
\delta x^{i} \wedge \delta \xi_{i} 
+
\delta p_{i} \wedge \delta q^{i}.
\end{eqnarray*}
If we take the Q-structure as
\begin{eqnarray}
\Theta &=& 
J^i{}_j(x) \xi_{i} q^j
+ \frac{\partial J^i{}_k}{\partial x^j}(x) 
p_{i} q^j q^k
\nonumber \\
&=& 
\left(
\begin{matrix}
\xi_{i} \ q^i& \cr
\end{matrix}
\!\!\right)
\left(
\begin{matrix}
0 & \frac{1}{2}J^i{}_j(x)& \cr
                - \frac{1}{2}J^j{}_i(x) &
 \frac{\partial J^i{}_k}{\partial x^j}(x) 
p_{i}  & \cr
\end{matrix}
\!\!\right)
 \left(
\begin{matrix}
\xi_{j} & \cr
                q^j & \cr
\end{matrix}
\!\!\right),
\nonumber
\end{eqnarray}
then $\sbv{\Theta}{\Theta} =0$ is equivalent to condition 
$2)$ 
in the definition of the complex structure $J$.

\subsection{
$n=2$}
\noindent
The following theorem is well known.
\cite{Roy01, Roy02}
\begin{theorem}\label{CalgeQP22}
A QP-structure of degree $2$ 
is equivalent to the Courant algebroid on a vector bundle $E$ over a smooth 
manifold $M$.
\end{theorem}
We explain this in detail.
\subsubsection{Courant Algebroid}\label{Calgebroid}
For $n=2$, the P-structure 
$\omega$ is an even form of degree $2$. 
The Q-structure function $\Theta$ has degree $3$.
$Q^2=0$ defines a Courant algebroid \cite{Courant, LWX}
structure on a vector bundle $E$.

First, let us introduce the most general form of the QP-manifold of degree $2$,
$(\calM, \omega, \Theta)$.
We denote the local coordinates of $\calM$ as $(x^i, \qone^a, \xi_i)$
of degrees $(0, 1, 2)$.
%
The P-structure $\omega$ of degree $2$ 
can be locally written as
\begin{eqnarray}
\omega =  \delta x^i \wedge \delta \xi_i 
+ \frac{k_{ab}}{2} \delta \qone^a \wedge \delta \qone^b,
\end{eqnarray}
where we have introduced a metric $k_{ab}$
on the degree one subspace.
The general form of the Q-structure function of degree $3$ is
\begin{eqnarray}
\Theta = f_1{}^i{}_{a} (x) \xi_{i} \qone{}^a 
+ \frac{1}{3!} f_{2abc} (x) \qone{}^a \qone{}^b \qone{}^c,
\label{courantQ}
\end{eqnarray}
where $f_1{}^i{}_{a} (x)$ and $f_{2abc} (x)$ are
local functions of $x$.
The Q-structure condition $\sbv{\Theta}{\Theta}=0$
imposes the following relations on these functions:
\begin{eqnarray}
&& k^{ab} f_{1}{}^i{}_{a} f_{1}{}^j{}_{b} = 0, \nonumber \\ 
&& \frac{\partial f_{1}{}^i{}_{b}}{\partial x^j} f_{1}{}^j{}_{c}
- \frac{\partial f_{1}{}^i{}_{c}}{\partial x^j} f_{1}{}^j{}_{b}
+ k^{ef}  f_{1}{}^i{}_{e}  f_{2fbc} = 0, \nonumber \\
&& \left( f_{1}{}^i{}_{d} \frac{\partial f_{2abc}}{\partial x^i}
- f_{1}{}^i{}_{c} \frac{\partial f_{2dab}}{\partial x^i}
+ f_{1}{}^i{}_{b} \frac{\partial f_{2cda}}{\partial x^i}
- f_{1}{}^i{}_{a} \frac{\partial f_{2bcd}}{\partial x^i} 
\right) 
\nonumber \\
&& \qquad
+ k^{ef} (f_{2eab}  f_{2cdf} 
+ f_{2eac}  f_{2dbf} 
+ f_{2ead}  f_{2bcf})
= 0.
\label{courantrelation}
\end{eqnarray}
We can prove that 
these identities (\ref{courantrelation}) are the same as the local coordinate 
expressions of the Courant algebroid conditions
on a vector bundle $E$.
The Courant algebroid is defined as:
\begin{definition}\label{courantdefinition}
A Courant algebroid is a vector bundle $E \longrightarrow M$,
and it has a nondegenerate symmetric bilinear form
$\bracket{\cdot}{\cdot}$ 
on the bundle, a bilinear operation $\circ$ on $\Gamma (E)$,
and a bundle map called an anchor map,
$\rho: E \longrightarrow TM$, satisfying the following properties:
%
\begin{eqnarray}
&& 1, \quad e_1 \circ (e_2 \circ e_3) = (e_1 \circ e_2) \circ e_3 
+ e_2 \circ (e_1 \circ e_3), 
  \label{courantdef1}
\\
&& 2, \quad \rho(e_1 \circ e_2) = [\rho(e_1), \rho(e_2)], 
  \label{courantdef2}
\\
&& 3, \quad e_1 \circ F e_2 = F (e_1 \circ e_2)
+ (\rho(e_1)F)e_2, 
  \label{courantdef3}
 \\
&& 4, \quad e_1 \circ e_2 = \frac{1}{2} {\cal D} \bracket{e_1}{e_2},
  \label{courantdef4}
\\ 
&& 5, \quad \rho(e_1) \bracket{e_2}{e_3}
= \bracket{e_1 \circ e_2}{e_3} + \bracket{e_2}{e_1 \circ e_3},
  \label{courantdef5}
\end{eqnarray}
where 
$e_1, e_2$, and $e_3$ are sections of $E$, $F$ is a function on
$M$ and 
${\cal D}$ is a map from the space of functions on $M$ to $\Gamma (E)$, 
defined as 
$\bracket{{\cal D}F}{e} = \rho(e) F$.
\end{definition}
Let $x^i$ be a local coordinate on $M$, and let 
$e^a$ be a local coordinate on the fiber of $E$.
We can write each operation on the local basis $x^i, e^a$,
as follows:
\begin{eqnarray}
&& e^a \circ e^b = k^{ad} k^{be} f_{2dec}(x) e^c,
\nonumber \\
&& \bracket{e^a}{e^b} = k^{ab},
\nonumber \\
&& \rho(e^a) F(x) = - k^{ab} f_{1}{}^i{}_{b} (x)
\frac{\partial F}{\partial x^i}(x),
\nonumber
\end{eqnarray}
where  $f_1$ and $f_2$ are local functions of $x$.
Substituting these expressions into
the relations given in equations \eqref{courantdef1}--\eqref{courantdef5}, we obtain
the identities (\ref{courantrelation}).

These operations can be constructed directly from a QP-manifold $\calM$,
without introducing local coordinates.
For this, we identify a section $e$ with an odd element $\qone$
in supergeometry computations using the shift functor $[1]$.
Then, the operations of the Courant algebroid 
can be represented as
\begin{eqnarray}
&& e_1 \circ e_2 \equiv - \sbv{\sbv{e_1}{\Theta}}{e_2},
\nonumber \\
&& \bracket{e_1}{e_2} \equiv  \sbv{e_1}{e_2},
\nonumber \\
&& \rho(e) F \equiv  \sbv{e}{\sbv{\Theta}{F}},
\nonumber \\
&& {\cal D}(*) \equiv \sbv{\Theta}{*},
\label{QPCourant}
\end{eqnarray}
where $F(x)$ is a function of degree $0$ and $e = \qone$ is a function of degree $1$.
We can also prove that $\sbv{\Theta}{\Theta}=0$ gives
the Courant algebroid structure \eqref{courantdef1}--\eqref{courantdef5}
 without using local coordinates.
Finally, a vector bundle $E$ is constructed from a graded manifold $\calM$
by a natural filtration of degree 
$\calM \longrightarrow E[1] \longrightarrow M$.

An important example of a Courant algebroid
is the direct sum of the tangent and cotangent bundles,
$E= TM \oplus T^*M$.
The bilinear operation is defined as
\begin{eqnarray}
(X+\alpha) \circ (Y+\beta)=[X+\alpha, Y+\beta]_{D}
=[X, Y] + L_X \beta - \iota_Y d \alpha.
\label{DorfmanTMTM}
\end{eqnarray}
Here, $X, Y \in TM$ are vector fields, $\alpha, \beta \in T^*M$ are $1$-forms,
$[-,-]$ is the ordinary Lie bracket on a vector field,
$L_X$ is the Lie derivative,
and $\iota_X$ is the interior product, respectively.
The bracket \eqref{DorfmanTMTM} is called the \textit{Dorfman bracket}, 
and generally it is not antisymmetric.
The Dorfman bracket is the most general bilinear form on $TM \oplus T^*M$ without background flux,
which satisfies the Leibniz identity.\footnote{
Note that $\circ$ is not necessarily assumed to be antisymmetric.
For a nonantisymmetric bracket, 
equation   \eqref{courantdef1}
is called \textsl{Leibniz identity}
instead of Jacobi identity.}
The antisymmetrization of the Dorfman bracket is called 
the Courant bracket.
The Courant bracket is antisymmetric, but it does not satisfy 
the Jacobi identity.
The symmetric form is 
$\bracket{X+\alpha}{Y+\beta} = \iota_X \beta + \iota_Y \alpha$
and the anchor map $\rho$ is the natural projection
to $TM$:
\begin{eqnarray}
\rho(X+\alpha) = X.
\end{eqnarray}

The corresponding QP-manifold is $\calM = T^*[2]T^*[1]M$.
The local Darboux coordinates are $(x^i, q^i, p_i, \xi_i)$, which have degrees $(0, 1, 1, 2)$\footnote{We can compare this formulation 
with the most general form 
of a QP-manifold by taking $\qone^a = (q^i, p_i)$.}.
Here, $q^i$
is a fiber coordinate of $T[1]M$,
$p_i$ a fiber coordinate of $T^*[1]M$, and
$\xi_i$ a fiber coordinate of $T^*[2]M$, respectively.
With degree shifting, $TM \oplus T^*M$ is naturally embedded into $T^*[2]T^*[1]M$
as $(x^i, d x^i, \frac{\partial}{\partial x^i}, 0) 
\mapsto (x^i, q^i, p_i, \xi_i)$.
The Courant algebroid structure on $TM \oplus T^*M$ is 
constructed from equation (\ref{QPCourant}).
The Dorfman bracket 
can be found via a derived bracket as $[-,-]_D = \sbv{\sbv{-}{\Theta}}{-}$ 
with $\Theta = \xi_i q^i$.
It means that $f_{1}{}^i{}_{j} = \delta^i{}_j$ and $f_{2ijk} = 0$.
This Courant algebroid is also called the \textsl{standard Courant algebroid}.

There is a freedom to introduce a closed $3$-form $H(x)$ as an extra datum.
If the Dorfman bracket is modified 
by $H(x)$ as
$(X+\alpha) \circ (Y+\beta)=[X+\alpha, Y+\beta]_{D}
=[X, Y] + L_X \beta - i_Y d \alpha + i_X i_Y H$,
the Courant algebroid structure is preserved.
This is
called the Dorfman bracket with a $3$-form $H$.
The P-structure remains the same, but
$\Theta$ is modified as 
$\Theta = \xi_i q^i + \frac{1}{3!} H_{ijk}(x) q^i q^j q^k$,
where $H(x)= \frac{1}{3!} H_{ijk}(x) dx^i \wedge dx^j \wedge dx^k$.
$\sbv{\Theta}{\Theta}=0$ is equivalent to $dH=0$.
This is called the \textsl{standard Courant algebroid with H-flux}.

There is an equivalent definition of the Courant algebroid
\cite{Kos2}, and it is closer to the construction from a QP-manifold.
\begin{definition}\label{CourantKS}
Let $E$ be a vector bundle over $M$ that is
equipped with a pseudo-Euclidean metric $(-,-)$,
a bundle map $\rho:E \longrightarrow TM$,
and a binary bracket $[-,-]_{D}$ on $\Gamma (E)$.
The bundle is called the Courant algebroid if
the following three conditions are satisfied:
\begin{eqnarray}\label{defcou1}
\label{cou1}
[e_1 ,[e_2 , e_3]_{D}]_{D}
&=&[[e_1 ,e_2]_{D},  e_3]_{D}+[e_2 ,[e_1, e_3]_{D}]_{D},
\\
\label{cou2} \rho(e_1)(e_2,e_3)&=&([e_1,e_2]_{D},e_3)+(e_2,[e_1,e_3]_{D}),
\\
\label{cou3} \rho(e_1)(e_2,e_3)&=&(e_1,[e_2,e_3]_{D}+[e_3,e_2]_{D}),
\end{eqnarray}
where $e_1 ,e_2 , e_3 \in\Gamma (E)$.
\end{definition}
We can prove that Definitions \ref{courantdefinition} and
\ref{CourantKS} are equivalent if the operations are
identified as 
$e_1 \circ e_2 = [e_1,e_2]_{D}$, $\bracket{e_1}{e_2} = (e_1,e_2)$,
with the same bundle map $\rho$.

\paragraph{Dirac structure}
A Dirac structure can be formulated in QP-manifold
language.
A \textit{Dirac structure} is a Lie algebroid, which is a substructure of a Courant algebroid, defined by: 
\begin{definition}
A \textsl{Dirac structure} $\calL$ is a 
maximally isotropic subbundle of a Courant algebroid $E$, 
whose sections are closed under the Dorfman bracket. That is,
\begin{eqnarray}
&& \bracket{e_1}{e_2}=0 \ \ (\mbox{isotropic}),
\label{Diracstr1}
\\
&& [e_1, e_2]_{C} \in \Gamma(\calL) \ \ (\mbox{closed}),
\label{Diracstr2}
\end{eqnarray}
for $e_1, e_2 \in \Gamma(\calL)$,
where $[e_1, e_2]_C = [e_1, e_2]_{D} - [e_2, e_1]_{D}$
is the Courant bracket.
\end{definition}
In QP-manifold language, the sections $\Gamma(\wedge^{\bullet} E)$ are identified with functions on the QP-manifold $C^{\infty}(\calM)$.
Then, the sections of the Dirac structure $\Gamma (\calL)$ 
are the functions with the conditions corresponding to
\eqref{Diracstr1} and \eqref{Diracstr2},
which are commutativity under the P-structure $\sbv{-}{-}$, and closedness under the derived bracket $\sbv{\sbv{-}{\Theta}}{-}$, respectively.

The Dirac structure on
the complexification of the Courant algebroid, 
$(TM \oplus T^*M )\otimes \bC$,
defines a generalized complex structure.
\cite{Hitchin:2004ut, Gualtieri:2003dx}
\subsection{
$n \geq 3$}
\noindent
We now define the algebraic and geometric structures which appear for $n \geq 3$ and give some examples.
An earlier analysis of the unification of algebraic and geometric structures
induced by higher QP-structures has been found in Ref.~\citenum{Severa:2001}.
\begin{definition}
A vector bundle $(E, \rho, [-,-]_L)$ is called an \textsl{algebroid}
if there is a bilinear operation 
$[-,-]_L: \Gamma (E) \times \Gamma (E) \to \Gamma (E)$,
and a bundle map $\rho:E\to TM$
satisfying the following conditions:
\begin{eqnarray}
&& \rho[e_{1},e_{2}]_L=[\rho(e_{1}),\rho(e_{2})], 
\label{algebroid1} \\
&& [e_{1},Fe_{2}]_L=F[e_{1},e_{2}]_L+\rho(e_{1})(F)e_{2},
\label{algebroid2}
\end{eqnarray}
where $F \in C^{\infty}(M)$ and
$[\rho(e_{1}),\rho(e_{2})]$
is the usual Lie bracket on $\Gamma (TM)$.
Note that 
$[-,-]_L$ need not be antisymmetric, and it need not satisfy the Jacobi identity.
$\rho$ is called \textsl{anchor map}.
\end{definition}
\begin{definition}
An algebroid $(E, \rho, [-,-]_L)$ is called 
a \textsl{Leibniz algebroid} 
if there is a bracket product 
$[e_{1},e_{2}]_L$ satisfying the Leibniz identity:
\begin{eqnarray}\label{defcou2}
\label{loday1}
[e_1 ,[e_2 , e_3]_L]_L&=&[[e_1, e_2]_L, e_3]_L+[e_2 ,[e_1, e_3]_L]_L,
\label{leibnizidentity}
\end{eqnarray}
where $e_{1},e_{2}, e_3 \in\Gamma (E)$.
\end{definition}
%
%

If the base manifold is a point $M = \{ pt \}$ and $\rho=0$, then 
the Leibniz algebroid reduces to a linear algebra, 
which is called Leibniz algebra \cite{LP, Loday}
\footnote{It is also called a Loday algebra.}.
A Leibniz algebra is a Lie algebra if the Leibniz bracket 
$[-, -]_L$ is antisymmetric.
Lie algebroids and Courant algebroids are also Leibniz algebroids.
The Lie bracket $[-,-]$ of the Lie algebroid and 
the Dorfman bracket $[-,-]_{D}$ 
of the Courant algebroid are identified as
special cases of the Leibniz bracket $[-, -]_L$.
Equation (\ref{cou1}) of the Dorfman bracket is 
equivalent to equation (\ref{leibnizidentity}).

The correspondence of a Leibniz algebroid 
to a homological vector field on a graded manifold is 
discussed in Ref.~\citenum{GKP}.
The following theorem has been
 presented in Ref.~\citenum{Kotov:2010wr}.
\begin{theorem}\label{QPloday}
Let $n > 1$. 
Functions of degree $n-1$ on a QP-manifold can be identified
with sections of a vector bundle.
The QP-structure induces a Leibniz algebroid structure on 
a vector bundle $E$.
\end{theorem}
%
Let $x$ be an element of degree $0$, and let 
$e^{(n-1)}$ be an element of degree $n-1$.
If we define
\begin{eqnarray}
[e_1, e_2]_L
&=& - \sbv{\sbv{e_1^{(n-1)}}{\Theta}}{e_2^{(n-1)}}, \\
\rho(e) F(x)
&=& (-1)^{n} \sbv{\sbv{e^{(n-1)}}{\Theta}}{F(x)},
\end{eqnarray}
then $e^{(n-1)}$ is identified with a section of a vector bundle, and
$[-,-]_L$ and $\rho$ satisfy the conditions in the definition of a Leibniz algebroid given by equations (\ref{algebroid1}), (\ref{algebroid2})
and (\ref{loday1}).


\subsubsection{$n=3$}\label{Talgebroid}
Let $n=3$. 
Let $(\calM, \omega, \Theta)$ be a QP-manifold of degree $3$.
$\calM$ has a natural filtration of degree 
$\calM \longrightarrow \calM_2 \longrightarrow \calM_1 \longrightarrow M$,
where $\calM_i ~(i=1,2)$ is a graded subspace of degree less than 
or equal to $i$.
The local coordinates are $(x^i, q^a, p_a, \xi_i)$
of degrees $(0, 1, 2, 3)$.
The P-structure $\omega$ is an odd symplectic form of degree $3$,
and it can be locally written as
\begin{eqnarray}
\omega = \delta x^i \wedge \delta \xi_i 
- 
\delta q^a \wedge \delta p_a.
\end{eqnarray}
Since the Q-structure function is of degree $4$, 
its general form is
\begin{eqnarray}\label{deftheta1}
\Theta &=& f{}_1{}^i{}_{a} (x) \xi_i q^a 
+\frac{1}{2} f_2{}^{ab}(x) p_a p_b
+\frac{1}{2} f_3{}^a{}_{bc}(x) p_a q^b q^c
+\frac{1}{4!} f_4{}_{abcd}(x) q^a q^b q^c q^d,
\end{eqnarray}
where the $f_i$'s are local functions of $x$.
The Q-structure condition $\sbv{\Theta}{\Theta}=0$
imposes the following relations on these functions:
\begin{eqnarray}
&&f{}_1{}^i{}_{b} f_2{}^{ba} = 0,\\
\label{fc1}
&&
f{}_1{}^k{}_{c} \frac{\partial f_2{}^{ab}}{\partial x^k} 
+ f_2{}^{da} f_3{}^b{}_{cd} + f_2{}^{db} f_3{}^a{}_{cd} = 0,\\
\label{fc2}
&&
f{}_1{}^k{}_{b} \frac{\partial f{}_1{}^i{}_{a}}{\partial x^k} 
- f{}_1{}^k{}_{a} \frac{\partial f{}_1{}^i{}_{b}}{\partial x^k} 
+ f{}_1{}^i{}_{c} f_3{}^c{}_{ab} = 0,\\
\label{fc3}
&& 
f{}_1{}^k{}_{[d} \frac{\partial f_3{}^a{}_{bc]}}{\partial x^k} 
+ f_2{}^{ae} f_4{}_{bcde}
- f_3{}^a{}_{e[b} f_3{}^e{}_{cd]} = 0,\\
\label{fc4}
&& 
f{}_1{}^k{}_{[a} \frac{\partial f_4{}_{bcde]}}{\partial x^k} 
+ f_3{}^f{}_{[ab} f_4{}_{cde]f} =0.
\label{fc5}
\end{eqnarray}
Here, $[abc \cdots]$ is the 'intermolecular antisymmetrization'
,
i.e., 
for two completely antisymmetric tensors 
$f_{a_1\cdots a_r}$ and $g_{b_1\cdots b_s}$,
this is an antisymmetric sum of only 
nonantisymmetric indices of $f$ and $g$ with unit weight,
\begin{eqnarray}
f_{[a_1\cdots a_r}g_{b_1\cdots b_s]}
= \frac{1}{r!s!} \sum_{\sigma \in \mathfrak{S}_{r+s}} 
{\rm sgn} (\sigma)
f_{a_{\sigma(1)}\cdots a_{\sigma(r)}}
g_{a_{\sigma(r+1)}\cdots a_{\sigma(r+s)}}.
\end{eqnarray}
For example, 
$f_3{}^a{}_{e[b} f_3{}^e{}_{cd]}
= f_3{}^a{}_{eb} f_3{}^e{}_{cd}
+ f_3{}^a{}_{ec} f_3{}^e{}_{db}
+ f_3{}^a{}_{ed} f_3{}^e{}_{bc}$
and 
$f_3{}^f{}_{[ab} f_4{}_{cde]f}$ has $\frac{5!}{2!3!} = 10$ terms.
\footnote{
If we take the notation that $[--]$ denotes complete 
antisymmetrization,
equation (\ref{fc4}) is
$f{}_1{}^k{}_{[d} \frac{\partial f_3{}^a{}_{bc]}}{\partial x^k} 
+ 2 f_2{}^{ae} f_4{}_{bcde}
- f_3{}^a{}_{e[b} f_3{}^e{}_{cd]} = 0$,
and
equation (\ref{fc5}) is
$f{}_1{}^k{}_{[a} \frac{\partial f_4{}_{bcde]}}{\partial x^k} 
+ \frac{1}{2} f_3{}^f{}_{[ab} f_4{}_{cde]f} =0
$.
}

These identities, equations (\ref{fc1})--(\ref{fc5}), define 
the Lie 3-algebroid on the vector bundle $E$, also called the Lie algebroid
up to homotopy, or the splittable H-twisted Lie algebroid
\cite{Ikeda:2010vz}.
It is a special case of 
the H-twisted Lie algebroid
\cite{Grutzmann}.

\subsubsection{Higher Dorfman Bracket}\label{liealoid}
Let $E$ be a vector bundle on $M$, and let $\calM = T^{*}[n]E[1]$
be a graded manifold of degree $n$, where $n\ge 4$.
We take local coordinates $(x^i, q^a, p_a, \xi_i)$
of degrees $(0, 1, n-1, n)$.
The QP-structure is naturally defined on $\calM = T^{*}[n]E[1]$, and
the P-structure $\omega$ is of degree $n$
and can be locally written as
\begin{eqnarray}
\omega = \delta x^i \wedge \delta \xi_i 
+ 
(-1)^{n|q|}
\delta q^a \wedge \delta p_a.
\end{eqnarray}
The general form of the Q-structure function is of degree $n+1$, and we have
\begin{eqnarray}\label{deftheta2}
\Theta &=& f{}_1{}^i{}_{a} (x) \xi_i q^a 
+\frac{1}{2} f_2{}^a{}_{bc}(x) p_a q^b q^c
+\frac{1}{(n+1)!} f_3{}_{a_1 \cdots a_{n+1}}(x) 
q^{a_1} q^{a_2} \cdots q^{a_{n+1}},
\end{eqnarray}
where the $f_i$'s are functions.
The Q-structure condition $\sbv{\Theta}{\Theta}=0$
imposes the following relations:
\footnote{In a complete-antisymmetrization notation, 
equation (\ref{fcn3}) is
$f{}_1{}^k{}_{[a_1} 
\frac{\partial f_3{}_{a_2\cdots a_{n+2}]}}{\partial x^k} 
+ \frac{2}{n+1} f_2{}^f{}_{[a_1 a_2} f_3{}_{a_3 \cdots a_{n+2}]f} =0$.
}
\begin{eqnarray}
&&
\label{fcn1}
f{}_1{}^k{}_{b} \frac{\partial f{}_1{}^i{}_{a}}{\partial x^k} 
- f{}_1{}^k{}_{a} \frac{\partial f{}_1{}^i{}_{b}}{\partial x^k} 
+ f{}_1{}^i{}_{c} f_2{}^c{}_{ab} = 0,\\
&& 
f{}_1{}^k{}_{[d} \frac{\partial f_2{}^a{}_{bc]}}{\partial x^k} 
- f_2{}^a{}_{e[b} f_2{}^e{}_{cd]} = 0,\\
\label{fcn2}
&& 
f{}_1{}^k{}_{[a_1} \frac{\partial f_3{}_{a_2\cdots a_{n+2}]}}{\partial x^k} 
+ f_2{}^f{}_{[a_1 a_2} f_3{}_{a_3 \cdots a_{n+2}]f} =0.
\label{fcn3}
\end{eqnarray}
A vector bundle $E\oplus\wedge^{n-1}E^{*}$ is naturally embedded into 
$T^{*}[n]E[1]$ by degree shifting.
The QP-structure induces an algebroid structure
on $E\oplus\wedge^{n-1}E^{*}$
by the derived bracket $[-,-]_{CD} = \sbv{\sbv{-}{\Theta}}{-}$, 
which is called the higher Dorfman bracket.
It has the following form:
\beq
[u+\alpha,v+\beta]_{CD}=[u,v]+L_{u}\beta-\iota_{v}d\alpha+H(u,v),
\eeq
where
$u,v\in\Gamma (E)$; $\alpha,\beta\in\Gamma(\wedge^{n-1}E^{*})$;
and $H$ is a closed $(n+1)$-form on $E$.
We refer to Refs.\citenum{Hagiwara, Wade, BS, Zambon:2010ka} 
for detailed studies on brackets of this type.
The graded manifold was analyzed in Ref.~\citenum{Zambon:2010ka}.

\subsubsection{Nonassociative Example}\label{nplusoneform}
A large class of nontrivial nonassociative algebras (algebroids) 
are included
in a QP-manifold of degree $n$, and we show one such example.
We define $\Theta$ as 
\begin{eqnarray}
\Theta=\Theta_0+\Theta_{2}+\Theta_{3}+\cdot\cdot\cdot+\Theta_{n},
\end{eqnarray}
where 
\begin{eqnarray}
\Theta_0 = f_{0}{}^{a(0)}{}_{b(1)}(x)
\xi_{a(0)} q^{b(1)},
\label{nonassociativeTheta0}
\end{eqnarray}
and 
\begin{eqnarray}
\Theta_i = \frac{1}{i!}
f_{i}{}_{a(n-i+1)b_1(1) \cdots b_i(1)}(x)
q^{a(n-i+1)} q^{b_1(1)} \cdots q^{b_i(1)},
\label{nonassociativeThetai}
\end{eqnarray}
where $i = 2, 3, \cdots n$, and $(x^{a(0)}, q^{a_1(1)}, \cdots, 
q^{a(n-1)}, \xi_{a(0)})$ have degrees $(0,1,\cdots,n-1,n)$. 
In particular, $\Theta_{n}$ is an $(n+1)$-form
on $\Gamma(\bigwedge^{n+1}E_{1})$.
Then, the master equation $\sbv{\Theta}{\Theta}=0$
is equivalent to
\begin{eqnarray}
\{\Theta_0,\Theta_0\}&=&0,
\label{nonassociative01}\\
\{\Theta_0,\Theta_{i}\}&=&0, \ \ i<n,
\label{nonassociative02}
\end{eqnarray}
and
\begin{eqnarray}
\{\Theta_0,\Theta_{n}\}+\sum\{\Theta_{i},\Theta_{n-i}\}&=&0,\ \ 
(n~\mbox{odd}),
\nonumber \\
\{\Theta_0,\Theta_{n}\}+\frac{1}{2}\{\Theta_{n/2},\Theta_{n/2}\}
+\sum\{\Theta_{i},\Theta_{n-i}\}&=&0,\ \ (n~\mbox{even}).
\label{nonassociative03}
\end{eqnarray}
The first condition \eqref{nonassociative01}
implies that $d:=\{\Theta_0,-\}$
is a differential,
and
the second one \eqref{nonassociative02} implies that
$\Theta_{i}$ is a closed $i$-form for each $i<n$.
The third condition \eqref{nonassociative03} says that $\Theta_{n}$
is a closed $(n+1)$-form up to homotopy \cite{Uchino}.
This structure is regarded as an $n$-term $L_{\infty}$-algebra.

\section{AKSZ Construction of Topological Field Theories
}\label{AKSZTFTgeneral}
\noindent
In this section, the superfield formalism of topological field theories presented in Sections 2 and 3 is reformulated by the AKSZ construction.
If a QP-structure on the target graded manifold $\calM$ is given, 
a QP-structure is induced 
on the mapping space (i.e., a space of fields)
from the world-volume graded manifold $\calX$ to 
the target graded manifold $\calM$.
\cite{Alexandrov:1995kv, Cattaneo:2001ys, Roytenberg:2006qz}

%
Let $(\calX, D)$ be a differential graded manifold 
(a dg manifold)
$\calX$
with a $D$-invariant nondegenerate measure $\mu$,
where
$D$ is a differential on $\calX$.
Let ($\calM, \omega, Q$) be a QP-manifold of degree $n$, where  
$\omega$ is a graded symplectic form of degree $n$ and
$Q= \{\Theta, -\}$ is a differential on $\calM$.
%
%
$\Map(\calX, \calM)$ is
a space of smooth maps from $\calX$ to $\calM$.
%
The QP-structure on $\Map(\calX, \calM)$
is constructed from the above data.

Since 
${\rm Diff}(\calX)\times {\rm Diff}(\calM)$ 
naturally acts on $\Map(\calX, \calM)$,
$D$ and $Q$ induce differentials 
on $\Map(\calX, \calM)$, 
$\hat{D}$ and $\check{Q}$. 
Explicitly, 
$\hat{D}(z, f) = D(z) d f(z)$ 
and $\check{Q}(z, f) = Q f(z)$,
for $z \in \calX$ and $f \in \calM^{\calX} = \Map(\calX, \calM)$.

%

Now, we introduce the following two maps.
The {\it evaluation map} 
${\rm ev}: \calX \times \calM^{\calX} \longrightarrow \calM$ 
is defined as
\begin{eqnarray}
{\rm ev}:(z, f) \longmapsto f(z),
\nonumber
\end{eqnarray}
where 
$z \in \calX$ and $f \in \calM^{\calX}$.

The {\it chain map} 
on the space of graded differential forms,
$\mu_*: \Omega^{\bullet}(\calX \times \calM^{\calX}) 
\longrightarrow \Omega^{\bullet}(\calM^{\calX})$, is defined as 
$$\mu_* \omega(f)(v_1, \ldots, v_k)
 = \int_{\calX} \mu(z)
 \omega(z, f) (v_1, \ldots, v_k),
$$
for a graded differential form $\omega$,
where $v$ is a vector field on $\calX$,
and 
$\int_{\calX} \mu$ is the integration over $\calX$.
When the degree is even, the integral is the standard one,  but when the degree is odd, it is the Berezin integral.
The map $\mu_* \ev^*: \Omega^{\bullet}(\calM) \longrightarrow
\Omega^{\bullet}(\calM^{\calX})$, which is called the 
\textsl{transgression map}, 
maps a graded differential form on the target space 
to a graded differential form on the mapping space.

The \textsl{P-structure} on $\Map(\calX, \calM)$ is defined as follows:
\begin{definition} 
For a graded symplectic form $\omega$ on 
$\calM$, 
a graded symplectic form $\bomega$ on $\Map(\calX, \calM)$ is
defined as
$\bomega := \mu_* \ev^* \omega$.
\end{definition}
%
Here, $\bomega$ is nondegenerate and closed,
because $\mu_* \ev^*$ preserves nondegeneracy and closedness.
Also, $\bomega$ is a graded symplectic form 
on $\Map(\calX, \calM)$ and induces a graded Poisson bracket 
$\sbv{-}{-}$, which is a BV antibracket
on $\Map(\calX, \calM)$.

Next, the \textsl{Q-structure} $S$ on $\Map(\calX,  \calM)$
is constructed.
$S$ corresponds to a {\it BV action} and
consists of two parts: $S = S_0 + S_1$.
%
%
We take a canonical $1$-form (the Liouville $1$-form) $\vartheta$ 
for the P-structure on $\calM$ such that 
$\omega= - \delta \vartheta$, and
we define $S_0 := \iota_{\hat{D}} \mu_* {\rm ev}^* \vartheta$,
which is equal to the kinetic term of the BF theory
$S^{(0)}$ presented in Section 2.\footnote{In the remainder of this paper, 
$S^{(0)}$ will be denoted as $S_0$.}
$S_1$ is constructed as follows: 
We take the Q-structure $\Theta$ on $\calM$ and 
map it by the transgression map, $S_1 := \mu_* \ev^* \Theta$.

From the definitions of $S_0$ and $S_1$, we can prove that
$S$ is a Q-structure on $\Map(\calX,  \calM)$ \cite{Cattaneo:2001ys}:
\begin{eqnarray}
\sbv{\Theta}{\Theta} =0
\Longleftrightarrow \ssbv{S}{S} =0.
\label{classicalmaster}
\end{eqnarray}
The right-hand side of this equation is 
the classical master equation in the BV formalism.
The homological vector field $\hQ$ on the mapping space 
is defined as $\hQ = \sbv{S}{-}$. 
By counting the degrees of $\sbv{-}{-}$ and $S$, it can be seen that the degree of $\hQ$ is 1.
$\hQ$ is a coboundary operator,
$\hQ^2=0$, by the classical master equation.
The cohomology defined by $\hQ$ is called BRST cohomology.
Since $\sbv{S_0}{S_0} = 0$, $S_0$ is considered to be a differential, and $S_1$ is considered to be a connection.
The classical master equation $\sbv{S}{S} =
2 \delta_0 S_1 + \sbv{S_1}{S_1} = 0$ is 
regarded as flatness condition, i.e., Maurer-Cartan equation.


The following theorem has been proved.
\cite{Alexandrov:1995kv}
\begin{theorem}\label{QPonmappingspace}
If $\calX$ is a differential graded manifold with a compatible measure
and 
$\calM$ is a QP-manifold,
then the graded manifold $\Map(\calX, \calM)$
inherits a QP-structure.
\end{theorem}
In fact, the QP-structure on $\Map(\calX,  \calM)$ yields
a topological field theory.

A topological field theory constructed from the BV formalism
is derived by considering a special super-world-volume $\calX$.
Let $X$ be an ($n+1$)-dimensional smooth manifold.
The supermanifold $\calX = T[1]X$ has 
a Berezin measure $\mu$ of degree $-n-1$, which is induced by the measure on $X$.
We can prove that the topological field theories in the previous sections 
can be constructed by the AKSZ construction on $T[1]X$.
Conversely, if $\calX = T[1]X$, 
a QP-structure on $\Map(\calX, \calM)$ is 
equivalent to the BV formalism of a topological field theory
\cite{Cattaneo:2001ys, Ikeda:2001fq}.
We can prove that this theory is gauge invariant and unitary
by physical arguments,
thus it defines a consistent quantum field theory.
\begin{definition}
An \textsl{AKSZ sigma model
(AKSZ topological field theory)
} in $n+1$ dimensions
is a QP-structure constructed in Theorem \ref{QPonmappingspace},
where $X$ in $\calX = T[1]X$ is an $n+1$ dimensional manifold 
and $\calM$ is a QP-manifold of degree $n$.
\end{definition}
In an AKSZ sigma model, $\Map(\calX, \calM)$ is a QP-manifold of degree $-1$,
since there is a measure of degree $-n-1$ on $\calX$
and a QP-structure on $\calM$ of degree $n$.
Therefore, it is an odd symplectic manifold.
The graded Poisson bracket $\sbv{-}{-}$ is of degree 1
and $S$ is of degree 0.

The AKSZ formalism can be applied to realize
the Batalin-Fradkin-Vilkovisky (BFV) formalism 
corresponding to the Hamiltonian formalism, 
if we choose an $n$-dimensional manifold $X$ and
$\calX$ has a measure of degree $-n$.
\cite{Cattaneo:2010re}
Then, the AKSZ construction defines a QP-structure of degree $0$
on $\Map(\calX, \calM)$.
Its P-structure is the usual Poisson bracket and
$\Theta$ is the BRST charge of the BFV formalism.

In order to quantize the theory by the BV formalism,
the classical master equation 
(\ref{classicalmaster}) must be modified to 
the \textsl{quantum master equation}.
An odd Laplace operator ${\Delta}$ 
on $\Map(\calX, \calM)$ can be constructed 
if $\Map(\calX, \calM)$ has
a measure $\brho$. 
\cite{Khudaverdian:2000zt, Khudaverdian:2002wz, Khudaverdian:2002ky}
It is defined as
\begin{eqnarray}
{\Delta} F = \frac{(-1)^{|F|}}{2} {\rm div}_{\brho} X_F,
\label{oddlaplace2}
\end{eqnarray}
where $F \in C^{\infty}(\Map(\calX, \calM))$
and $X_F$ is the Hamiltonian vector field of $F$.
Here, the divergence ${\rm div}$ of the vector field $X$ is defined as
$
\int_{\Map(\calX, \calM)} \brho \ ({\rm div}_{\rho} X) F 
= - \int_{\Map(\calX, \calM)} \brho \ X(F)$
for arbitrary $F \in C^{\infty}(\Map(\calX, \calM))$.
%
%
If an odd Laplace operator is given,
an odd Poisson bracket can be constructed by the derived bracket:
\begin{eqnarray}
\{F,G\}:&=& (-1)^{|F|}[[{\Delta},{F}],{G}](1)
\nonumber \\
&=& (-1)^{|F|}\Delta(FG) - (-1)^{|F|}\Delta(F)G - F\Delta(G).
\nonumber
\end{eqnarray}
The classical master equation is modified to the following equation:
\begin{eqnarray}
\Delta (e^{\frac{i}{\hbar} S_q})=0,
\nonumber
\end{eqnarray}
where $S_q$ is the quantum BV action, which is a deformation 
of a classical BV action $S_q = S + \cdots$.
This equation is equivalent to the quantum master equation:
\begin{eqnarray}
-2 i \hbar \Delta S_q + 
\sbv{S_q}{S_q} =0.
\end{eqnarray}

The above definition of the odd Laplace operator $\Delta$ is formal,
because $\Map(\calX, \calM)$ is infinite dimensional in general.
The naive measure $\brho$ is divergent and needs regularization.
Moreover, even if the graded manifold is finite dimensional,
the solutions of the quantum master equation have obstructions,
that depend on the topological properties of the base manifold.
We refer to Refs.\citenum{Bonechi:2007ar, Bonechi:2009kx, Cattaneo:2008yf}
for analyses of the 
obstructions of the quantum master equation related to the odd Laplace operator
in AKSZ theories.

%
%
%
%

\section{Deformation Theory}\label{deformationtheory}
\noindent
In this section, 
we apply the deformation theory to the AKSZ formalism of TFTs 
and determine the most
 general consistent local BV action $S$ under physical conditions.
This method is also called homological perturbation theory.

We begin with $S=S_0$.
In fact, $S_0 = S^{(0)}$ is determined from the P-structure only,
and it trivially satisfies 
the classical master equation $\sbv{S_0}{S_0}=0$.
Next, we deform $S_0$ to 
\begin{eqnarray}
S= \sum_{n=0}^{\infty} g^n S^{(n)}
= S^{(0)} + g S^{(1)} + g^2 S^{(2)} + \cdots
\label{deformationofS}
\end{eqnarray}
in order to obtain a consistent $S_1$ term,
where $g$ is the deformation parameter.
$S$ is required to satisfy 
the classical master equation 
$\sbv{S}{S}=0$ in order to be a Q-structure.

The deformation $S^{\prime}$ is equivalent to 
$S$ if there exist 
local redefinitions of superfields
$\bbe^{a(i)} \mapsto \bbe^{\prime a(i)}=F(\bbe^{a(i)})$ satisfying
$S^{\prime}(\bbe^{\prime a(i)}) = S(\bbe^{a(i)})$, 
where $F$ is a function on $\Map(\calX, \calM)$.
If we expand $\bbe^{\prime a(i)}= \sum_m g^m F^{(m)}(\bbe^{a(i)})$,
then 
$ S(\bbe^{a(i)}) 
= S^{\prime}(\bbe^{\prime a(i)}) = 
S^{\prime}(\sum_m g^m F^{(m)}(\bbe^{a(i)})) 
= S^{\prime}(\bbe^{a(i)}) 
+ g \frac{\delta S^{\prime}(\bbe^{a(i)})}
{\delta \bbe^{b(j)}}F^{(1)}(\bbe^{b(j)})
+ \cdots$.
Therefore, the difference between the two actions is BRST exact to first order in $g$:
\begin{eqnarray}
S^{\prime} - S = \pm g \hQ^{\prime} \left(\int d \bbe \ F^{(1)} \right),
\label{equivalence}
\end{eqnarray}
where $\hQ^{\prime}$ is the BRST transformation defined by $S^{\prime}$.
It has been proved that higher-order terms can be 
absorbed order by order by the BRST exact terms.
Therefore, 
$S$ is equivalent to $S_0$ by field redefinition
if the deformation is exact $S = S_0 + \delta (*)$.
Therefore, computing the $\hQ$ cohomology class is sufficient for 
determining $S$.

If we substitute equation (\ref{deformationofS}) into $\sbv{S}{S}=0$
and expand it in $g$, 
we obtain the following series of equations:
\begin{eqnarray}
&& \sbv{S^{(0)}}{S^{(0)}} =0, 
\nonumber \\
&& \sbv{S^{(0)}}{S^{(1)}} =0, 
\nonumber \\
&& 2 \sbv{S^{(0)}}{S^{(2)}} + \sbv{S^{(1)}}{S^{(1)}} =0, 
\nonumber \\
&& \cdots.
\end{eqnarray}
The first equation is already satisfied by construction.
The second equation is $\hQ_0 S^{(1)}= 0$. 
Therefore, $S^{(1)}$ is a cocycle of $\hQ_0$.

The third equation is an obstruction. 
We assume that the action is local. 
Thus,
$S^{(1)}$ and $S^{(2)}$ are integrals of local Lagrangians. 
This means that it is the transgression of a function 
$\Theta^{(2)}$ on the target space,
$S^{(2)} = \mu_* \ev^* \Theta^{(2)}$,
where $\Theta^{(2)}
\in C^{\infty}(\calM)$.
Since $\sbv{S_0}{\bbe^{a(i)}} = \bbd \bbe^{a(i)}$
for all superfields $\bbe^{a(i)}$,
$\sbv{S^{(0)}}{S^{(2)}} =\hQ_0 S^{(2)} 
=0$,
provided the integral of the total derivative terms vanishes, 
$\int_{\calX} \mu \bbd (*) = 0$.
Therefore, if we assume that $\calX$ has no boundary, 
each term must be equal to zero:
$\sbv{S^{(0)}}{S^{(2)}} =0$, $\sbv{S^{(1)}}{S^{(1)}} =0$.

From $\sbv{S^{(0)}}{S^{(2)}} =0$, we can absorb $S^{(2)}$ into 
$S^{(1)}$
by the following redefinition: $\tilde{S}^{(1)} = S^{(1)} + g S^{(2)}$.
Then, we have $\sbv{S^{(0)}}{\tilde{S}^{(1)}}=0$.
Repeating this process, we obtain 
$S = S_0 + S_1$, where $S_1 = \sum_{n=1}^{\infty} g^n S^{(n)}$.
Here, $S_1$ is an element of the cohomology class of $\hQ_0$, 
\begin{lemma}\label{S1d}
Denote $S_1 = 
\int_{\calX} \mu \
\calL_1$.
If $\calL_1$ contains a superderivative $\bbd
$, then $\calL_1$ is $\hQ_0$-exact.
\end{lemma}
\begin{proof}
It is sufficient to prove the lemma under the assumption
that $\calL_1$ is a monomial.
Assume that $\calL_1$ contains at least one derivative,
$\calL_1(\bbe) = F(\bbe) \bbd G(\bbe)$, 
where $F(\bbe)$ and $G(\bbe)$ are functions of superfields.
$F$ and $G$ can be expanded in component superfields
by the number of odd supercoordinates
$\theta^{\mu}$ as $F(\bbe) = \sum_{i=0}^{n+1} F_{i}$ 
and $G(\bbe) = \sum_{i=0}^{n+1} G_{i}$. 
$F_{i}$ and $G_{i}$ are terms of 
$i$-th order in $\theta^{\mu}$.
Since $\hQ_0 F = \bbd F$ and $\hQ_0 G = \bbd G$,
from the properties of $\hQ_0$, 
we obtain the following expansions:
\begin{eqnarray}
&& \hQ_0 F_{0} = 0, 
\nonumber \\
&& \hQ_0 F_{i} = \bbd F_{i-1} 
\quad \mbox{for $1 \leq i \leq n+1$},
\nonumber \\
&& \bbd F_{n+1} = 0, 
\nonumber \\
&& \hQ_0 G_{0} = 0, 
\nonumber \\
&& \hQ_0 G_{i} = \bbd G_{i-1} \quad \mbox{for $1 \leq i \leq n+1$},
\nonumber \\
&& \bbd G_{n+1} = 0. 
\label{expbrstike}
\end{eqnarray}
For $S_1 = \int_{\calX} \mu \calL_1(\bbe)
= \sum_{i=0}^{n} \int_{X} \mu F_{n-i} \bbd G_{i}$,
two consecutive terms 
$F_{n-i} \bbd G_{i} + F_{n-i-1} \bbd G_{i+1}$
are combined (for even $i$) as 
\begin{eqnarray}
&& F_{n-i} \bbd G_{i} + F_{n-i-1} \bbd G_{i+1}
=
(-1)^{n-i} \hQ_0( F_{n-i} G_{i+1} )
- (-1)^{n-i} \bbd( F_{n-i-1} G_{i+1} ),
\label{2termsike}
\end{eqnarray}
by $(\ref{expbrstike})$, which gives a $\hQ_0$-exact term up to a $\bbd$-exact term.

If $n$ is odd, 
$S_1 = \sum_{i=0}^{n} \int_{X} F_{n-i-1} \bbd G_{i}$ 
has an even number of terms,
and the terms can be combined as in equation (\ref{2termsike}).
Therefore, the integral $S_1$ is $\hQ_0$-exact.

If $n$ is even, the term $F_{0} \bbd G_{n}$ remains.
This term is $\hQ_0$-exact itself, since
$F_{0} \bbd G_{n} = \hQ_0 (F_{0} G_{n+1})$.
Therefore, $S_1$ is also $\hQ_0$-exact.
\end{proof}
From Lemma \ref{S1d}, nontrivial deformation terms of $S_1$ do not include $\bbd$.
The remaining condition is $\sbv{S_1}{S_1} =0$.
Therefore, the following theorem has been proved.
\begin{theorem}\label{localS1}
Assume that $\calX$ is a world-volume without boundary, that is, 
$\int_{\calX} \mu \ \bbd (*) = 0$,
and locality of the BV action.
If and only if
%
$S^{(1)}
$ is 
a $\hQ_0$-cohomology class
such that 
$\sbv{S^{(1)}}{S^{(1)}}=0$,
and $S^{(n)} = 0$ for $n \geq 2$,
then $\sbv{S}{S}=0$.
Let $S_1 = g S^{(1)} = \int_{\calX} \mu \ \calL_1(\bbe)$,
then $\calL_1(\bbe)$ is a function of a superfield $\bbe$,
which does not contain the superderivative $\bbd$.
%
\end{theorem}
%
%
%
If we relax the assumption of no-boundary or locality in Theorem 
\ref{localS1},
we obtain more general AKSZ type sigma models, such as 
the WZ-Poisson sigma model 
and the Dirac sigma model. \cite{Kotov:2004wz}

\section{AKSZ Sigma Models in Local Coordinates}
\noindent
In this section,
we give 
local coordinate expressions of the P-structure 
graded symplectic form $\bomega$, the BV antibracket,
the BV action $S$, \eqref{deformationofS} and 
the odd Laplacian in the previous section.

Let us take an ($n+1$)-dimensional manifold $X$ and
a $d$-dimensional manifold $M$.
We also take a graded manifold $\calX = T[1]X$,
 and a QP-manifold $\calM$.
%
Local coordinates on $T[1]X$ are denoted by $(\sigma^{\mu}, \theta^{\mu})$, where $\sigma^{\mu}$ is a local coordinate of degree 0 on the base manifold $X$, and 
$\theta^{\mu}$ is a local coordinate  of degree 1.

Let $\calM^{(i)}$ be the degree $i$ part of $\calM$.
Local coordinates on $\calM^{(i)}$ are denoted by
$e^{a(i)}$. 
The local coordinates $e^{a(i)}$ are also denoted by
\begin{enumerate}
 \item $x^{a(0)}$ of degree $0$
 \item $q^{a(i)}$ of degree $i$, for $0 \leq i \leq \floor{n/2}$
 \item $p_{a(n-i)}$ of degree $n-i$, for $\floor{n/2} < i \leq n$
 \item $\xi_{a(0)}$ of degree $n$
\end{enumerate}
where $\floor{m}$ is the floor function 
(that is, its value is the largest integer less than 
or equal to $m$). 
\footnote{Indices $a(i)$ run $a(i) = 1, 2, \cdots, {\rm dim} \calM^{(i)}$.
}

As explained in Section \ref{AKSZTFTgeneral}, fields in a classical field theory correspond to maps $\calX \rightarrow \calM$. 
Local coordinates on the mapping space are superfields,
which we denote by the corresponding boldface letters.
$\bbx^{a(0)}$ of degree 0
is a smooth map $\bbx^{a(0)}: T[1]X \longrightarrow M$, and 
superfields $\bbe^{a(i)}$ of degree $i$ are bases of
sections of $T^*[1]X \otimes \bbx^*(\calM^{(i)})$,
for $1 \leq i \leq n$.
$\bbx^{a(0)}$ is also denoted by $\bbe^{a(0)}$ and
$\bbe^{a(n)}$ by $\bbxi_{a(0)}$.
%

The P-structure can be written as
\begin{eqnarray}\label{BVbracket1}
\bomega &=& 
\int_{\calX} 
\mu 
\ 
\left(\frac{1}{2}
\delta \bbe_{a(i)}
\wedge \bomega^{a(i)b(j)} 
\delta \bbe_{b(j)}
\right) 
\nonumber \\
&=& 
\sum_{i=0}^{\floor{n/2}}
\int_{\calX} d^{n+1}\sigma d^{n+1}\theta \
(-1)^{ni}
\delta \bbq^{a(i)}
\wedge 
\delta \bbp_{a(i)},
\label{gPoissonDarboux}
\end{eqnarray}
where we used Darboux coordinates,
$\bbq^{a(i)}$, for $0 \leq i \leq \floor{n/2}$,
and
$\bbp_{a(n-i)}$, for $\floor{n/2} < i \leq n$.
This defines the graded Poisson bracket such that
\begin{eqnarray}\label{BVbracket3}
\left\{\bbq^{a(i)}(\sigma, \theta),
\bbp_{b(j)}(\sigma^{\prime}, \theta^{\prime})\right\} 
&=& \delta^{i}{}_{j} \delta^{a(i)}{}_{b(j)} 
\delta^{n+1}(\sigma-\sigma^{\prime}) 
\delta^{n+1}(\theta-\theta^{\prime}).
\nonumber
\end{eqnarray}
If $n$ is even, 
$\bbp_{a(n/2)}$ is identified with
$k_{ab} \bbq^{b(n/2)}$ and
the degree ($n/2$) part of the P-structure symplectic form can be written as
\begin{eqnarray}\label{BVbracket2}
\int_{\calX} 
d^{n+1}\sigma d^{n+1}\theta
\ 
\left(
\frac{1}{2}
\delta \bbq^{a(n/2)}
\wedge k_{ab} 
\delta \bbq^{b(n/2)}
\right),
\nonumber
\end{eqnarray}
where $k_{ab}$ is a fiber metric.
%
%
%
%
%
The corresponding Poisson bracket of the part, for which $i = j = n/2$, is
\begin{eqnarray}\label{BVbracket4}
\left\{\bbq^{a(n/2)}(\sigma, \theta),
\bbq^{b(n/2)}(\sigma^{\prime}, \theta^{\prime}) \right\} 
&=& k^{a(n/2)b(n/2)}
\delta^{n+1}(\sigma-\sigma^{\prime}) 
\delta^{n+1}(\theta-\theta^{\prime}).
\nonumber
\end{eqnarray}
%
%
%
%

The differential $D$ on the differential graded manifold $\calX$ is induced from
the exterior derivative $d$ on $X$.
This defines a superdifferential 
$\bbd = \theta^{\mu} \frac{\partial}{\partial \sigma^{\mu}}$
on $\Map(\calX, \calM)$.

Next, let us consider the local coordinate expression of the Q-structure $S$
on the mapping space.
From the definition in Section \ref{AKSZTFTgeneral}, $S$ has two terms, $S = S_0 + S_1$.

$S_0$ is determined from the P-structure.
If $n$ is odd, 
\begin{eqnarray}
S_0
&=&
\int_{\calX} 
\mu \
\frac{1}{2} \bbe_{a(i)} 
\bomega^{a(i)b(j)}
\bbd \bbe_{b(j)}
\nonumber \\
&=& 
\int_{\calX} 
d^{n+1} \sigma d^{n+1} \theta 
\ 
\sum_{0 \leq i \leq (n-1)/2} 
(-1)^{n+1-i} \bbp_{a(i)} \bbd \bbq^{a(i)}
\nonumber \\
&=&
\int_{\calX} 
d^{n+1} \sigma d^{n+1} \theta \ 
\left( (-1)^{n+1} \bbxi_{a(0)} \bbd \bbx^{a(0)}
+ 
\sum_{1 \leq i \leq (n-1)/2} 
(-1)^{n+1-i}  \bbp_{a(i)} \bbd \bbq^{a(i)}
\right),
\label{oddSzero}
\end{eqnarray}
and if $n$ is even,
\begin{eqnarray*}
S_0
&=&
\int_{\calX} 
\mu \
\frac{1}{2} \bbe_{a(i)} 
\omega^{a(i)b(j)}
\bbd \bbe_{b(j)}
\nonumber \\
&=&
\int_{\calX} 
d^{n+1} \sigma d^{n+1} \theta 
\ 
\left( 
\sum_{0 \leq i \leq (n-2)/2} 
(-1)^{n+1-i} \bbp_{a(i)} \bbd \bbq^{a(i)}
+ (-1)^{\frac{n+1}{2}} k_{a(n/2)b(n/2)} \bbq^{a(n/2)} \bbd \bbq^{b(n/2)}
\right)
\nonumber \\
&=&
\int_{\calX} 
d^{n+1} \sigma d^{n+1} \theta 
\ \left( 
(-1)^{n+1} \bbxi_{a(0)} \bbd \bbx^{a(0)}
+ 
\sum_{1 \leq i \leq (n-2)/2} 
(-1)^{n+1-i} \bbp_{a(i)} \bbd \bbq^{a(i)}
\right.
\nonumber \\
&& \qquad 
\left.
+ (-1)^{\frac{n+1}{2}}
k_{a(n/2)b(n/2)} \bbq^{a(n/2)} \bbd \bbq^{b(n/2)}\right).
\label{evenSzero}
\end{eqnarray*}
If we define
$\bbp_{a(n/2)} \equiv k_{a(n/2)b(n/2)} \bbq^{a(n/2)}$, 
then the $S_0$'s for odd and even $n$
can be unified to 
\begin{eqnarray*}
S_0
&=&
\int_{\calX} 
\mu \
\frac{1}{2} \bbe_{a(i)} 
\omega^{a(i)b(j)}
\bbd \bbe_{b(j)}
\nonumber \\
&=&
\int_{\calX} 
d^{n+1} \sigma d^{n+1} \theta 
\ 
\left( 
\sum_{0 \leq i \leq \floor{n/2}} 
(-1)^{n+1-i} \bbp_{a(i)} \bbd \bbq^{a(i)}
\right).
\end{eqnarray*}

A superfield of degree $i$,
$\bPhi(\sigma, \theta)$,
can be expanded by $\theta^{\mu}$ as
\begin{eqnarray}
\bPhi(\sigma, \theta) 
= \sum_{k} \bPhi^{(k)}(\sigma, \theta)
= \sum_{k} \frac{1}{k!}
\theta^{\mu(1)} \cdots \theta^{\mu(k)}
\Phi^{(k)}_{\mu(1)\cdots \mu(k)}(\sigma), 
\nonumber
\end{eqnarray}
where $\Phi^{(k)}_{\mu(1)\cdots \mu(k)}(\sigma)$ 
depends only on $\sigma^{\mu}$.
Since $\theta^{\mu}$ has degree 1,
%
$\Phi^{(k)}_{\mu(1)\cdots \mu(k)}(\sigma)$ has degree $i-k$. 
This
is the same as the {\it ghost number} in gauge theory. The
fields $\Phi^{(k)}_{\mu(1)\cdots \mu(k)}(\sigma)$ are classified 
by their ghost numbers.
If $\Phi^{(k)}_{\mu(1)\cdots \mu(k)}(\sigma)$ has degree 0,
it is a physical field.
In particular, it is a $k$-th order antisymmetric tensor field.
If $\Phi^{(k)}_{\mu(1)\cdots \mu(k)}(\sigma)$ has positive degree,
it is a ghost field, or it is a ghost for ghosts, etc.
If $\Phi^{(k)}_{\mu(1)\cdots \mu(k)}(\sigma)$ has negative degree,
it is the Hodge dual of the antifield that is introduced in the BV formalism.

Let us consider expansions of the Darboux coordinate superfields: 
\begin{eqnarray}
\bbq^{a(i)}(\sigma, \theta) 
&=& \sum_{k} \frac{1}{k!}
\theta^{\mu(1)} \cdots \theta^{\mu(k)}
q^{(k), a(i)}_{\mu(1)\cdots \mu(k)}(\sigma),
\\
\bbp_{a(i)}(\sigma, \theta) 
&=& \sum_{k} \frac{1}{k!}
\theta^{\mu(1)} \cdots \theta^{\mu(k)}
p^{(k)}_{a(i), \mu(1)\cdots \mu(k)}(\sigma).
\end{eqnarray}
The antifield for the ghost 
$q^{(k), a(i)}_{\mu(1)\cdots \mu(k)}(\sigma)$ for $i-k>0$ 
is $p^{(n+1-k)}_{a(i), \mu(1)\cdots \mu(n+1-k)}(\sigma)$,
and 
the antifield for the ghost 
$p^{(k), a(i)}_{\mu(1)\cdots \mu(k)}(\sigma)$ for $k-i>0$ 
is $q^{(n+1-k)}_{a(i), \mu(1)\cdots \mu(n+1-k)}(\sigma)$.
We can see that
this coincides with the BF theory for abelian $i$-form fields that was presented in Section 3.
Note that, if $n$ is even,
a superfield of degree $i= n/2$ is a self-conjugate superfield
$\bbq_{a(n/2)}(\sigma, \theta) 
= 
\sum_{k, \mu(k)} 
\frac{1}{k!}
\theta^{\mu(1)} \cdots \theta^{\mu(k)}
q^{(k)}_{a(n/2), \mu(1)\cdots \mu(k)}(\sigma)
$.
The antifield 
$q^{(k), a(n/2)}_{\mu(1)\cdots \mu(k)}(\sigma)$ 
for 
$k \leq n/2$ is
$q^{(n+1-k), a(n/2)}_{\mu(1)\cdots \mu(n+1-k)}(\sigma)$,
which is contained in the same superfield.

If the component fields of nonzero ghost number are set to zero
and the $d\theta$ integration is carried out,
we obtain the kinetic term for a BF theory
of general $k$-forms:
\begin{eqnarray*}
S_0 
=
S_A
&=& \sum_{0 \leq i \leq \floor{n/2}} 
\frac{1}{i!(n-i)!}
\int_{\calX} 
d^{n+1} \sigma 
\ 
(-1)^{n+1-i} 
\epsilon^{\mu(0)\cdots\mu(n)} 
p^{(n-i)}_{a(i), \mu(i+1)\cdots \mu(n)}
\partial_{\mu(i)}
q^{(i), a(i)}_{\mu(0)\cdots \mu(i-1)}.
\label{abelianBF2}
\end{eqnarray*}
This coincides with the action $S_A$ given in Section \ref{ABFiform}.


The interaction term $S_1$ was determined in 
Theorem \ref{localS1} in Section \ref{deformationtheory}.
The local coordinate expression of $S_1$ is as follows:
\begin{eqnarray*}
S_1&=&
\sum_{\lambda, a(\lambda), |\lambda| = n + 1} 
\int_{\calX} \mu 
\ 
\left(
f_{\lambda, a(\lambda_1)\cdots a(\lambda_m)}(\bbx)
\bbe^{a(\lambda_1)} \bbe^{a(\lambda_2)} \cdots \bbe^{a(\lambda_m)}
\right),
\end{eqnarray*}
where the integrand contains arbitrary functions of superfields 
of degree $n+1$ without the superderivative.
$f_{\lambda, a(\lambda_1)\cdots a(\lambda_m)}(\bbx)$ is a
local structure function of $\bbx$
and $|\lambda| = \sum_k \lambda_k$.
The consistency condition 
$\sbv{S_1}{S_1} =0$ imposes algebraic conditions on the
structure functions $f_{\lambda, a(\lambda_1)\cdots a(\lambda_m)}(\bbx)$.
Since $S_1= \int_{\calX} \mu \ \ev^* \Theta$, this consistency condition
is equivalent to $\sbv{\Theta}{\Theta} =0$,
and determines the mathematical structure on the target space.
Thus, by solving $\sbv{\Theta}{\Theta} =0$, 
we obtain consistent local expressions for the AKSZ sigma models
in $n+1$ dimensions.
\medskip\\
\indent
Finally, we give the expression of the odd Laplace operator, which appears in the quantum BV master equation.
Let $\brho = 
\rho_v 
d^{n+1} \bq d^{n+1} \bp$ be a volume form on 
$\Map(\calX, \calM)$.
The odd Laplace operator,
\begin{eqnarray}
\Delta F = \frac{(-1)^{|F|}}{2} {\rm div}_{\brho} X_F,
\end{eqnarray}
can be written as 
\begin{eqnarray}
\Delta = 
\int_{\calX} d^{n+1}\sigma d^{n+1}\theta \sum_{i=0}^{n} (-1)^i 
\frac{\partial}{\partial \bq^{a(i)}}\frac{\partial}{\partial \bp_{a(i)}}
+\frac{1}{2} \sbv{\ln \rho_v}{-}.
\end{eqnarray}
If we take coordinates such that $\rho_v=1$, 
we obtain the following simple expression:
\begin{eqnarray}
\Delta = \int_{\calX} d^{n+1}\sigma d^{n+1}\theta \sum_{i=0}^{n} (-1)^i 
\frac{\partial}{\partial \bq^{a(i)}}\frac{\partial}{\partial \bp_{a(i)}}.
\label{oddlaplacelocal2}
\end{eqnarray}
%

\section{Examples of AKSZ Sigma Models}
\noindent
In this section, we list some important examples.
\subsection{$n=1$}
\subsubsection{The Poisson Sigma Model}\label{AKSZPSM}
We take $n=1$. 
In Example \ref{pmanifold} we showed that 
a QP-structure of degree 1 on
$\calM = T^*[1]M$ is equivalent to a Poisson structure on $M$.
Let $X$ be a two-dimensional manifold, and let $\calX = T[1]X$.
The AKSZ construction defines a TFT on $\Map(T[1]X, T^*[1]M)$.

Let $\bbx^{i}$ be a map from $T[1]X$ to $M$, and let 
$\bbxi_{i}$ be a section of $T^*[1]X \otimes \bbx^*(T^*[1]M)$,
which are superfields induced by the local coordinates $(x^i, \xi_i)$.
Here, we denote the indices $a(0), b(0)$ by $i,j$.
The P-structure on $\Map(T[1]X, T^*[1]M)$ is
\begin{eqnarray*}
\bomega &=& \int_{\calX} 
d^{2}\sigma d^{2}\theta
\ \delta \bbx^{i} \wedge \delta \bbxi_{i}.
\end{eqnarray*}
The BV action (Q-structure) is
\begin{eqnarray}
S &=& \int_{\calX} 
d^{2}\sigma d^{2}\theta
\ \left(\bbxi_{i} \bbd \bbx^{i}
+ 
\frac{1}{2}
f^{ij}(\bbx)
\bbxi_{i} \bbxi_{j} 
\right).
\label{2DPSMBVaction}
\end{eqnarray}
This action is the superfield BV formalism of
the Poisson sigma model, 
where the superfields are identified with 
$\bbx^i = \bphi^i$ and $\bbxi_i = \ba_i$.
The Q-structure condition is equivalent to 
equation (\ref{PoissonJacobi}) on $f^{ij}(x)$.

Take $M=\mathfrak{g}^*$,
where $\mathfrak{g}$ is a semi-simple Lie algebra.
Then, $\calM = T^*[1]\mathfrak{g}^*$, and the Q-structure reduces to 
$\Theta = \frac{1}{2} f^{ij}{}_k x^k \xi_{i} \xi_{j}$,
where $f^{ij}{}_k$ is a structure constant of the Lie algebra.
The AKSZ construction yields the BV action
\begin{eqnarray*}
S&=& \int_{\calX} 
d^{2}\sigma d^{2}\theta
\ \left(\bbxi_{i} \bbd \bbx^{i}
+ 
\frac{1}{2}
f^{ij}{}_k \bbx^k 
\bbxi_{i} \bbxi_{j} 
\right),
\end{eqnarray*}
which is the BV formalism of a nonabelian BF theory in two dimensions.
\subsubsection{B-Model}\label{BModel}
Let $X$ be a Riemann surface, and $M$ a complex manifold.
Let us consider the supermanifold
$\calX = T[1]X$
and 
the QP-manifold $\calM = T^*[1] T[1]M$
given in Example \ref{ComplexStructure}.
This QP-manifold realizes a complex structure.
The AKSZ construction for $n=1$ induces a TFT on
$
\Map(T[1] X, T^*[1]T[1]M).
$

Let $\bbx$ be 
$\bbx: T[1] X \longrightarrow M$, let  
$\bbxi$ be a section of $T^*[1] X \otimes \bbx^*(T^* [1] M)$,
let $\bbq$ be a section of $T^*[1] X \otimes \bbx^*(T [1] M)$,
and let
$\bbp$ be a section of $T^*[1] X \otimes \bbx^*(T^* [0] M)$.
The superfield expression of the P-structure is
\begin{eqnarray*}
\bomega &=& \int_{\calX} 
d^{2}\sigma d^{2}\theta
\ (\delta \bbx^{i} \wedge \delta \bbxi_{i} 
-
\delta \bbq^{i} \wedge \delta \bbp_{i}).
\end{eqnarray*}
The Q-structure BV action is 
\begin{eqnarray}
S_B &=& \int_{\calX}
d^{2}\sigma d^{2}\theta
\left(
\bbxi_{i} \bbd \bbx^i - \bbp_i \bbd \bbq^i 
+ J^i{}_j(\bbx) \bbxi_{i} \bbq^j
+ \frac{\partial J^i{}_k}{\partial \bbx^j}(\bbx) 
\bbp_{i} \bbq^j \bbq^k
\right)
\nonumber \\
&=& \int_{\calX}
d^{2}\sigma d^{2}\theta
\left[\left(
\begin{matrix}
\bbxi_{i} \ \bbq^i
               & \cr
\end{matrix}
\!\!\right)
\bbd \left(
\begin{matrix}
\bbx^i & \cr
                \bbp_{i} & \cr
\end{matrix}
\!\!\right)
+ \left(
\begin{matrix}
\bbxi_{i} \ \bbq^i& \cr
\end{matrix}
\!\!\right)
\left(
\begin{matrix}
0 & \frac{1}{2}J^i{}_j(\bbx)& \cr
                - \frac{1}{2}J^j{}_i(\bbx) &
 \frac{\partial J^i{}_k}{\partial \bbx^j}(\bbx) 
\bbp_{i}  & \cr
\end{matrix}
\!\!\right)
 \left(
\begin{matrix}
\bbxi_{j} & \cr
                \bbq^j & \cr
\end{matrix}
\!\!\right)
\right].
\nonumber
\label{bmodelAKSZ}
\end{eqnarray}
Proper gauge fixing of this action describes the so-called 
B-model action of a topological string.
\cite{Alexandrov:1995kv, Ikeda:2007rn}

\subsection{$n=2$}
\subsubsection{The Courant Sigma Model}\label{CourantSM}
We consider the case, where 
$\calM$ is a QP-manifold of degree $n=2$.
Here, $\calM$ has the Courant algebroid structure,
discussed in Example \ref{Calgebroid}.
We take a three-dimensional manifold $X$ 
and consider $\calX = T[1]X$ as the world-volume of the AKSZ sigma model.
Let $\bbx^{i}$ be a map from $T[1]X$ to $M = \calM^{(0)}$, 
$\bbxi_{i}$ be a section of $T^*[1]X \otimes \bbx^*(\calM^{(2)})$
and $\bqone^{a}$ be 
a section of $T^*[1]X \otimes \bbx^*(\calM^{(1)})$.
$k_{ab}$ is a fiber metric on $\calM^{(1)}$.
Here, we denote $a(0), b(0), \cdots$
by $i, j, \cdots$ and
$a(1), b(1), \cdots$ by $a, b, \cdots$.
The P-structure on $\Map(\calX, \calM)$ is
\begin{eqnarray*}
\bomega &=& \int_{\calX} 
d^{3}\sigma d^{3}\theta
\ 
\left(\delta \bbx^{i} \wedge \delta \bbxi_{i} 
+
\frac{1}{2} k_{ab} \delta \bqone^{a} \wedge \delta \bqone^{b}
\right),
\end{eqnarray*}
and the Q-structure BV action has the following form:
\begin{eqnarray}
S&=& \int_{\calX} d^{3}\sigma d^{3}\theta
\ \left(- \bbxi_{i} \bbd \bbx^{i}
+ \frac{1}{2} k_{ab} \bqone^{a} \bbd \bqone^{b}
+ f_{1}{}^{i}{}_{a}(\bbx)
\bbxi_{i} \bqone^{a} 
+ \frac{1}{3!}
f_{2}{}_{abc}(\bbx)
\bqone^{a} \bqone^{b} \bqone^{c}
\right).
\label{3DCSM}
\end{eqnarray}
This model has the Courant algebroid structure given in 
Theorem \ref{CalgeQP22}, and therefore, it is called the \textsl{Courant sigma model}
\cite{Ikeda:2002wh, Ikeda:2002qx, Hofman:2002rv, Roytenberg:2006qz}.

We can derive the action of the physical fields 
from equation (\ref{3DCSM}) 
by setting the components of the nonzero ghost number to zero:
$\bbx^{i} = \bbx^{(0) i} = x^i$,
$\bbxi_{i} = \bbxi^{(2)}_{i} 
= \frac{1}{2} \theta^{\mu}\theta^{\nu} \xi^{(2)}_{\mu\nu, i}$
and 
$\bqone^{a} = \bqone^{(1)a} = \theta^{\mu} \qone^{(1)a}_{\mu}$.
Then, we obtain 
\begin{align}
S&= \int_{X} 
\left(- \xi_{i} \wedge d x^{i}
+ \frac{1}{2} k_{ab} \qone^{a} \wedge d \qone^{b}
+ f_{1}^{i}{}_{a}(x)
\xi_{i} \wedge \qone^{a} 
+ \frac{1}{3!}
f_{2}{}_{abc}(x)
\qone^{a} \wedge \qone^{b} \wedge \qone^{c}
\right),
\label{3DCSM2}
\end{align}
where $d$ is the exterior differential on $X$,
$\xi_{i} = \frac{1}{2} d \sigma^{\mu} \wedge d \sigma^{\nu} 
\xi^{(2)}_{\mu\nu, i}$ and 
$\qone^{a} = d \sigma^{\mu} \qone^{(1)a}_{\mu}$.

\subsubsection{Chern-Simons Gauge Theory}\label{cstheory}
In the Courant sigma model, (\ref{3DCSM2}),
if we take $\xi_i=0$, $f_{1}^{i}{}_{a}(x) = 0$ and 
$f_{2}{}_{abc}(x)= f_{2}{}_{abc}=$ constant,
the action reduces to the Chern-Simons theory:
\begin{eqnarray}
S&=& \int_{X} 
\left(
\frac{1}{2} k_{ab} A^{a} \wedge d A^{b}
+ \frac{1}{3!}
f_{2}{}_{abc}
A^{a} \wedge A^{b} \wedge A^{c}
\right),
\label{3DChernSimons}
\end{eqnarray}
where we denote the $1$-form by $A^a = \qone^a$.
Therefore, the Chern-Simons theory can be obtained by the AKSZ construction.

In fact, the AKSZ construction in three dimensions for
a Lie algebra target space 
yields the Chern-Simons theory.
Let $\mathfrak{g}$ be a Lie algebra
and let $k_{ab}$ be a metric on $\mathfrak{g}$.
If $\mathfrak{g}$ is semi-simple, we can take $k_{ab}$ as 
the Killing metric.
Note that $\calM = \mathfrak{g}[1] 
$ has QP-manifold structure of degree $2$,
and $M = \{pt\}$.
The P-structure is defined as
\begin{eqnarray*}
\omega =
\frac{1}{2} k_{ab} \delta \qone^{a} \wedge \delta \qone^{b},
\end{eqnarray*}
where $a=a(1), b=b(1), \cdots $.
The Q-structure is
\begin{eqnarray*}
\Theta = \frac{1}{3!} f{}_{abc}
\qone^{a} \qone^{b} \qone^{c},
\end{eqnarray*}
where $f{}_{abc}$ is the structure constant of $\mathfrak{g}$.

Let $X$ be a three-dimensional manifold and $\calX = T[1]X$.
Then, $\bqone^{a}$ is a section of $T^*[1]X \otimes \bbx^*(\mathfrak{g}[1])$.
The AKSZ construction 
on $\Map(T[1]X, \mathfrak{g}[1])$
yields the P-structure:
\begin{eqnarray*}
\bomega &=& \int_{\calX} 
d^{3}\sigma d^{3}\theta
\ \frac{1}{2} k_{ab} \delta \bqone^{a} \wedge \delta \bqone^{b}
\end{eqnarray*}
and the Q-structure function
\begin{eqnarray*}
&& S= 
\int_{\calX} 
d^{3}\sigma d^{3}\theta
\ \left(\frac{1}{2} k_{ab} \bqone^{a} \bbd \bqone^{b}
+ \frac{1}{3!} f{}_{abc}
\bqone^{a}\bqone^{b}\bqone^{c}
\right).
\end{eqnarray*}
The action satisfies $\sbv{S}{S} =0$.
This is the AKSZ sigma model of the action (\ref{3DChernSimons}) for
the Chern-Simons theory in three dimensions \cite{Alexandrov:1995kv},
which coincides with the BV action obtained in Ref.~\citenum{Axelrod:1991vq}.
\subsection{$n=3$}
\subsubsection{AKSZ Sigma Model in $4$ Dimensions}\label{4dtft}
We take $n=3$.
Then, $X$ is a four-dimensional manifold, and
$\calM$ is the QP-manifold of degree $3$ in Example \ref{Talgebroid}.
Let $\bbx^{i}$ be a map from $T[1]X$ to $M = \calM^{(0)}$ and 
$\bbxi_{i}$ be a section of $T^*[1]X \otimes \bbx^*(\calM^{(3)})$.
Let $\bbq^{a}$ be a section of $T^*[1]X \otimes \bbx^*(\calM^{(1)})$
and 
$\bbp_{a}$ be a section of $T^*[1]X \otimes \bbx^*(\calM^{(2)})$.
Here, we denote $a(0), b(0), \cdots$ by $i, j, \cdots$ and $a(1), b(1), \cdots $ by $a, b, \cdots$.
Note that $(\bbx^{i}, \bbxi_{i}, \bbq^a, \bbp_a)$ are superfields 
of degrees $(0, 3, 1, 2)$.
The P-structure is
\begin{eqnarray*}
\bomega &=& \int_{\calX} 
d^{4}\sigma d^{4}\theta
\ \left(
\delta \bbx^{i} \wedge \delta \bbxi_{i}
-
\delta \bbq^{a} \wedge \delta \bbp_{a}
\right).
\end{eqnarray*}
The Q-structure funciton is
\begin{eqnarray*}
S&=&S_0 + S_1, \\
S_0&=& \int_{\calX} 
d^{4}\sigma d^{4}\theta
\ (\bbxi_{i} \bbd \bbx^{i}
- \bbp_{a} \bbd \bbq^{a}),
\nonumber\\
S_1 &=&
\int_{\calX} 
d^{4}\sigma d^{4}\theta
\ 
\Bigl(
f_{1}{}^{i}{}_{a}(\bbx)
\bbxi_{a} \bbq^{i} 
+ \frac{1}{2}
f_{2}{}^{ab}(\bbx)
\bbp_{a} \bbp_{b}
+ \frac{1}{2}
f_{3}{}^{a}{}_{bc}(\bbx)
\bbp_{a} \bbq^{b} \bbq^{c}
+ \frac{1}{4!}
f_{4}{}_{abcd}(\bbx)
\bbq^{a} \bbq^{b} \bbq^{c} \bbq^{d}
\Bigr).
\end{eqnarray*}
This topological sigma model has the structure of 
a Lie 3-algebroid, which is also called a Lie algebroid up to homotopy or  H-twisted Lie algebroid,
that appeared in Example \ref{Talgebroid}.
\cite{Ikeda:2010vz, Grutzmann}

\subsubsection{Topological Yang-Mills Theory}\label{topoYM}
We consider a semi-simple Lie algebra $\mathfrak{g}$
and a graded vector bundle 
$\calM = T^*[3] \mathfrak{g}[1]
\simeq
\mathfrak{g}^*[2]\oplus\mathfrak{g}[1]$ of degree 3
on a point 
$
M = 
\{pt\}$.
The world-volume supermanifold is $\calX = T[1]X$, where
$X$ is a four-dimensional manifold.
Then, 
$\bbq^{a}$ is a section of $T^*[1]X \otimes \bbx^*(\mathfrak{g}[1])$ and
$\bbp_{a}$ is a section of $T^*[1]X \otimes \bbx^*(\mathfrak{g}^*[2])$,
where $a(1)=a, b(1)=b, \cdots$.
The P-structure is
\begin{eqnarray*}
\bomega &=& \int_{\calX} 
d^{4}\sigma d^{4}\theta
\ 
\left(
- \delta \bbq^{a} \wedge \delta \bbp_{a}
\right).
\end{eqnarray*}
The dual space $\mathfrak{g}^{*}$ has the metric
$(\cdot,\cdot)_{K^{-1}}$, which is
the inverse of the Killing form on $\mathfrak{g}$.
We can define the Q-structure
\begin{eqnarray}
\Theta =
k^{ab}p_a p_b
+\frac{1}{2} f{}^a{}_{bc} p_a q^b q^c,
\label{killingQ}
\end{eqnarray}
where $q^a$ is a coordinate on $\mathfrak{g}[1]$, 
$p_a$ is a coordinate on $\mathfrak{g}^*[2]$,
$k^{ab} p_a p_b:=(p_a,p_b)_{K^{-1}}$ and 
$f^{a}{}_{bc}$ is the structure constant of the Lie algebra
$\mathfrak{g}$.
The AKSZ construction determines
the following BV action:
\begin{eqnarray*}
&& S= \int_{\calX} 
d^{4}\sigma d^{4}\theta
\ (- \bbp_a \bbF^a
+ k^{ab} \bbp_a \bbp_b),
\end{eqnarray*}
where $\bbF^a = \bbd \bq^a - \frac{1}{2} f{}^a{}_{bc} \bq^b \bq^c$.
This derives a topological Yang-Mills theory,
if we integrate out $\bbp_a$ and
make a proper gauge fixing of the remaining superfields. 
\cite{Ikeda:2011xr}

\subsection{General $n$}
\subsubsection{Nonabelian BF Theories in $n+1$ Dimensions}
Let $n \geq 2$, and let $\mathfrak{g}$ be a Lie algebra.
$X$ is an ($n+1$)-dimensional manifold, and we define $\calX = T[1]X$.
We consider $\calM = T^*[n] \mathfrak{g}[1]
\simeq \mathfrak{g}[1] \oplus \mathfrak{g}^*[n-1]$
with a point base manifold, $M = \{pt\}$.
Let $\bbq^{a}$ be a section of $T^*[1]X \otimes \bbx^*(\mathfrak{g}[1]])$
of degree 1, and $\bbp_{a}$ be a section of 
$T^*[1]X \otimes \bbx^*(\mathfrak{g}^*[n-1]])$
of degree $n$.
Here, we denote $a(1)=a, b(1)=b, \cdots $.
%
The P-structure is defined as
\begin{eqnarray*}
\bomega &=& \int_{\calX} 
d^{n+1}\sigma d^{n+1}\theta
\ (-1)^{n|q|} \delta \bbq^{a} \wedge \delta \bbp_{a}.
\end{eqnarray*}
The curvature is defined as
$\bbF^{a} = \bbd \bbq^{a} 
+ (-1)^{n} \frac{1}{2} f{}^{a}{}_{bc} 
\bbq^{b}\bbq^{c}
$.
The BV action is 
\begin{eqnarray*}
S
&=& 
\int_{\calX} 
d^{n+1}\sigma d^{n+1}\theta
\ ((-1)^{n} \bbp_{a} \bbF^{a})
\nonumber \\
&=& 
\int_{\calX} d^{n+1}\sigma d^{n+1}\theta
\ \left((-1)^{n} \bbp_{a} \bbd \bbq^{a} 
+ \frac{1}{2} f{}^{a}{}_{bc} \bbp_{a} \bbq^{b} \bbq^{c}
\right).
\end{eqnarray*}
The master equation $\sbv{S}{S}=0$ is easily confirmed.
This action is equivalent to the BV formalism of
a nonabelian BF theory in $n+1$ dimensions.
\cite{Cattaneo:2000rt, Cattaneo:2000mc}

\subsubsection{Nonassociative Topological Field Theory}
We consider the QP-structure that was presented in Example \ref{nplusoneform}.
We obtain a TFT with a nontrivial nonassociativity 
based on a Lie n-algebroid structure.

$\calM$ is a QP-manifold of degree $n$, 
$X$ is an ($n+1$)-dimensional manifold, and $\calX = T[1]X$.
From the Q-structure $\Theta$ in Example \ref{nplusoneform},
the BV action $S=S_0 + S_1$ on $\Map(\calX, \calM)$ is 
constructed by the AKSZ construction.
When $n$ is odd, $S_0$ has the form of equation (\ref{oddSzero}), and when $n$ is even, it has the form of equation (\ref{evenSzero}).
$S_1$ has the following expression:
\begin{eqnarray*}
S_1&=&\int_{\calX} 
\mu 
\ \ev^* \Theta
= \int_{\calX} \mu \ \ev^* 
(\Theta_0+\Theta_{2}+\Theta_{3}+\cdot\cdot\cdot+\Theta_{n}),
\end{eqnarray*}
where the $\Theta_i$'s are given in \eqref{nonassociativeTheta0}
and \eqref{nonassociativeThetai}.
After transgression, we obtain the superfield expressions, 
$$
\int_{\calX} \mu \ \ev^* \Theta_0 = 
\int_{\calX} 
d^{n+1}\sigma d^{n+1}\theta 
\ (f_{0}{}^{a(0)}{}_{b(1)}(\bbx)
\bbxi_{a(0)} \bbq^{b(1)} )
$$
and 
\begin{eqnarray*}
&& \int_{\calX} \mu \ \ev^* \Theta_i 
\nonumber \\
&=&
\int_{\calX} 
d^{n+1}\sigma d^{n+1}\theta 
\ 
\left(\frac{1}{i!}
f_{i,}{}_{a(n-i+1)}{}_{b_1(1) \cdots b_i(1)}(\bbx)
\bbe^{a(n-i+1)} \bbq^{b_1(1)} \cdots \bbq^{b_i(1)} \right).
\end{eqnarray*}
In particular, for the $(n+1)$-form $\Theta_{n}$, 
\begin{eqnarray*}
&&
\int_{\calX} \mu \ \ev^* \Theta_n
\nonumber \\
&=&
\int_{\calX} 
d^{n+1}\sigma d^{n+1}\theta 
\ 
\left(\frac{1}{(n+1)!}
f_{n,}{}_{b_0(1)b_1(1) \cdots b_n(1)}(\bbx)
\bbq^{b_0(1)} \bbq^{b_1(1)} \cdots \bbq^{b_n(1)}
\right).
\end{eqnarray*}
The master equation $\sbv{S}{S}=0$ defines the structure of 
the $(i+1)$-forms $\Theta_{i}$.

\section{AKSZ Sigma Models with Boundary}\label{AKSZboundary}
\noindent
So far,
we have considered AKSZ sigma models on a closed base manifold $X$.
In this section, we will consider AKSZ models,
where the base manifold $X$ has 
boundaries.
These have important applications.
In the case where $n=1$, it corresponds to a topological open string
and it yields the deformation quantization formulas \cite{Cattaneo:1999fm}.
The quantization of the $n=1$ case will be discussed below.
If $n \geq 2$, the theory describes a topological open $n$-brane
\cite{Park2000au, Hofman:2002rv}.

\subsection{$n=2$: WZ-Poisson Sigma Model}
\noindent
We will explain the construction of the AKSZ theory with boundary
using the WZ-Poisson sigma model, the simplest nontrivial example.
Nontrivial boundary structures are described in supergeometry terminology.


%
We take $n=2$ and the target graded manifold $\calM = T^*[2]T^*[1]M$.
As discussed,
$T^*[2]T^*[1]M$ has a natural QP-manifold structure.
Let $x^i$ be a coordinate of degree $0$ on $M$, 
$q^i$ be a coordinate of degree $1$ on the fiber of $T[1]M$, 
$p_i$ be a coordinate of degree $1$ on the fiber of $T^*[1]M$, 
and $\xi_i$ be a coordinate of degree $2$ on the fiber of $T^*[2]M$.

We take the following P-structure:
\begin{eqnarray}
\omega &=& 
\delta x^{i} \wedge \delta \xi_{i} +
\delta q^{i} \wedge \delta p_i.
\label{WZPpstructure}
\end{eqnarray}
By introducing a $3$-form $H$ on $M$,
the Q-structure function is defined as
\begin{eqnarray}
\Theta = \xi_{i} q^{i} 
+ \frac{1}{3!}
H_{ijk}(x) q^{i} q^{j} q^{k}.
\end{eqnarray}
Note that $\sbv{\Theta}{\Theta}=0$ is equivalent to $dH=0$.

Let us consider a three-dimensional manifold $X$ with boundary
$\partial X$. 
The AKSZ construction 
defines a topological sigma model on $\Map(T[1]X, T^*[2]T[1]M)$.
This model is a special case of the Courant sigma model
on an open manifold.
The P-structure becomes
\begin{eqnarray}
\bomega &=& \int_{\calX} 
d^{3}\sigma d^{3}\theta
\ 
(
\delta \bbx^{i} \wedge \delta \bbxi_{i} +
\delta \bp_i \wedge \delta \bq^{i}
).
\end{eqnarray}
The Q-structure BV action has the following form:
\begin{eqnarray}
S&=& \int_{\calX} d^{3}\sigma d^{3}\theta
\ \left(- \bbxi_{i} \bbd \bbx^{i}
+ \bbq^{i} \bbd \bbp_{i}
+ 
\bbxi_{i} \bq^{i} 
+ \frac{1}{3!}
H_{ijk}(\bbx) \bq^{i} \bq^{j} \bq^{k}
\right).
\label{WZPoissonBVaction}
\end{eqnarray}

We need to determine the boundary conditions to complete the theory.
Consistency with the variation principle restricts the possible boundary conditions.
The variation $\delta S$ is
\begin{eqnarray*}
\delta S =
\int_{\calX} d^{3}\sigma d^{3}\theta
\ \left(- \delta \bbxi_{i} \bbd \bbx^{i}
- \bbxi_{i} \bbd \delta \bbx^{i}
+ \delta \bbq^{i} \bbd \bbp_{i}
+ \bbq^{i} \bbd \delta \bbp_{i}
+ \cdots
\right).
\end{eqnarray*}
To derive the equations of motion, we use integration by parts for the terms 
$- \bbxi_{i} \bbd \delta \bbx^{i} + \bbq^{i} \bbd \delta \bbp_{i}$.
The boundary terms must vanish, i.e.,
\begin{eqnarray}
\delta S|_{\partial \calX} =
\int_{\partial \calX} d^{2}\sigma d^{2}\theta
\ \left(
- \bbxi_{i} \delta \bbx^{i}
- \bbq^{i} \delta \bbp_{i}
\right) = 0.
\label{boundaryintegration}
\end{eqnarray}
Any boundary condition must be consistent with 
equation (\ref{boundaryintegration}).

Two kinds of local boundary conditions are possible:
$\bbxi_{//i} = 0$ or $\delta \bbx_{//}^{i}=0$, and
$\bbq_{//}^{i} =0$ or $\delta \bbp_{//i}=0$,
where $//$ indicates the component that is parallel to the boundary.\footnote{Hybrids of these boundary conditions are also possible.
}
As an example, 
we take the boundary conditions 
$\bbxi_{//i} = 0$ and $\bbq_{//}^{i} =0$ on $\partial \calX$.
These boundary conditions can be written using the components 
of the superfields as follows:
$\xi^{(0)}_{i} = \xi^{(1)}_{0 i} 
= \xi^{(1)}_{1 i} = \xi^{(2)}_{01 i} = 0$ 
and $q^{(0)i} = q_{0}^{(1)i} 
= q_1^{(1)i} = q_{01}^{(2)i} =0$ on $\partial \calX$.

Another consistency condition is that the boundary conditions
must not break the classical master equation 
$\sbv{S}{S}=0$.
Direct computation using the BV action \eqref{WZPoissonBVaction} gives
\begin{equation}
\sbv{S}{S} = 
\int_{\partial \calX} \!\!\!\! d^{2}\sigma d^{2}\theta
\left(- \bbxi_{i} \bbd \bbx^{i}
+ \bbq^{i} \bbd \bbp_{i}
+ 
\bbxi_{i} \bq^{i} 
+ \frac{1}{3!}
H_{ijk}(\bbx) \bq^{i} \bq^{j} \bq^{k}
\right).
\label{boundarymasterequation}
\end{equation}
The boundary conditions 
$\bbxi_{//i} = 0$ and $\bbq_{//}^{i} =0$ are consistent with the classical master equation.
The kinetic terms on the right-hand side in equation (\ref{boundarymasterequation}) vanish
 on the boundary:
\begin{eqnarray}
\int_{\partial \calX} d^{2}\sigma d^{2}\theta \
\bvartheta = 
\int_{\partial \calX} d^{2}\sigma d^{2}\theta
\ \left(- \bbxi_{i} \bbd \bbx^{i}
+ \bbq^{i} \bbd \bbp_{i}
\right) = 0.
\label{kinticboundary}
\end{eqnarray}
The interaction terms in equation (\ref{boundarymasterequation}) also vanish:
\begin{eqnarray}
\int_{\partial \calX} d^{2}\sigma d^{2}\theta \
\bTheta = 
\int_{\partial \calX} d^{2}\sigma d^{2}\theta
\ \left(
\bbxi_{i} \bq^{i} 
+ \frac{1}{3!}
H_{ijk}(\bbx) \bq^{i} \bq^{j} \bq^{k}\right)=0.
\label{interactionboundary}
\end{eqnarray}
It is accidental that the second condition does not impose a new condition.
Generally, we have more conditions on the boundary, such as in the next example.

The consistency of the boundary conditions is described in the language of the target QP-manifold $\calM$.
Equation (\ref{kinticboundary}) is satisfied if 
$\xi_i= q^i = 0$.
From equation (\ref{WZPpstructure}),
this is satisfied if the image of a boundary is in
a Lagrangian subspace of the P-structure $\omega$.
Equation (\ref{interactionboundary}) is satisfied if
$\bTheta|_{\partial \calX}=0$, that is, 
the Q-structure vanishes ($\Theta = 0$) on the Lagrangian subspace.

Note that there exists an ambiguity in the total derivatives of $S_0$,
and this comes from the ambiguity in the expression for the local coordinates of $\vartheta$.
Here, we choose an $S_0$ such that the classical master equation is satisfied 
if we take $\bTheta|_{\partial \calX}=0$.
For example, if we use the boundary condition $\xi_i = p_i =0$, then
we should take $S_0 = \int_{\calX} d^{3}\sigma d^{3}\theta
\ \left(- \bbxi_{i} \bbd \bbx^{i}
+ \bbp_{i} \bbd \bbq^{i}
\right)$.

We can change the boundary condition by introducing consistent boundary terms.
For the present example, 
the boundary terms must be pullbacks of a 
degree two function $\alpha$ by the transgression map,
$\mu_* \ev^* \alpha$. 
As an example, we take
$
\alpha =
\frac{1}{2}
f^{ij}(x)
p_{i} p_{j}$ \cite{Park2000au} and find consistency conditions for $H_{ijk}(x)$ and $f^{ij}(x)$.\footnote{Equation (\ref{3DCSMboundary}) is just one example of a boundary term; we can consider 
more general boundary terms, such as 
$$- \int_{\partial \calX} 
d^{2}\sigma d^{2}\theta \
\left(\bp_i \bbd \bbx^i + 
\frac{1}{2} f^{ij}(\bbx) \bp_{i} \bp_{j} 
+ g^{i}{}_j(\bbx) \bp_{i} \bq^{j}
+ \frac{1}{2} h_{ij}(\bbx) \bq^{i} \bq^{j}
\right).$$} The modified action is given by
\begin{eqnarray}
S&=& \int_{\calX} d^{3}\sigma d^{3}\theta
\ \left(- \bbxi_{i} \bbd \bbx^{i}
+ \bbq^{i} \bbd \bbp_{i}
+ 
\bbxi_{i} \bq^{i} 
+ \frac{1}{3!}
H_{ijk}(\bbx) \bq^{i} \bq^{j} \bq^{k}
\right)
\nonumber \\
&& - 
\int_{\partial \calX} 
d^{2}\sigma d^{2}\theta \
\frac{1}{2}
f^{ij}(\bbx)
\bp_{i} \bp_{j}.
\label{3DCSMboundary}
\end{eqnarray}
In order to derive the equations of motion
from the variation of $\delta S$, the following boundary 
integral must vanish:
\begin{eqnarray*}
\delta S|_{\partial \calX}
=
\int_{\partial \calX} d^{2}\sigma d^{2}\theta
\ \left[\left(- \bbxi_{i} 
- \frac{1}{2} \frac{\partial f^{jk}(\bbx)}{\partial \bbx^i}
\bp_{j} \bp_{k} \right)
\delta \bbx^{i}
+ \left(- \bbq^{i} 
+ 
f^{ij}(\bbx) \bp_{j}
\right)
\delta \bbp_{i} 
\right].
\end{eqnarray*}
This determines the 
boundary conditions as
\begin{eqnarray}
\bbxi_{i}|_{//} =
- \frac{1}{2} \frac{\partial f^{jk}}{\partial \bbx^i}(\bbx) 
\bp_{j} \bp_{k}|_{//},
\qquad 
\bq^i|_{//} &=& f^{ij}(\bbx) \bp_{j} |_{//}.
\label{WZPoissonboundary}
\end{eqnarray}
In addition, we must also consider a boundary term in $\sbv{S}{S}$.
In this example, the classical master equation, $\sbv{S}{S}=0$,
requires
the integrand of $S_1$ to be zero on the boundary:\footnote{Equation \eqref{boundaryS1} is the same as equation \eqref{interactionboundary}. We can prove that this condition does not depend on the boundary conditions.}
\begin{eqnarray}
\left(\bbxi_{i} \bq^{i} 
+ \frac{1}{3!}
H_{ijk}(\bbx) \bq^{i} \bq^{j} \bq^{k}
\right) \bigg|_{//}
= 0.
\label{boundaryS1}
\end{eqnarray}
Equations (\ref{WZPoissonboundary}) and (\ref{boundaryS1}) 
show that the image of the boundary must satisfy the following conditions,
\begin{eqnarray}
&& \xi_{i} q^{i} 
+ \frac{1}{3!}
H_{ijk}(x) q^{i} q^{j} q^{k}
= 0,
\label{WZPS1}
\\
&& \xi_{i} = - \frac{1}{2} \frac{\partial f^{jk}}{\partial x^i}(x) 
p_{j} p_{k},
\label{WZPLagrangian1}
\\
&& q^i = f^{ij}(x) p_{j}.
\label{WZPLagrangian2}
\end{eqnarray}
This means that 
equation \eqref{WZPS1} is satisfied on the Lagrangian subspace 
$\calL_{\alpha}$ 
of a target QP-manifold $\calM$
defined by \eqref{WZPLagrangian1} and \eqref{WZPLagrangian2}.
By substituting equations  (\ref{WZPLagrangian1}) and (\ref{WZPLagrangian2}) into equation
(\ref{WZPS1}), we obtain the geometric structures on the image of the boundary 
$\partial X$,
\begin{eqnarray}
&& \xi_{i} q^{i} 
+ \frac{1}{3!}
H_{ijk}(x) q^{i} q^{j} q^{k}
\nonumber \\
&=& -
\frac{1}{2} \frac{\partial f^{jk}}{\partial x^l}(x) 
f^{li}(x) p_{j} p_{k} p_{i}
+ \frac{1}{3!}
H_{ijk}(x) f^{il}(x)
f^{jm}(x) f^{kn}(x) p_{l} p_{m} p_{n} 
\nonumber \\ 
&=& 0.
\label{twistedPoisson}
\end{eqnarray}
If we define a bivector field
$\pi = \frac{1}{2} f^{ij}(x) \partial_i \wedge \partial_j$,
then equation (\ref{twistedPoisson}) is equivalent to
\begin{eqnarray}
[\pi,\pi]_S = \wedge^3 \pi^{\#} H.
\label{twistedPoisson2}
\end{eqnarray}
Here, $[-,-]_S$ is the Schouten-Nijenhuis bracket on 
the space of multivector fields $\Gamma (\wedge^{\bullet} TM)$,
which is an odd Lie bracket on the exterior algebra
such that 
$\partial_i \wedge \partial_j = - \partial_j \wedge \partial_i$.
The operation $\pi^{\#}: T^*M \to TM $ is locally 
defined by 
$\frac{1}{2} f^{ij}(x) \partial_i \wedge \partial_j(dx^k) 
= f^{kj}(x) \partial_j$.
Equation \eqref{twistedPoisson2} is called a \textsl{twisted Poisson structure}
\cite{Severa:2001qm}.

The ghost number 0 part of the BV action, equation (\ref{3DCSMboundary}), 
becomes 
\begin{eqnarray}
S|_0&=& \int_{X} 
\ \left(- \xi^{(2)}_{i} \wedge d x^{i}
+ q^{(1)i} \wedge d p^{(1)}_{i}
+ 
\xi^{(2)}_{i} \wedge q^{(1)i} 
+ \frac{1}{3!}
H_{ijk}(x) q^{(1)i} \wedge q^{(1)j} \wedge q^{(1)k}
\right)
\nonumber \\
&& 
- 
\int_{\partial X} 
\frac{1}{2}
f^{ij}(x)
p^{(1)}_{i} \wedge p^{(1)}_{j},
\label{3DCSMboundary2}
\end{eqnarray}
after integration with respect to $\theta^{\mu}$, where $x = x^{(0)}$.
Integrating out $\xi^{(2)}_i$, we obtain a
topological field theory in two dimensions with a Wess-Zumino term:
\begin{eqnarray*}
S|_0&=& \int_{\partial X} 
\ \left(- p^{(1)}_{i} \wedge d x^{i}
- 
\frac{1}{2}
f^{ij}(x)
p^{(1)}_{i} \wedge p^{(1)}_{j} 
\right)
+ 
\int_{X}
\frac{1}{3!}
H_{ijk}(x)
d x^{i} 
\wedge d x^{j} 
\wedge d x^{k}.
\end{eqnarray*}
This model is called
the WZ-Poisson sigma model or the twisted Poisson sigma model \cite{Klimcik:2001vg}.
The constraints are first class if and only if 
the target space manifold has a twisted Poisson structure.


\subsection{General Structures of AKSZ Sigma Models with Boundary}
\noindent
In the previous subsection, a typical example 
for  boundary structures of AKSZ sigma models was presented.
In this subsection, we discuss the general theory 
in $n+1$ dimensions.

Assume that $X$ is an ($n+1$)-dimensional
manifold with boundary, 
$\partial X \neq \emptyset$.
Let $\calM$ be a QP-manifold of degree $n$.
Then, by the AKSZ construction,
a topological sigma model on $\Map(T[1]X, \calM)$
can be constructed.
The boundary conditions on $\partial X$
must be consistent with the QP-structure.

First, let us take a Q-structure function
$S= S_0 + S_1 = \iota_{\hat{D}} \mu_* \ev^* \vartheta +\mu_* \ev^* \Theta$
without boundary terms.
Then, $\sbv{S}{S}$ yields the integrated boundary terms,
\begin{eqnarray}
\sbv{S}{S}= 
\iota_{\hat{D}} \mu_{\partial \calX *} \
(i_{\partial} \times {\rm id})^* \ 
\ev^* \vartheta 
+ \mu_{\partial \calX *} \
(i_{\partial} \times {\rm id})^* \ \ev^* \Theta,
\label{boundaryCME}
\end{eqnarray}
where
$\mu_{\partial \calX} $ is the boundary measure
induced from $\mu$ on $\partial \calX$ by
the inclusion map
$i_{\partial}: \partial \calX \longrightarrow \calX$.
The map $(i_{\partial} \times {\rm id})^*: 
\Omega^{\bullet}(\calX \times \calM)
\longrightarrow \Omega^{\bullet}(\partial \calX \times \calM)
$ is the restriction of the bulk graded differential forms 
on the mapping space to the boundary $\partial \calX$.
In order to satisfy the master equation,
the right-hand side of equation (\ref{boundaryCME}) must vanish.
Thus we obtain the following theorem,
\begin{theorem}\label{boundarytheta}
Assume that $\partial \calX \neq \emptyset$. 
$\sbv{S}{S}=0$ requires 
$
\iota_{\hat{D}} \mu_{\partial \calX *} \
(i_{\partial} \times {\rm id})^* \ 
\ev^* \vartheta 
+ \mu_{\partial \calX *} \
(i_{\partial} \times {\rm id})^* \ \ev^* \Theta =0$.
\end{theorem}
If we consider the consistency with the variational principle
of a field theory,
the two terms must vanish independently.
We explain this using the local coordinate expression.

The kinetic term in the AKSZ sigma model is
\begin{eqnarray}
S_0 = 
\int_{\calX} 
d^{n+1} \sigma d^{n+1} \theta 
\ 
\sum_{0 \leq i \leq \floor{n/2}} 
(-1)^{n+1-i} \bbp_{a(i)} \bbd \bbq^{a(i)}.
\end{eqnarray}
In order to derive the equations of motion, we take the variation.
We find that the boundary integration of the variation of the total action,
should vanish for consistency:
\begin{eqnarray}
\delta S|_{\partial \calX} = 
\int_{\partial \calX} 
d^{n} \sigma d^{n} \theta 
\ 
\sum_{0 \leq i \leq \floor{n/2}} 
(-1)^{n+1-i} \bbp_{a(i)} \delta \bbq^{a(i)}=0.
\end{eqnarray}
This imposes the boundary conditions $\bbp_{a(i)} =0$ 
or $\delta \bbq^{a(i)} =0$ on $
\partial X$.
This implies 
that the image of the boundary lies
in a Lagrangian submanifold $\calL \subset \calM$,
which is the zero locus of $\vartheta$, 
$\vartheta |_{\mathcal L} =0$,
on the target space.
Under this condition, the first term
in equation \eqref{boundaryCME}, $\iota_{\hat{D}} \mu_{\partial \calX *} \
(i_{\partial} \times {\rm id})^* \ 
\ev^* \vartheta$, vanishes.
Therefore, Theorem \ref{boundarytheta} reduces to a simpler form, that is, 
the condition that the second term vanishes.
This can be reinterpreted as a condition on $\Theta$ on the target space.
\begin{proposition}\label{boundarytheta2}
Let $\calL$ be a Lagrangian 
submanifold of $\calM$, i.e.,
$\vartheta|_{\calL} =0$.
Then $\sbv{S}{S}=0$ is satisfied if 
$\Theta |_{\calL} =0$.
\cite{Hofman:2002rv}
\end{proposition}

\subsection{Canonical Transformation of Q-structure Function}
\noindent
In the remainder of this section,
we discuss the general theory of boundary terms.
Let us define an exponential adjoint operation $e^{\delta_\alpha}$
on a general QP-manifold $\calM$,
\begin{eqnarray}
e^{\delta_\alpha} \Theta
= \Theta + \sbv{\Theta}{\alpha} 
+ \frac{1}{2} \sbv{\sbv{\Theta}{\alpha}}{\alpha}
+ \cdots,
\end{eqnarray}
where $\alpha \in C^{\infty}(\calM)$.

\begin{definition}
Let $(\calM, \omega, \Theta)$ be a QP-manifold of degree $n$,
$\alpha \in C^{\infty}(\calM)$
be a function of degree $n$, then,
$e^{\delta_\alpha}$ is called a 
{\rm twist} by $\alpha$.
\end{definition}
This transformation preserves degree, since $\alpha$ is of degree $n$.
Note that a twist 
satisfies $\sbv{e^{\delta_\alpha} f}{e^{\delta_\alpha} g}
= e^{\delta_\alpha} \sbv{f}{g}$ for any function 
$f, g \in C^{\infty}(\calM)$, therefore, the twist by $\alpha$ is a canonical transformation.

Now we consider a canonical transformation of a QP-manifold $(\calM, \omega, \Theta)$ by a twist $e^{\delta_{\alpha}}$.
Since the Q-structure function $\Theta$ changes
to $e^{\delta_\alpha} \Theta$,
the Q-structure function in the corresponding AKSZ sigma model
is changed to
\begin{eqnarray}
S &=& S_0 + S_1
\nonumber \\
&=& \iota_{\hat{D}} \mu_* {\rm ev}^* \vartheta
+ \mu_* \ev^* 
e^{\delta_\alpha} \Theta.
\label{akszwithboundary}
\end{eqnarray}
If $\partial X = \emptyset$, the consistency condition
of the theory is not changed,
since a canonical transformation preserves the graded Poisson bracket
and the classical master equation.
However, if $\partial X \neq \emptyset$, 
the twist changes the boundary conditions.
Applying Proposition \ref{boundarytheta}
to equation (\ref{akszwithboundary}), we obtain 
the following conditions on $\alpha$ for the consistent boundary conditions 
of the AKSZ sigma models.
\begin{proposition}\label{boundaryQstr}
Assume $\partial \calX \neq \emptyset$.
Let $(\calM, \omega, \Theta)$ be a QP-manifold of degree $n$,
$\calL$ be a Lagrangian submanifold of $\calM$,
which is the zero locus of $\vartheta$, and
$\alpha \in C^{\infty}(\calM)$ be a function of degree $n$.
If the twist generated by $\alpha$ vanishes on $\calL$, 
$e^{\delta_\alpha} \Theta |_{\calL} =0$,
then the Q-structure function (\ref{akszwithboundary}) satisfies 
the classical master equation $\sbv{S}{S}=0$.
\cite{Hofman:2002rv}
\end{proposition}
A function $\alpha$ with the property defined in Proposition 
\ref{boundaryQstr} is called 
a Poisson function
\cite{Terashima, Kosmann-Schwarzbach:2007}
or a canonical function \cite{Ikeda:2013wh}.
The structures  for general $n$ have been analyzed in Ref.~\citenum{Ikeda:2013wh}.
\subsection{From Twist to Boundary Terms}
\noindent
In this subsection, we show that a canonical function $\alpha$, 
defined 
in the previous section, generates a boundary term.
Let $I = \mu_* \ev^* \alpha$ be a functional constructed by a transgression
of $\alpha$.
In equation (\ref{akszwithboundary}),
the change in the Q-structure 
by the twist
is converted into the change in the P-structure 
by the following inverse canonical transformation
on the mapping space,
\begin{eqnarray}
S^{\prime} &=& e^{- \delta_{I}} S
\nonumber \\
&=& 
e^{- \delta_{I}} S_0
+ \mu_* \ev^* 
e^{- \delta_{\alpha}} e^{\delta_\alpha} \Theta
\nonumber \\
&=& 
e^{- \delta_{I}} S_0
+ \mu_* \ev^* \Theta.
\label{ChangeS}
\end{eqnarray}
This QP-structure $(\bomega^{\prime}= - d(e^{- \delta_{I}} S_0), 
S^{\prime})$ is equivalent to
the original QP-structure $(\bomega, S)$.
\cite{Hofman:2002rv}

For a physical interpretation of $\alpha$, we consider 
the simple special case in which 
$\alpha$ satisfies $\sbv{\alpha}{\alpha}=0$, and thus $\sbv{I}{I}=0$. 
Then, since $e^{- \delta_{I}} S_0 = 
S_0 - \sbv{S_0}{I}$, the BV action becomes
\begin{eqnarray}
S^{\prime} 
&=& 
S_0 
- 
\sbv{S_0}{I}
+ \mu_* \, \ev^* \Theta.
\label{Swithboundaryalpha}
\end{eqnarray}
The second term, $- \sbv{S_0}{I}$, is nothing but a boundary term:
\begin{eqnarray}
- \sbv{S_0}{I}
&=& 
- \left\{S_0, \int_{\calX} \mu \ \ev^* \alpha \right\}
\nonumber \\ 
&=& \int_{\calX} d^{n+1}\sigma d^{n+1}\theta \, \bbd \ev^* \alpha
= \int_{\partial \calX} 
d^{n+1}\sigma d^{n+1}\theta \, \ev^* \alpha.
\nonumber
\end{eqnarray}
Therefore, a canonical transformation by a twist induces
a boundary term generated by the $\alpha$ in the BV action $S$.
The boundary term generally carries a nonzero charge.
In physics, this charge can be identified  with the 
number of $n$-branes, and 
the above action \eqref{Swithboundaryalpha} defines
a so-called \textsl{topological open $n$-brane theory}.
This structure has been applied to the analysis of 
T-duality geometry. \cite{Bessho:2015tkk}
If $\sbv{\alpha}{\alpha}\neq 0$, we cannot make a simple interpretation 
as local boundary terms, but it still gives a consistent deformation of an AKSZ sigma model.
As a special case of this construction,
the Nambu-Poisson structures are realized by the AKSZ 
sigma models on a manifold with boundary.
\cite{Bouwknegt:2011vn}

In this section, 
we have discussed Dirichlet-like fixed boundary conditions.
We can also impose Neumann-like free boundary conditions.
The AKSZ sigma models with free boundary conditions
 are called the AKSZ-BFV theories on a manifold with boundary, 
and they have been analyzed in Ref.~\citenum{Cattaneo:2012qu, Cattaneo:2012zs}.

\section{Topological Strings from AKSZ Sigma Models}
\noindent
In this section, we discuss derivations of 
the A- and B-models \cite{Witten:1991zz} 
from the AKSZ sigma models in two dimensions, 
which is equivalent to the Poisson sigma model.
The A- and B-models are derived by gauge fixing of this
AKSZ sigma model. \cite{Alexandrov:1995kv}

\subsection{A-Model}
\noindent
Let the worldsheet $X = \Sigma$ be a compact Riemann surface and 
the target space $M$ be a K\"ahler manifold.
Let us consider the AKSZ formalism of the Poisson sigma model in 
Example \ref{AKSZPSM}.
Here, we take the theory where
$S_0=0$ 
in the Q-structure 
BV action 
\eqref{2DPSMBVaction}, i.e.,
\begin{eqnarray}
S= S_1&=& 
\int_{T[1]\Sigma} 
d^{2}\sigma d^{2}\theta \
f^{ij}(\bbx)
\bbxi_{i} \bbxi_{j}.
\label{AmodelPSM}
\end{eqnarray}
Here, we take the normalization of $S_1$ in Ref.~\citenum{Alexandrov:1995kv}.
The classical master equation, $\sbv{S_1}{S_1}=0$,
is satisfied if $f^{ij}(x)$ satisfies
equation (\ref{PoissonJacobi}) as in the case of 
the Poisson sigma model, i.e., if $M$ is a Poisson manifold.
This condition is satisfied on a K\"ahler manifold $M$, 
by taking $f^{ij}$ as the inverse of the K\"ahler form.
As in Example \ref{AKSZPSM}, the superfields $(\bbx^i, \bbxi_i)$ 
of degree $(0, 1)$ can be identified 
with $(\bphi^i, \ba_i)$ in Section 
\ref{superfieldformalismPSM}.
The superfields are expanded in the supercoordinate $\theta^{\mu}$,
\begin{eqnarray}
\bbx^i &=& 
\bphi^i = \phi^i + A^{+i} + c^{+i}
(= x^{(0)i} +  x^{(1)i} +  x^{(2)i}),
\nonumber \\
\bbxi_i &=& \ba_i 
=  - c_i + A_i + \phi^+_i
(= \xi^{(0)}_i +  \xi^{(1)}_i +  \xi^{(2)}_i).
\nonumber
\end{eqnarray}

We take the complex coordinates $(z, \bar{z})$ on the worldsheet $\Sigma$ and on the target space $M$
with holomorphic and antiholomorphic indices $i =(a, \dot{a})$.
Let $J^i{}_j$ be a complex structure and
$g_{ij}$ be a K\"ahler metric.
Then, the inverse of the K\"ahler form $f^{ij}$ is
expressed as $f^{ij} = - J^i{}_k g^{kj}$.
We decompose the holomorphic and antiholomorphic parts 
of the fields with respect to the worldsheet complex structure.
$A_{z}^{+i} = - {A_{0}^{+i} + i A_{1}^{+i}}$
and $A_{\bar{z}}^{+i} = {A_{0}^{+i} + i A_{1}^{+i}}$,
$A_{z i} = - A_{1i} - i A_{0i}$
and $A_{\bar{z} i} = A_{1i} - i A_{0i}$,
$\phi_{z \bar{z}}{}^+_{i}  = 2 i \phi^*_i$
and 
$c_{z \bar{z}}{}^{+i}  = 2 i c^{*i}$.
%
The BV antibrackets are
\begin{eqnarray}
&& 
\sbv{A_{zi}}{A_{\bar{z}^{\prime}}^{+j}} 
= 2 \delta^j_i \delta(z - z^{\prime})\delta(\bar{z} - \bar{z}^{\prime}),
\quad
\sbv{A_{\bar{z}i}}{A_{z^{\prime}}^{+j}} 
= 2 \delta^j_i \delta(z - z^{\prime})\delta(\bar{z} - \bar{z}^{\prime}),
\nonumber \\
&& 
\sbv{\phi^i}{\phi_{z \bar{z}}{}^+_{j}}
= 2i \delta^i{}_j \delta(z - z^{\prime})\delta(\bar{z} - \bar{z}^{\prime}),
\quad
\sbv{c_i}{c_{z \bar{z}}{}^{+j}}
= 2i \delta_i{}^j \delta(z - z^{\prime})\delta(\bar{z} - \bar{z}^{\prime}),
\nonumber 
\end{eqnarray}
and all other antibrackets are zero.
Taking linear combinations of the fields,
we obtain the complex fields with respect to the target complex structure.
For example, for $A_{z i}$, $A_{\bar{z} i}$, $A_{z}^{+i}$ and 
$A_{\bar{z}}^{+i}$,
we take linear combinations such that 
$\overline{A_{z}^{+a}} = A_{\bar{z}}^{+\dot{a}}$,
$\overline{A_{z a}} = A_{\bar{z}\dot{a}}$.
Their BV brackets are
\begin{eqnarray}
&& 
\sbv{A_{za}}{A_{\bar{z}^{\prime}}^{+b}} 
= \delta^b_a \delta(z - z^{\prime})\delta(\bar{z} - \bar{z}^{\prime}),
\quad
\sbv{A_{\bar{z}a}}{A_{z^{\prime}}^{+b}} 
= \delta^b_a \delta(z - z^{\prime})\delta(\bar{z} - \bar{z}^{\prime}),
\nonumber
\end{eqnarray}
and their complex conjugates.

If the gauge symmetry of the theory is partially fixed
by the BV gauge fixing procedure,
the action reduces to the A-model action given in Ref.~\citenum{Witten:1991zz}.
We fix $c^{+i}$, $A_{\bar{z} a}$, $A_{z \dot{a}}$ and $\phi^+_{i}$
by taking the following gauge fixing fermion
\begin{eqnarray}
&& \Psi = \int_{T[1]\Sigma} d^2 z 
g_{a\dot{a}}(\phi)
(A_{z}^{+\dot{a}} \partial_{\bar{z}} \phi^a
- A_{\bar{z}}^{+a} \partial_z \phi^{\dot{a}}).
\nonumber
\end{eqnarray}
We obtain the gauge fixing conditions,
\begin{eqnarray}
&& c^{+i} = 0,
\nonumber \\
&& 
A^z{}_{\dot{a}} = i g_{a\dot{a}}(\phi) \partial_{\bar{z}} \phi^a,
\nonumber \\
&& 
A^{\bar{z}}{}_{a} = - i g_{a\dot{a}}(\phi)
\partial_{z} \phi^{\dot{a}},
\nonumber \\
&& \phi_{z \bar{z}}{}^+_{a} 
= - i \partial_{\bar{z}} 
(g_{a\dot{a}}(\phi) 
A_z^{+\dot{a}} ),
\nonumber \\
&& \phi_{z \bar{z}}{}^+_{\dot{a}} 
= i \partial_{z} 
(g_{a\dot{a}}(\phi) 
A_{\bar{z}}^{+{a}}).
\label{gaugefixingA}
\end{eqnarray}
Substituting equations (\ref{gaugefixingA}) into equation (\ref{AmodelPSM})
and integrating out $A_{za}$ and $A_{\bar{z}\dot{a}}$,
we obtain the original A-model action,
\begin{eqnarray*}
S_1&=& \int_{\Sigma} 
d^{2} z 
\left(
g_{a\dot{a}}
\partial_{\bar{z}} \phi^a
\partial_{z} \phi^{\dot{a}}
- i \psi^{a}_{\bar{z}} D_{{z}} \chi_{a}
- i \psi^{\dot{a}}_z D_{\bar{z}} \chi_{\dot{a}}
+ R_{a\dot{a}}{}^{b\dot{b}} 
\psi^{a}_{\bar{z}} \psi^{\dot{a}}_z \chi_b \chi_{\dot{b}}
\right),
\end{eqnarray*}
where
\begin{eqnarray*}
\chi_i &=& \frac{1}{2i} c_i, \quad
\psi_{\mu}^a = A_{\mu}^{+a},
\quad
\psi_{\mu}^{\dot{a}} = A_{\mu}^{+\dot{a}}, 
\nonumber 
\end{eqnarray*}
and
\begin{eqnarray*}
D_{{z}} \chi_i &=& \partial_{{z}} \chi_i
- \Gamma^k_{ij} \partial_{{z}} \phi^j \chi_k,
\nonumber \\
D_{\bar{z}} \chi_i &=& \partial_{\bar{z}} \chi_i
- \Gamma^k_{ij} \partial_{\bar{z}} \phi^j \chi_k,
\end{eqnarray*}
and $\Gamma^k_{ij}$ is the Christoffel symbol on the target space.

\subsection{B-Model}
\noindent
We start from Example \ref{BModel}, the QP-manifold realization of a complex structure on a smooth manifold $M$,
and take local coordinates on the target space such that 
$J^i{}_{j} = \left(
\begin{matrix}
0 & 1 & \cr
- 1 & 0 & \cr
\end{matrix}
\!\right) 
= \epsilon^{ik} \delta_{kj}
$. 
Then, the BV action \eqref{bmodelAKSZ} 
is simplified to
\begin{eqnarray}
S_B &=& \int_{\calX}
d^2 z d^2 \theta 
\left(
\bbxi_{i} \bbd \bbx^i - \bbp_i \bbd \bbq^i 
+ \epsilon^i{}_j \bbxi_{i} \bbq^j
\right).
\label{BmodelPSM}
\end{eqnarray}
The superfields can be expanded in $\theta^{\mu}$ as
\begin{eqnarray}
\bbx^i &=& x^{(0)i} + x^{(1)i} + x^{(2)i},
\nonumber \\
\bbxi_i &=&  \xi^{(0)}_i + \xi^{(1)}_i + \xi^{(2)}_i,
\nonumber \\
\bbq^i &=& q^{(0)i} + q^{(1)i} + q^{(2)i},
\nonumber \\
\bbp_i &=& p^{(0)}_i + p^{(1)}_i + p^{(2)}_i.
\nonumber
\end{eqnarray}
We consider partial gauge fixing, as in the A-model.
Different gauge fixing conditions 
for the holomorphic and antiholomorphic parts are imposed
as follows,
\begin{eqnarray}
&& x^{(1)\dot{a}} = 0,
\nonumber \\
&& x_{z \bar{z}}^{(2)a} + \Gamma^a_{bc} x_z^{(1)b} x_{\bar{z}}^{(1)c} = 0,
\nonumber \\
&& \xi^{(0)}_{\dot{a}} = 0, 
\nonumber \\
&& \xi^{(1)}_{z a} + \Gamma^b_{ac} \xi_{b}^{(0)} x_{z}^{(1)c}
= g_{a\dot{a}}(\phi) \partial_z x^{(0)\dot{a}},
\nonumber \\
&& \xi^{(1)}_{\bar{z}a} 
- \Gamma^b_{ac} \xi_{b}^{(0)} x_{\bar{z}}^{(1)c}
= g_{a\dot{a}}(\phi) \partial_{\bar{z}} x^{(0)\dot{a}},
\nonumber \\
&& \xi^{(2)}_{z\bar{z}a} = 0,
\quad
\xi^{(2)}_{z\bar{z}\dot{a}} 
- R^c_{a\dot{a}b} x_{\bar{z}}^{(1)a} x_{{z}}^{(1)b}
\xi_{c}^{(0)}
= - (D_z x_{\bar{z}}^{(1)a} + D_{\bar{z}} x_z^{(1)a})
g_{a\dot{a}},
\nonumber \\
&& q^{(0)\dot{a}} = q^{(1)\dot{a}} = q^{(2)\dot{a}} = 0,
\nonumber \\
&& p^{(0)}_a = p^{(1)}_a = p^{(2)}_a = 0.
\label{gaugefixingB}
\end{eqnarray}
Substituting equations (\ref{gaugefixingB}) into equation (\ref{BmodelPSM}),
we obtain the original B-model action,
\begin{eqnarray*}
S&=& \int_{\Sigma} 
d^{2} z 
\left(
g_{ij}
\partial_{z} \phi^i
\partial_{\bar{z}} \phi^j
+ i \eta^{\dot{a}}_z (D_{z} \rho_{\bar{z}}^a
+ D_{\bar{z}} \rho_z^a) g_{a\dot{a}}
+ i \theta_{a} (
D_{\bar{z}} \rho_z^a
- D_{z} \rho_{\bar{z}}^a ) 
\right.
\nonumber \\ &&
\left.
- R_{a\dot{a}b\dot{b}} 
\rho^{a}_{{z}} \rho^{{b}}_{\bar{z}} \eta^{\dot{a}} \theta_c
g^{c\dot{b}}\right),
\end{eqnarray*}
where 
$\phi^i = x^{(0)i}$,
$\rho^a = x^{(1)a}$,
$\theta_a = \xi^{(0)}_{a}$
and
$\eta^{\dot{a}} = g^{a\dot{a}} p^{(0)}_{a}$.

\section{Quantization}
\noindent
We discuss the quantization of the AKSZ sigma models
in two dimensions as an important 
example. The quantization is carried out
by the usual procedure of the BV formalism.
Quantization in general dimensions is not well understood, yet.

\subsection{Poisson Sigma Model on a Disc}
\noindent
The path integral quantization of the Poisson sigma model on a disc 
yields the Kontsevich deformation quantization formula
on a Poisson manifold. \cite{Cattaneo:1999fm}
We briefly explain this model as an example of the quantization of 
an AKSZ sigma model.
For details, we refer to Ref.~\citenum{Cattaneo:1999fm}.

\subsubsection{Deformation Quantization}
\noindent
Recall that a Poisson manifold is a manifold $M$ with a Poisson bracket
$\{-, -\}_{PB}$. 
\begin{definition}\label{defquant}[deformation quantization]
Let $M$ be a Poisson manifold
and $C^{\infty}(M)[[\hbar]]$ be 
a set of formal power series on $C^{\infty}(M)$,
where $\hbar$ is a formal parameter.
A deformation quantization is a product (star product) $*$ on 
$C^{\infty}(M)[[\hbar]]$ satisfying the following conditions:
\begin{enumerate}
\item[$(1)$]
For $F, G \in C^{\infty}(M)[[\hbar]]$,
$F * G = \sum_k \left(\frac{i \hbar}{2}\right)^k {\cal B}_k(F, G)$
is bilinear, where
${\cal B}_k$ is a bidifferential operator such that
${\cal B}_0$ is a product, ${\cal B}_0(F, G)= FG$, and 
${\cal B}_1$ is a Poisson bracket, ${\cal B}_1(F, G)= \{F, G\}_{PB}$.

\item[$(2)$]
For $F,G,H \in C^{\infty}(M)[[\hbar]]$, 
$*$ is associative, i.e., 
$$
(F*G)*H = F*(G*H).
$$

\item[$(3)$]
Two star products $*$ and $*^{\prime}$ corresponding to the same Poisson 
bracket are equivalent if they coincide
by the following linear transformation:
$F^{\prime} = R F = \sum_k \left(\frac{i \hbar}{2}\right)^k {\cal D}_k(F)$,
where ${\cal D}_k$ is a differential operator.
i.e.
\begin{eqnarray}
F *^{\prime} G(x) &=& R^{-1} (RF * RG).
\nonumber
\end{eqnarray}
\end{enumerate}

\end{definition}
We review the following theorem proved in Ref.
\citenum{Cattaneo:1999fm}.
\begin{theorem}
The correlation functions of the Poisson sigma model 
of observables on the boundary of a disc
coincide with the star product formula
on a Poisson manifold,
called the Kontsevich formula.
i.e.
\begin{eqnarray}
F * G(x) &=&
\left\langle F(\bphi(1)) G(\bphi(0)) \right\rangle 
= \int_{\bphi(\infty) = x} \calD \Phi\
F(\bphi(1)) G(\bphi(0)) e^{\frac{i}{\hbar}S_q}.
\nonumber
\end{eqnarray}
\end{theorem}

\subsubsection{Path Integrals}
\noindent
Let us consider the disc $D = \{z \in \bC | |z| \leq 1 \}$.
Since the Poisson sigma model is invariant under conformal 
transformations,
we map the disc to the upper half-plane 
$\Sigma = 
\{z = \sigma^0 + i \sigma^1| \sigma^1 \geq 0\}$ 
by a conformal transformation.
Then, we consider 
\begin{eqnarray}
S = S_0 + S_1 = \int_{T[1] \Sigma} 
d^2 \sigma d^2 \theta
\left(
\ba_i \bbd \bphi^i
+ 
\frac{1}{2} f^{ij}(\bphi) \ba_i \ba_j
\right),
\label{qPSM}
\end{eqnarray}
where $\bphi^i = \bbx^i
$ and
$\ba_i = \bbxi_i
$.

The partition function 
$Z = \int_L \calD \Phi \ e^{\frac{i}{\hbar}S_{q}}$
and correlation functions
are calculated by a formal perturbative expansion in $\hbar$
in the path integral,
$$
Z( \calO_1 \cdots \calO_r )
= \int_L \calD \Phi \ \calO_1 \cdots \calO_r e^{\frac{i}{\hbar}S_{q}}
= \sum_{k=0}^{\infty} \hbar^k
Z_k( \calO_1 \cdots \calO_r ).
$$
Here, $S_{q}$ is the gauge fixed quantum action and the
$\calO_s$ are observables.

Since a complete superfield formalism is not known for gauge fixed
actions of AKSZ theories, we expand it in the component fields.
The superfields are expanded in $\theta^{\mu}$ as follows,
\begin{eqnarray}
\bphi^i &=& \phi^i + A^{+i} + c^{+i},
\nonumber \\
\ba_i &=& - c_i + A_i + \phi^+_i.
\end{eqnarray}
%

\subsubsection{BV Quantization}\label{BVQuantizations}
\noindent
In general, the gauge symmetry algebra of an AKSZ sigma model
is an open algebra. Thus, we apply the BV quantization procedure 
\cite{Henneaux:1989jq, Gomis:1994he}.
We consider the gauge fixing of the action $S$.

First, we introduce an
FP antighost $\bar{c}^i$ of ghost number 
$\gh \ \bar{c}^i = -1$, 
a Nakanishi-Lautrup multiplier field $b^i$
of $\gh \ b^i = 0$
and their antifields $\bar{c}_i^{+} 
= \frac{1}{2} \theta^{\mu} \theta^{\nu} \bar{c}_{\mu \nu i}^{+}$ 
of $\gh \ \bar{c}_i^{+} = 0$ and
$b_i^{+} = \frac{1}{2} \theta^{\mu} \theta^{\nu} b_{\mu \nu i}^{+}
$ of $\gh \ b_i^{+} = -1$.
Then, the P-structure (antibracket) is extended as
\begin{eqnarray}
\sbv{\bar{c}^i}{\bar{c}_j^{+}}=\sbv{b^i}{b_j^{+}} = \delta^i{}_j,
\end{eqnarray}
and the other antibrackets are zero.

The following gauge fixing term is added to the classical 
BV action $S$,
\begin{eqnarray}
S_{GF} = - \int_{T[1]\Sigma} d^2 \sigma d^2 \theta \ b^i \bar{c}_i^{+},
\label{gfterm}
\end{eqnarray}
and we denote $S_q = S + S_{GF}$.

Next, the gauge fixing fermion 
$\Psi(\Phi)$ 
of ghost number one is determined
such that it restricts the path integral to the subspace of the gauge fixed fields and ghosts.
We take the gauge fixing fermion as
\begin{eqnarray}
\Psi = \int_{T[1]\Sigma} d^2 \sigma d^2 \theta \ \bar{c}^i \bbd * A_i,
\nonumber
\end{eqnarray}
where $*$ is the Hodge star on $\Sigma$.
The BV gauge fixing is carried out by 
imposing the following equation,
\begin{eqnarray}
\Phi^+ = \frac{\delta \Psi}{\delta \Phi}.
\nonumber
\end{eqnarray}
All the antifields are fixed by this gauge fixing condition.
%
In components, we obtain
\begin{eqnarray}
&& \bar{c}^+_i = \bbd * A_i,
\quad 
{A}^{+i} = * \bbd \bar{c}^i,
\nonumber \\
&& \phi^+_i = 0,
\quad 
{c}_i^{+} = b_i^{+} =0.
\label{PSMgaugefixing}
\end{eqnarray}
Substituting equations (\ref{PSMgaugefixing}) into the BV action 
$S_q$,
we obtain the gauge fixed quantum BV action,
$S_{q|fix}(\Phi) = 
S_q \left(\Phi, \Phi^+ = \frac{\delta \Psi}{\delta \Phi}\right)$:
\begin{eqnarray}
S_{q|fix} &=& \int_{T[1] \Sigma} 
d^2 \sigma d^2 \theta
\biggl(
A_i \bbd \phi^i
- * \bbd \bar{c}^i \bbd c_i
- b^i \bbd * A_i
+ \frac{1}{2} f^{ij}(\phi) A_i A_j
\nonumber \\
&& 
- \frac{\partial f^{ij}}{\partial \phi^k}(\phi) 
* \bbd \bar{c}^k A_i c_j
+ 
\frac{1}{4} 
\frac{\partial^2 f^{ij}}{\partial \phi^k \partial \phi^l}(\phi) 
* \bbd \bar{c}^k * \bbd \bar{c}^l c_i c_j
\biggr).
\label{gaugefixedaction}
\end{eqnarray}

The partition function $Z$ must be independent of 
the gauge fixing conditions.
This means that the partition function is invariant under
arbitrary infinitesimal changes of the gauge fixing fermion $\Psi$,
$$
Z(\Psi) = Z(\Psi + \delta \Psi).
$$
This requirement gives the following consistency condition for the quantum BV
action $S_q= S+ S_{GF}$,
\begin{eqnarray}
\Delta e^{\frac{i}{\hbar}S_q(\Phi, \Phi^+)} =0,
\label{laplacianofs}
\end{eqnarray}
where $\Delta$ 
is the odd Laplace operator (\ref{oddlaplace2}) introduced in Section \ref{AKSZTFTgeneral}.
This equation is equivalent to
the \textsl{quantum master equation},
\begin{eqnarray}
2 i \hbar \Delta S_q - \sbv{S_q}{S_q} = 0.
\label{quantummaster}
\end{eqnarray}
We can prove that the AKSZ sigma models formally satisfy this equation.
More precisely, the AKSZ sigma models satisfy 
$\Delta S_q = 0$ and $\sbv{S_q}{S_q} =0$.
Since these equations contain divergences in general,
we need to renormalize in order to prove these equations beyond 
the formal expressions.
As we discuss later, we can properly regularize the equation in the Poisson sigma model.

The correlation function of an observable $\calO$,
\begin{eqnarray}
\cfun{\calO}
= \int_{
\Phi^+ = \frac{\delta \Psi}{\delta \Phi}, 
} \calD \Phi\
\calO e^{\frac{i}{\hbar}S_{q}},
\nonumber
\end{eqnarray}
must also be invariant 
under infinitesimal changes of the gauge fixing fermion $\Psi$.
This condition is equivalent to
\begin{eqnarray}
\Delta \left( \calO e^{\frac{i}{\hbar}S_{q}} \right)=0
\label{observabledelta}
\end{eqnarray}
and can be rewritten as
\begin{eqnarray}
i \hbar \Delta \calO - \sbv{S_q}{\calO} = 0.
\label{observabledef}
\end{eqnarray}

\subsubsection{Boundary Conditions}
\noindent
Here, we determine the boundary conditions of the classical theory, using the same procedure as explained in Section \ref{AKSZboundary}.

The boundary conditions on each field are determined by two consistency conditions.
The variation of the action is
\begin{align}
\delta S &= \int_{T[1] \Sigma} 
d^2 \sigma d^2 \theta
\left(
\delta \ba_i \bbd \bphi^i
+ \ba_i \bbd \delta \bphi^i
+ 
\delta \bphi^i
\frac{1}{2} \frac{\partial f^{jk}}{\partial \bphi^i}(\bphi) \ba_j \ba_k
+ f^{ij}(\bphi) \delta \ba_i  \ba_j
\right).
\label{variationofaction}
\end{align}
In order to obtain the equations of motion,
we need to integrate the second term $\ba_i \bbd \delta \bphi^i$ by parts.
Its boundary integral must vanish. 
Thus, we obtain
\begin{eqnarray}
\int_{T[1] \Sigma} 
\!\!\!\!\!
d^2 \sigma d^2 \theta \
\bbd (\ba_i \delta \bphi^i) 
&=&
\int_{\partial T[1] \Sigma} 
\!\!\!\!\!
d \sigma^0 d \theta^0
\ba_i \delta \bphi^i
\nonumber \\
&=&
\int_{\partial T[1] \Sigma} 
\!\!\!\!\!
d \sigma^0 d \theta^0
\left(A_i \delta \phi^i - c_i \delta A^{+i}\right)
=0.
\label{boudnarycondition}
\end{eqnarray}
The possible boundary conditions
that satisfy equation (\ref{boudnarycondition}) are
$A_{//i}|=0$ or $\delta \phi^i|=0$,
and $c_i|=0$ or $\delta A^{+i}_{//}|=0$.\footnote{More general boundary conditions have been analyzed in Ref.~\citenum{Cattaneo:2003dp}.}\, 
Here, the notation $A_{//i}= A_{0i}$ means the 
component parallel to the boundary and
$\Phi|$ denotes the value of $\Phi$ on the boundary.
In order to obtain a nontrivial solution for the embedding map 
from $\Sigma$ to $M$, $\phi^i$, 
we take $A_{//i}|=0$ and $c_i|=0$.

The classical equations of motion are
\begin{eqnarray}
&& \bbd \bphi^i + f^{ij}(\bphi) \ba_j =0,
\nonumber \\
&& \bbd \ba_i 
+ \frac{1}{2} \frac{\partial f^{jk}}{\partial \bphi^i}(\bphi) \ba_j \ba_k 
= 0.
\nonumber
\end{eqnarray}
From the equations of motion and the boundary conditions,
$A_{//i}|=0$ and $c_i|=0$,
we obtain the boundary conditions 
$\phi^i|=$ constant and $A^{+i}_{//}|=0$.
Therefore, the boundary conditions for all fields are
\begin{eqnarray}
&& \phi^i|=x^i=\mbox{constant},
\quad 
A_{//i}|=0,
\nonumber \\
&& c_i|=0,
\quad 
A^{+i}_{//}|=0.
\label{classicalboundary}
\end{eqnarray}
Here  $x^i$ parametrize the boundary.

Next, we determine the boundary conditions for 
other extra fields.
The consistency conditions for the equations of motion 
of the gauge fixed action (\ref{gaugefixedaction})
fix the boundary conditions for the ghost $b^i=0$.
The boundary conditions for the other ghosts and antifields are determined by
consistency with the gauge fixing conditions of 
equation (\ref{PSMgaugefixing}) as 
\begin{eqnarray}
&& \phi^+_i| = 0,
\quad 
{c}_i^{+}| = b_i^{+}| =0,
\nonumber \\
&& \bar{c}^+_i| = \bbd * A_i|, 
\quad
\bar{c}^i| =\mbox{constant}.
\nonumber
\end{eqnarray}
These boundary conditions are consistent 
with the master equation.

\subsubsection{Propagators}
\noindent
The propagators are defined by the first three terms
of the gauge fixed action (\ref{gaugefixedaction}),
\begin{eqnarray}
S_F 
&=& \int_{T[1] \Sigma} 
d^2 \sigma d^2 \theta
\biggl(
A_i \bbd \phi^i
- b^i \bbd * A_i
- * \bbd \bar{c}^i \bbd c_i
\biggr)
\nonumber \\
&=& \int_{T[1] \Sigma} 
d^2 \sigma d^2 \theta
\biggl(
A_i \bbd \phi^i
+ A_i * \bbd b^i 
- c_i \bbd * \bbd \bar{c}^i
\biggr).
\end{eqnarray}
If we introduce the gauge fixed superfields,
\begin{eqnarray}
\bvarphi^i &=& \varphi^i + * \bbd \bar{c}^i + 0,
\nonumber \\
\ba_i &=& - c_i + A_i + 0,
\label{gaugefixedsuperfield}
\end{eqnarray}
the propagators of each component field are combined 
to a superfield propagator,
where $\bvarphi$ is defined by $\bphi^i = x^i + \bvarphi^i$.

Let $\bbd_z$ and $\bbd_w$ be superderivatives with respect to the variables $z$ and $w$.
Let $G(z, w)$ be a Green's function such that
$\bbd_w * \bbd_w G(z, w) = 2 \pi \delta(z-w)$,
where $G(z, w)$ is determined by the Dirichlet boundary condition for $z$ and the Neumann boundary condition for $w$, respectively.
The solution is
$G(z, w) = \frac{1}{2i} \ln 
\frac{(z-w)(z-{\bar w})}{({\bar z}-{\bar w})({\bar z}-{w})}$.
Using this Green's function, the superpropagator of
$(\varphi^i, \bbd \bar{c}^i, c_i, A_i)$ is determined as
\begin{eqnarray}
\cfun{\bvarphi^i(w) \ba_j(z)}
= \frac{i\hbar}{2 \pi} \delta^i{}_j (\bbd_z + \bbd_w) G(z, w),
\label{superpropagator}
\end{eqnarray}
which is consistent with the boundary conditions for each field.
In addition to equation (\ref{superpropagator}),
there is the propagator of $A_i$ and $b^i$, which we 
omit, since the star product does not involve the propagator $\cfun{A_i(w) b^j(z)}$.

\subsubsection{Vertices}
\noindent
The last three terms of the gauge fixed action (\ref{gaugefixedaction})
are interaction terms denoted by $S_I$, and they define the vertices. 
$S_I$ is 
simplified using 
gauge fixed
superfields (\ref{gaugefixedsuperfield})
as follows,
\begin{eqnarray}
S_I &=& \int_{T[1] \Sigma} 
d^2 \sigma d^2 \theta
\biggl(
\frac{1}{2} f^{ij}(\phi) A_i A_j
- 
\frac{\partial f^{ij}}{\partial \phi^k}(\phi) 
* \bbd \bar{c}^k A_i c_j
\nonumber \\ &&
+ 
\frac{1}{4} 
\frac{\partial^2 f^{ij}}{\partial \phi^k \partial \phi^l}(\phi) 
* \bbd \bar{c}^k * \bbd \bar{c}^l c_i c_j
\biggr)
\nonumber \\
&=& \int_{T[1]\Sigma} \!\!\!\! 
d^2 \sigma d^2 \theta \
\frac{1}{2} f^{ij}(\bphi) \ba_i \ba_j.
\end{eqnarray}
In order to identify the vertices, 
$\bphi^i$ and $\ba_i$ are
expanded around the classical 
solutions $\bphi^i = x^i + \bvarphi^i$
and $\ba_i = 0 + \ba_i$. 
Taylor expansion of $f^{ij}(\bphi)$ gives 
\begin{equation}
S_I =
\frac{1}{2} \int_{T[1]\Sigma} \!\!\!\! 
d^2 \sigma d^2 \theta \
\sum_{k=0}^{\infty}
\frac{1}{k!} 
\partial_{l_1} \partial_{l_2} \cdots 
\partial_{l_k} 
f^{ij}(x) 
\bvarphi^{l_1} {\bvarphi}^{l_2} \cdots 
{\bvarphi}^{l_k} 
\ba_i \ba_j,
\label{taylorexpansion}
\end{equation}
which determines the vertices of order $\hbar^{-1}$.
Note that there is an infinite number of vertices. 
From equation (\ref{taylorexpansion}),
the $k$-th vertex has two $\ba$ lines and 
$k$ ${\bvarphi}$ lines
that have the weight $\frac{1}{2} \frac{1}{k!} 
\partial_{l_1} \partial_{l_2} \cdots 
\partial_{l_k} f^{ij}(x) $.

The path integral of an observable $\calO$ can be expanded as
\begin{eqnarray}
\cfun{\calO}
&=&
\int \calO e^{\frac{i}{\hbar}(S_F+S_I)}
= 
\sum_{n=0}^{\infty}
\frac{i^n}{\hbar^n n!}
\int \calO e^{\frac{i}{\hbar}S_F}
S_I^n.
\end{eqnarray}
Since $\calO$ 
is a function of superfields $\bvarphi$ and $\ba$,
it is computed by Wick's theorem
using the propagators $\cfun{\bvarphi^i(w) \ba_j(z)}$,
as in usual perturbation theory.

\subsubsection{Renormalization of Tadpoles}\label{tadpolesection}
\noindent
Contributions from tadpoles are renormalized 
to zero in order to derive a star product.
Although this renormalization is different from the
one usually used in quantum field theory,
it can be carried out consistently with the quantum master equation.
We can add a gauge invariant counter term
that subtracts all tadpole contributions,
\begin{eqnarray}
S_{ct} &=& \int_{T[1]\Sigma} \!\!\!\! d^2 \theta d^2 \sigma \
\frac{\partial f^{ij}(\bphi)}{\partial \bphi^i} \ba_j \kappa,
\nonumber
\end{eqnarray}
where $\kappa$ is the subtraction coefficient of the renormalization.

\subsubsection{Correlation Functions of Observables on the Boundary}
\noindent
An arbitrary function of $\bphi$, $F(\bphi)$,
restricted to the boundary of $\Sigma$,
is an  observable
since it satisfies equation \eqref{observabledef}.
We now compute the correlation functions of these observables
(often called vertex operators).
They satisfy the first condition in 
the definition of a deformation quantization, Definition \ref{defquant}.

We consider an observable $\calO = F(\bphi(t))G(\bphi(s))$ which depends on two points,
where $t$ and $s$ are coordinates 
on the boundary $\partial \Sigma$
and $F$ and $G$ are arbitrary functions of $\bphi$.
The conformal transformation of the disc worldsheet fixes the three points
$0,1,\infty$ on the boundary circle $S^1$.
The boundary condition of $\bphi$ is fixed at $\sigma^0 = \infty$ as 
$\bphi^i(\infty) = x^i$,
and $\calO$ can be transformed to $\calO = F(\bphi(1))G(\bphi(0))$
by conformal transformation.

We compute the correlation function $\cfun{F(\bphi(1))G(\bphi(0))}$ 
by the Feynman rules.
The order $\hbar^n$ amplitudes consist of $n$ vertices and $2n$ propagators.
We choose $n+2$ points on $\Sigma$.
There are two points $z=u_L=0$ and $z=u_R=1$ on the boundary
where two vertex operators $F(\bphi(1))$ and $G(\bphi(0))$ are inserted.
Other $n$ points are located in the interior of $\Sigma$.
These points are denoted by $u_j \in \Sigma$, $(j= 1, 2, \cdots, n, L, R)$,
where $u_j$ for $j= 1, 2, \cdots, n$ are the points of $n$ vertices.
A propagator $d G(z,w)$ connects two points chosen from the above $n+2$ points.
We introduce a map $v_a:\{1, 2, \cdots, n\}
\rightarrow \{1, 2, \cdots, n, L, R\}$, where $a = 1, 2$, 
and $d G(u_j, u_{v_a(j)})$ denotes the propagator from $u_j$ to $u_{v_a(j)}$,
where $j=1, 2, \cdots, n$, since two vertex operators on the boundary are 
functions of $\bphi$.
$v_a(j) \neq j$ for all $j$,
since we renormalize the tadpole graphs to zero as in Section 
\ref{tadpolesection}.
Since all the vertices contain precisely two $\ba_i$'s,
the weight of the nonzero Feynman diagram is obtained as
$$
\frac{1}{n!}\left(\frac{(i \hbar)^n}{(2 \pi)^{2n}}\right)
\int \wedge_{j=1}^{n}
\bbd G(u_j, u_{v_1(j)}) \wedge
\bbd G(u_j, u_{v_2(j)}),
$$
where $\bbd = \bbd_z + \bbd_w$.
This gives coefficients of
the $\hbar^n$ term of the star product
$(-1)^n \calB_{\Gamma n}(F, G)$ 
induced from the Feynman diagram $\Gamma$.

The first two terms of the perturbative expansion
are 
\begin{eqnarray}
\left\langle F(\bphi(1)) G(\bphi(0)) \right\rangle 
&=& \int_{\bphi(\infty) = x} \calD \Phi\
F(\bphi(1)) G(\bphi(0)) e^{\frac{i}{\hbar}S_q}
\nonumber \\
&=& F(x) G(x)
+ \frac{i \hbar}{2} f^{ij}(x) \frac{\partial F(x)}{\partial x^i} 
 \frac{\partial G(x)}{\partial x^j} + O(\hbar^2)
\nonumber \\
&=& F(x) G(x)
+ \frac{i \hbar}{2} \sbv{F(x)}{G(x)}_{PB} + O(\hbar^2),
\label{correlationfunction}
\end{eqnarray}
where the first term is the solution of the classical equations of motion
and the second term is the Poisson bracket of $F$ and $G$.
This correlation function satisfies the first condition in 
Definition \ref{defquant}.

Higher-order terms are determined by the Feynman diagrams.
From equation (\ref{correlationfunction}),
the Poisson sigma model has been determined only by
the Poisson structure on $M$, and thus
higher-order terms in the expansion are expressed by 
$f^{ij}$ and its derivatives.

If $f^{ij}(x)
$ is a constant, the perturbation is simplified 
at all orders.
In this case, \eqref{taylorexpansion} has one vertex 
without derivatives of $f$,
$\frac{1}{2} f^{ij}(x) 
\ba_i \ba_j
$.
Therefore, we obtain
\begin{eqnarray}
\left\langle F(\bphi(1)) G(\bphi(0)) \right\rangle 
&=& \int_{\bphi(\infty) = x} \calD \Phi\
F(\bphi(1)) G(\bphi(0)) e^{\frac{i}{\hbar}S_q}
\nonumber \\
&=& 
\lim_{y \rightarrow x} 
\exp{\left(\frac{i \hbar}{2} f^{ij}
 \frac{\partial }{\partial x^i} 
 \frac{\partial }{\partial y^j}\right)} 
F(x) G(y).
\nonumber 
\end{eqnarray}
This is nothing but the Moyal product, which is the star product derived 
from the constant antisymmetric tensor $f^{ij}$.

\subsubsection{Associativity and Equivalence}
\noindent
In this section, we explain how the correlation function
(\ref{correlationfunction})
satisfies Condition (2)
of Definition \ref{defquant}, i.e., 
the associativity condition.

The associativity condition is derived from the Ward-Takahashi identity 
of the gauge symmetry of this theory.
In the BV formalism, 
the Ward-Takahashi identity is derived from the quantum master equation (\ref{laplacianofs})
and its path integral,
\begin{eqnarray}
\int_{\phi(\infty) = x} \calD \Phi\
\Delta \left(\calO
e^{\frac{i}{\hbar}S_q}
\right) = 0.
\label{WTidentity}
\end{eqnarray}
Take an observable 
$\calO = F(\bphi(1)) G(\bphi(t)) H(\bphi(0))$
on the boundary,
where $t$ is a coordinate on the boundary such that $0 < t < 1$, and
let $\tau$ be a supercoordinate partner 
 of $t$.
Since the conformal transformation in two dimensions 
fixes only three points, this observable has the modulus $t$.
Substituting this observable into equation (\ref{WTidentity}), we get
\begin{eqnarray}
\int_{\phi(\infty) = x, 1 > t > 0} dt d\tau \calD \Phi\
\Delta \left(
F(\bphi(1)) G(\bphi(t)) H(\bphi(0)) 
e^{\frac{i}{\hbar}S_q}
\right) = 0.
\nonumber
\end{eqnarray}
From equations (\ref{laplacianofs}) and (\ref{WTidentity}),
we obtain 
\begin{eqnarray}
\int_{\phi(\infty) = x, 1 > t > 0} dt d\tau
\calD \Phi\
\sbv{S_q}{F(\bphi(1)) G(\bphi(t)) H(\bphi(0))}
e^{\frac{i}{\hbar}S_q}
= 0.
\nonumber
\end{eqnarray}
Substituting 
$$\sbv{S_q}{F(\bphi(1)) G(\bphi(t)) H(\bphi(0))}
= - \bbd \left(F(\phi(1)) G(\phi(t)) H(\phi(0))\right),$$
and applying Stokes' theorem,
this path integral becomes a boundary integral on the moduli space,
\begin{eqnarray}
&& \lim_{t \rightarrow 1} 
\int_{\phi(\infty) = x} \calD \Phi\
\left(
F(\bphi(1)) G(\bphi(t)) H(\bphi(0)) 
e^{\frac{i}{\hbar}S_q}
\right) 
\nonumber \\
&& -
\lim_{t \rightarrow 0} 
\int_{\phi(\infty) = x} \calD \Phi\
\left(
F(\bphi(1)) G(\bphi(t)) H(\bphi(0)) 
e^{\frac{i}{\hbar}S_q}
\right) = 0.
\label{boundaryidentititesofassociativity}
\end{eqnarray}
This equation leads to the associativity relation
$$
(F*G)*H - F*(G*H) =0,
$$
for $F,G,H \in C^{\infty}(M)[[\hbar]]$.

Next, we discuss Condition (3) in Definition \ref{defquant}.
It is sufficient to
prove the following statement:
Let $F(x)$ be a function such that 
$\sbv{F(x)}{G(x)}_{PB} = 0$ for any $G$. Then, $F*G(x)$
is equivalent to the normal product $F(x) G(x)$ by a redefinition
$F^{\prime} = RF$.\footnote{Note that if 
$\sbv{F}{G}_{PB} =0$, then
$F*G(x)= F(x) G(x)$ is a trivial solution of the deformation quantization.}

If $\sbv{F(x)}{-}_{PB} = 0$, $F(\bphi(u))G(\bphi(0))$ is an observable,
where $u$ is an interior point on the disc. Thus, the correlation function
\begin{eqnarray}
\cfun{F(\bphi(u))G(\bphi(0))}
= \int_{\phi(\infty) = x} 
\calD \Phi\
F(\bphi(u))G(\bphi(0))
e^{\frac{i}{\hbar}S_q}
\label{conditions3WTeq}
\end{eqnarray}
satisfies the following Ward-Takahashi identity,
\begin{eqnarray}
\int_{\phi(\infty) = x} 
\calD \Phi\
\Delta \left(F(\bphi(u))G(\bphi(0))
e^{\frac{i}{\hbar}S_q}
\right) = 0.
\label{conditions3WTeq}
\end{eqnarray}
From equation \eqref{conditions3WTeq} and
a similar computation to the
derivation of 
\eqref{boundaryidentititesofassociativity}
using $\sbv{S}{F(\bphi(u))} = \bbd F(\bphi(u))$, 
we obtain
\begin{eqnarray}
\int_{\phi(\infty) = x} 
\calD \Phi\
\bbd F(\bphi(u))G(\bphi(0))
e^{\frac{i}{\hbar}S_q} = 0.
\label{conditions3WTeq2}
\end{eqnarray}
This means that the correlation function $\cfun{F(\bphi(u))G(\bphi(0))}$
is independent of $u$.

For $G=0$, we obtain the one-point function,
\begin{eqnarray}
\cfun{F(\bphi(u))}
= \int_{\phi(\infty) = x} 
\calD \Phi\
F(\bphi(u))
e^{\frac{i}{\hbar}S_q}
= F(x) + O(\hbar^2),
\nonumber
\end{eqnarray}
which is expressed by a formal series of derivatives of $F(x)$ as
$\sum_k \left(\frac{i \hbar}{2}\right)^k {\cal D}_k(F)$.
Then, we can take $RF(x) = \cfun{F(\bphi(u))}$.

We can prove that
\begin{eqnarray}
RF * G(x) = \cfun{F(\bphi(1))G(\bphi(0))} &=& \lim_{\epsilon \rightarrow +0}
\cfun{F(\bphi(1 + i \epsilon))G(\bphi(0))},
\end{eqnarray}
by the factorization property of the path integral.
This shows that 
$RF * G(x)$ is equivalent to $F(x) G(x)$.

\subsection{Formality}
\noindent
The mathematical proof 
of the existence of a deformation quantization 
on a Poisson manifold 
\cite{Kontsevich:1997vb, Cattaneo:1999fm}
is called the formality theorem,
and it is closely related to the quantization of the 
Poisson sigma model.
In this article, we discuss the correspondence between mathematical terms
and physical concepts appearing in the AKSZ sigma model.

\subsubsection{Differential Graded Lie Algebras}

The input data
of the deformation quantization is 
a Poisson bracket $\{F, G\}_{PB}$.
As we saw in Example \ref{pmanifold},
the Poisson structure can be interpreted in terms of supergeometry.
Thus, a deformation quantization is also reformulated in terms of 
supergeometry or graded algebras.
First, we introduce a differential graded Lie algebra.
\begin{definition}
A \textsl{differential graded Lie algebra} (dg Lie algebra)
$(\mathfrak{g}, \sbv{-}{-}, d)$ is 
a graded algebra with $\mathbb{Z}$-degree 
$\mathfrak{g} = \oplus_{k \in \mathbb{Z}} \mathfrak{g}^k[-k]$,
where $\mathfrak{g}^k$ is the degree $k$ part of $\mathfrak{g}$.
$\sbv{-}{-}: \mathfrak{g}^k \times \mathfrak{g}^l 
\rightarrow \mathfrak{g}^{k+l}$ is a graded Lie bracket
and $d: \mathfrak{g}^k \longrightarrow \mathfrak{g}^{k+1}$ 
is a differential of degree 1 such that $d^2=0$.
\end{definition}

\subsubsection{Maurer-Cartan Equations of Poisson Bivector Fields}

We consider a QP-manifold of degree 1, $(\calM, \omega, \Theta)$.
The graded Poisson bracket, $\sbv{-}{-}$, induced by the P-structure
is identified with the graded Lie bracket of the dg Lie algebra,
where the degree is shifted by 1.
The corresponding differential is $d=0$.
The space of functions of degree 2 in $C^{\infty}(\calM)$ is identified 
with $\mathfrak{g}^1$, which is isomorphic to the space of the bivector fields, 
$\alpha_1 = \frac{1}{2} \alpha^{ij}(x) \partial_i \wedge \partial_j$.
Then, the space $(\mathfrak{g}^1 = \Gamma(\wedge^2 TM), \sbv{-}{-}, d=0)$
is a dg Lie algebra and denoted by $\mathfrak{g}^1_1 = T^1_{poly}(M)$.


Next, we consider the subspace of the solutions of 
the Maurer-Cartan equation 
$d \alpha_1 + \frac{1}{2} \sbv{\alpha_1}{\alpha_1} =0$ 
in $\mathfrak{g}^1_1$.
This space is denoted by $\calMC(\mathfrak{g}^1_1)=
\mathfrak{g}^1_1/_{\sim}$.
It is equivalent to the solutions of the classical 
master equation $\sbv{\Theta}{\Theta}=0$
since $d=0$ and $\Theta$ is of degree 2 and can be identified 
with a bivector field.
Therefore, 
the QP-manifold of degree 1, $(\calM, \omega, \Theta)$, is identified with
$\calMC(\mathfrak{g}^1_1)$.

\subsubsection{Hochschild Complex of Polydifferential Operators}

The $\hbar^1$-th order of the deformation quantization
corresponds to the classical theory in physics.
The Poisson bivector $\alpha_1 = \frac{1}{2} f^{ij}(x) \partial_i \wedge \partial_j$ determines 
first two terms of the star product
as
$(F, G) \mapsto
\calB_0(F,G) + \frac{i \hbar}{2}
\calB_1(F,G) 
= FG + \frac{i \hbar}{2}
\frac{1}{2} f^{ij}(x) \partial_i F \partial_j G
\in \Hom(A^{\otimes 2}, A) $,
where $F, G \in A=C^{\infty}(M)$.

From Condition (3) in Definition \ref{defquant},
the two expressions of
$\calB_0(F,G) + \frac{i \hbar}{2} \calB_1(F,G)$
and 
$\calB_0(F^{\prime},G) + \frac{i \hbar}{2} \calB_1(F^{\prime},G)$
are equivalent in $\hbar^1$-th order, if they coincide 
after $F$ is redefined as 
$F^{\prime} = F + \frac{i \hbar}{2} \calD F$.
The redefinition map is an element of $\Hom(A, A)$.

In order to prove associativity, 
we must consider
a map $C(F, G, H)$ in $\Hom(A^{\otimes 3}, A)$.
The following associativity relation is obtained at
classical level, i.e., at $\hbar^1$-th order,
\begin{eqnarray}
C: (F, G, H) 
&\mapsto& 
C(F, G, H) 
\nonumber \\
&& = (FG)H - F(GH) 
\nonumber \\
&&
+ \frac{i \hbar}{2}
\left(\calB_1(FG, H)
- \calB_1(F, GH)
+ \calB_1(F,G)H - F\calB_1(G, H)
\right)
\nonumber \\
&&
+ \left(\frac{i \hbar}{2}\right)^2
\left(\calB_1(\calB_1(F,G), H) - \calB_1(F, \calB_1(G, H))\right).
\end{eqnarray}
The classical associativity holds,
if 
\begin{eqnarray}
C(F, G, H)=0.
\label{dfassociative}
\end{eqnarray}

To formulate associativity for all orders in $\hbar$,
we define a second dg Lie algebra in $\Hom(A^{\otimes k+1}, A)$.
Let $\mathfrak{g}_2 = \oplus_{k \in \mathbb{Z}, k \geq -1} 
\mathfrak{g}_2^k[-k]$,
where $\mathfrak{g}_2^k = \Hom(A^{\otimes k+1}, A)$.
For an element $C \in \mathfrak{g}_2^k$,
a differential $d$ and a graded Lie bracket $[-,-]$ are defined 
in such a way 
that 
equation \eqref{dfassociative} is obtained
as a part of the Maurer-Cartan equation.
The differential is defined as
\begin{eqnarray}
(d C)(F_0 \otimes \cdots \otimes F_{k+1})
&=& F_0 C(F_1 \otimes \cdots \otimes F_{k+1})
- \sum_{r=0}^k C(F_0 \otimes \cdots \otimes (F_r F_{r+1}) 
\otimes \cdots \otimes F_{k+1}) 
\nonumber \\
&&
+ (-1)^k C(F_0 \otimes \cdots \otimes F_k)F_{k+1}.
\end{eqnarray}
The graded Lie bracket is defined as
\begin{eqnarray}
[C_1, C_2]
&=& C_1 \circ C_2 - (-1)^{k_1 k_2} C_2 \circ C_1,
\\
C_1 \circ C_2 (F_0 \otimes \cdots \otimes F_{k_1 + k_2})
&=& \sum_{r=0}^{k_1} (-1)^{rk_2} C_1(F_0 \otimes \cdots \otimes F_{r-1} 
\otimes C_2(F_r \otimes \cdots \otimes F_{r+k_2} ) 
\nonumber \\ &&
\qquad \otimes 
F_{r+ k_2+ 1} \otimes \cdots \otimes F_{k_1 + k_2}),
\nonumber
\end{eqnarray}
where $C_1 \in \mathfrak{g}_2^{k_1}$ and $C_2 \in \mathfrak{g}_2^{k_2}$. 
Note that $(\mathfrak{g}_2, d)$ 
is called the Hochschild complex of polydifferential operators, 
and is also denoted as 
$\mathfrak{g}_2^k = D^k_{poly}(M)$ and $\mathfrak{g}_2 = D_{poly}(M)$.
The bracket $[-,-]$ is called the Gerstenhaber bracket. 

For an element $\widetilde{\alpha} 
\in \mathfrak{g}_2^{1}$ of degree 1,
the Maurer-Cartan equation 
$d \widetilde{\alpha} 
+ \frac{1}{2}[\widetilde{\alpha}, \widetilde{\alpha}]=0$
is equivalent to 
the associativity equation (\ref{dfassociative}).
Equivalence under redefinition, Condition (3),
is also expressed by the Maurer-Cartan equation 
in elements on $\mathfrak{g}_2^{0}$.
Therefore, a solution of the Maurer-Cartan equation in $\mathfrak{g}_2$
gives the star product at order $\hbar^1$.
The space of solutions of the Maurer-Cartan equation is 
denoted by $\calMC(\mathfrak{g}_2)= \mathfrak{g}_2/_{\sim}$.

\subsubsection{Morphisms of Two Differential Graded  Lie Algebras}

At classical level, i.e., at $\hbar^1$-th order, we define a map 
$U_1:\mathfrak{g}_1^1 \longrightarrow \mathfrak{g}_2^1$,
such that $U_1:\frac{1}{2} f^{ij}(x) \partial_i \wedge \partial_j
\mapsto
\left(F_0 \otimes F_1 \mapsto
\frac{1}{2} f^{ij}(x) \partial_i F_0 \partial_j F_1
\right)$.
Since this map preserves the Maurer-Cartan equations, 
this induces the map
$U_1:\calMC(\mathfrak{g}_1^1) \longrightarrow \calMC(\mathfrak{g}_2^1)$.

A deformation quantization is expressed as follows.
Fix the map $U_1$.
The problem is to find a morphism on $\hbar$ deformations of two
dg Lie algebras,
$U: \calMC(\mathfrak{g}_1^1[[\hbar]]) 
\longrightarrow \calMC(\mathfrak{g}_2[[\hbar]])$.

In general, the Maurer-Cartan equation 
on $\calMC(\mathfrak{g}_2[[\hbar]])$
is not preserved by a linear deformation of $U_1$,
since $U_1$ does not preserve graded Lie brackets.
To find $U$ consistent with the MC equations,
we extend the two dg Lie algebras to 
 $L_\infty$-algebras.
Then, we construct 
the map $U$ as an $L_{\infty}$-morphism between them.

We extend $\mathfrak{g}^1_1$ to 
the space of polyvector fields $T_{poly}(M)
= \mathfrak{g}_1 = \oplus_{k \in \mathbb{Z}, k \geq -1} 
\mathfrak{g}_1^k[-k]$,
where $\mathfrak{g}_1^k = \Gamma(\wedge^{k+1}TM)$.
An element of $\mathfrak{g}_1^k$ is a $k$-th multivector field
(an order $k$ antisymmetric tensor field), $\alpha_k =
\alpha^{j_0\cdots j_k}(x) \partial_{j_0} \wedge \cdots \wedge \partial_{j_k} 
= \alpha^{j_0\cdots j_k}(x) \xi_{j_0} \cdots \xi_{j_k}
\in \mathfrak{g}_1^k$.
The differential and the graded Lie bracket on $\mathfrak{g}^1_1$
are generalized to $\mathfrak{g}_1$ as follows.
The differential is kept trivial, $d=0$, 
and the graded Lie bracket is the Schouten-Nijenhuis bracket $[-,-]_S$
of multivector fields, i.e.,
the graded Poisson bracket $\sbv{-}{-}$ of degree 1 
of the QP-manifold $\calM = T^*[1]M$.

The map between two dg Lie algebras in the classical theory 
is defined as follows,
\begin{eqnarray}
U_1: && T_{poly}(M) \longrightarrow D_{poly}(M), 
\nonumber \\
&& \alpha_k
\mapsto
\left(F_0 \otimes \cdots \otimes F_{k} \mapsto
\frac{1}{(k+1)!}
\alpha^{j_0\cdots j_k}(x) \partial_{j_0} F_0 
\cdots \partial_{j_k} F_k
\right).~~~
\label{classicalmap}
\end{eqnarray}
Although this map is not isomorphic, 
$U_1$ induces an isomorphism between the
$d$-cohomologies of the two spaces, $T_{poly}(M)$ and $D_{poly}(M) $
\cite{Vey}.\footnote{
By definition, 
if their cohomologies are isomorphic, 
two spaces are called \textsl{quasi-isomorphic}.
The cohomology on $T_{poly}(M)$ 
is trivial because $d=0$.}

\subsubsection{$L_{\infty}$-Algebras and $L_{\infty}$-Morphisms}

A dg Lie algebra is embedded into the more general 
algebra, an $L_{\infty}$-algebra. In this section, we discuss
$L_{\infty}$-algebras and $L_{\infty}$-morphisms to describe
the statement of the formality theorem.

For a graded vector space $V = \oplus_{k \in \mathbb{Z}} V^k$,
we consider a graded commutative tensor algebra,
$T(V) = \oplus_{n=1}^{\infty} V^{\otimes n}$, which is a space of the sum of 
infinite tensor products.
On this space, a coassociative and cocommutative coproduct 
$\triangle$
is defined as
\begin{eqnarray}
\triangle (v_1, \cdots, v_n) 
&=& \sum_{\sigma \in \mathfrak{S}}
\sum_{k=1}^{n-1} \epsilon(\sigma) \frac{1}{k!(n-k)!}
(v_{\sigma(1)} \cdots v_{\sigma(k)}) \otimes
(v_{\sigma(k+1)} \cdots v_{\sigma(n)}),
\nonumber
\end{eqnarray}
where $v_k \in T(V)$.
Next, we assume the following multilinear maps of degree 1,
\begin{eqnarray}
\mathfrak{l}_k: && V^{\otimes k} \longrightarrow V,
\nonumber \\
&& (v_1 \otimes \cdots \otimes v_k) 
\mapsto
\mathfrak{l}_k(v_1 \cdots v_k),
\nonumber
\end{eqnarray}
and define 
a \textsl{codifferential} $Q = \sum_{k=1}^{\infty} Q_k$ as
\begin{eqnarray}
Q_k (v_1, \cdots, v_n) 
&=& \sum_{\sigma \in \mathfrak{S}}
\epsilon(\sigma) \frac{1}{k!(n-k)!}
\mathfrak{l}_k(v_{\sigma(1)} \cdots v_{\sigma(k)}) \otimes
v_{\sigma(k+1)} \otimes \cdots \otimes v_{\sigma(n)}.
\nonumber
\end{eqnarray}
\begin{definition}
A pair $(V, Q)$ 
is called an $L_{\infty}$-algebra 
(a strong
homotopy Lie algebra) if $Q^2=0$.
\cite{SchSta, LS}
\end{definition}
The first two operations in $\mathfrak{l}_k$ are a differential 
$\mathfrak{l}_1 = d$ and 
a superbracket $\mathfrak{l}_2(-,-) = \sbv{-}{-}$.
Moreover, a graded differential Lie algebra is embedded
by the identification,
$
\mathfrak{g}^{k-1}[1] 
\sim V^{k-1}$,
and $\mathfrak{l}_k = 0$, for $k \geq 3$.\footnote{A set of functions of a QP-manifold is regarded as an $L_{\infty}$-algebra, where degree of a function on the QP-manifold is equal to degree as an element of the $L_{\infty}$-algebra.}

We now define an $L_{\infty}$-morphism between two $L_{\infty}$-algebras.
\begin{definition}
A map between two $L_{\infty}$-algebras, 
$U: (V_1, Q) \longrightarrow (V_2, Q)$, 
is called a cohomomorphism if the map preserves degree
and satisfies
$\triangle \circ U = (U \otimes U) \circ \triangle$.
\end{definition}
\begin{definition}
A cohomomorphism $U$ between two $L_{\infty}$-algebras
is called an $L_{\infty}$-morphism
if $U Q = Q U$.
\end{definition}
Let us denote
$e^v = 1 + v + \frac{1}{2!} v \otimes v + 
\frac{1}{3!} v \otimes v \otimes v + \cdots$
and $\mathfrak{l}_*(e^v)= \mathfrak{l}_1(v)
+ \frac{1}{2!} \mathfrak{l}_2(v \otimes v) 
+ \frac{1}{3!} \mathfrak{l}_3(v \otimes v \otimes v) + \cdots$.
\begin{definition}
The Maurer-Cartan equation on an $L_{\infty}$-algebra $(V, Q)$
is $\mathfrak{l}_*(e^v)=0$.
\end{definition}
The Maurer-Cartan equation $\mathfrak{l}_*(e^v) =0$
is equivalent to 
$Q(e^v)= \mathfrak{l}_*(e^v) \otimes e^v =0$.
If an $L_{\infty}$-algebra is a dg Lie algebra,
then $Q(e^v)=0$ is equivalent to the ordinary Maurer-Cartan equation 
$d\alpha + \frac{1}{2}[\alpha, \alpha]=0$,
since $\mathfrak{l}_k=0$ for $k \geq 3$, where $v=\alpha$.

If we regard two dg Lie algebras $\mathfrak{g}_1$ and $\mathfrak{g}_2$ 
as $L_{\infty}$-algebras,
the nonlinear correspondence between the two Maurer-Cartan equations 
on the two dg Lie algebras becomes transparent.
Let $V_1 = \mathfrak{g}_1 = T_{poly}(M)[1]$ and
$V_2 = \mathfrak{g}_2 = D_{poly}(M)[1]$.
Then, the existence of 
a deformation quantization can be derived as the special case
with $\alpha_k=0$ except for $k=2$  if the following theorem is proved.
\begin{theorem}[formality theorem]\label{formalitytheorem}
\cite{Kontsevich:1995, Kontsevich:1997vb}
There exists an $L_{\infty}$-morphism from 
$(T_{poly}(M)[1], Q)$ to $(D_{poly}(M)[1], Q)$
such that $U_1$ is the map in equation (\ref{classicalmap}).
\end{theorem}
We refer to Ref.~\citenum{Kontsevich:1997vb} for the rigorous proof.
In this article, we observe that a two-dimensional AKSZ sigma model
contains all the structures required above to find the formality map.

\subsubsection{Correspondence to $n=1$ AKSZ Sigma Model}

The field theoretical realization of the Poisson structure is 
the Poisson sigma model, and
that of the Maurer-Cartan equations of a dg Lie algebra is 
the 
quantum BV master equations.
(The MC equation with $d=0$ corresponds to the classical master equation.)
The 
deformation of a dg Lie algebra in $\hbar$ corresponds to
the perturbative quantization of a physical theory.
The subalgebra $\calMC(\mathfrak{g})$ corresponds to 
the space of correlation functions which 
satisfy the Ward-Takahashi identities induced from 
the quantum master equation.

In order to generalize the Poisson sigma model to the $L_{\infty}$ setting,
we have to consider
the AKSZ sigma model where
the target space is generalized 
to the space of multivector fields, $\mathfrak{g}_1$.
The BV action of the AKSZ sigma model based on multivector fields is 
\begin{eqnarray}
S &=& S_0 + \sum_{p=0}^{d-1} S_{\alpha_p} 
\nonumber \\
&=& \int_{T[1] \Sigma} 
d^2 \sigma d^2 \theta
\left(
\ba_i \bbd \bphi^i
+ 
\sum_{p=0}^{d-1} \frac{1}{(p+1)!} \alpha^{j_0\cdots j_p}(\bphi) 
\ba_{j_0} \cdots \ba_{j_p}
\right),
\nonumber
\end{eqnarray}
where $\alpha_p
= \frac{1}{(p+1)!} 
\alpha^{j_0\cdots j_p}(x)
\frac{\partial}{\partial x^{j_0}} \wedge \cdots \wedge 
\frac{\partial}{\partial x^{j_p}} 
\in \Gamma (\wedge^{p+1}TM)$ is a multivector field
satisfying the MC equation in $\calMC(\mathfrak{g}_1)$.
We denote the term of the order $p$ multivector field by
$$S_{\alpha_p} = 
\int_{T[1] \Sigma} 
d^2 \sigma d^2 \theta
\frac{1}{(p+1)!} \alpha^{j_0\cdots j_p}(\bphi) 
\ba_{j_0} \cdots \ba_{j_p},
$$
and $\alpha^{j_0 j_1}(\bphi) = f^{j_0 j_1}(\bphi)$ corresponds to 
the original Poisson bivector field.
This action $S$ no longer has degree 0.
The MC equation on $\calMC(\mathfrak{g}_1)$
is equivalent to 
the classical master equation $\sbv{S}{S}=0$.

We take the same gauge fixing fermion and the same boundary conditions
as in the case of the Poisson sigma model in Section \ref{BVQuantizations}.
Observables are correlation functions of $m+1$ vertex operators 
on the boundary.
From the analysis of the moduli of insertion points
of the observables,
the observables on the boundary have the following form,
\begin{eqnarray}
\calO_x (F_0, \ldots, F_m) &=&
\int_{B_m} d^{m-1} t
\left[
F_0(\bphi(t_0, \theta_0))\cdots F_m(\bphi(t_m, \theta_m))
\right]^{(m-1)}
\delta_x (\phi(\infty)),
\nonumber
\end{eqnarray}
where $t_i$ are the points on the boundary circle such that 
$1 = t_0 > t_i > \cdots > t_{m-1} > t_m =0$.
$B_m$ is the space of the parameters $t_i$
and $[\cdots]^{(m-1)}$ denotes the order $\tau^{m-1}$-th term
which is given by the integration over supercoordinates $\tau_i$.
The map $U$ is given by
\begin{eqnarray}
U(\alpha)(F_0 \otimes \cdots \otimes F_m)(x)
&=& 
\int \calO_x (F_0, \ldots, F_m) e^{\frac{i}{\hbar}S_q},
\nonumber
\end{eqnarray}
where the path integral includes the $t_i$ integration over
the moduli space $B_m$.
We can obtain the $L_{\infty}$-morphism 
$U(\alpha) = \sum_{n=1}^{\infty} \frac{1}{n!} U_n(\alpha_1, \cdots, \alpha_n)$,
where $U_n: \mathfrak{g}_1^{\otimes n} \longrightarrow \mathfrak{g}_2$. 
The concrete equation of $U$ is computed by the perturbative expansion of 
the path integral.
The MC equation of the $L_{\infty}$-morphism is derived 
by using the WT identity induced from the quantum master equation as
\begin{eqnarray}
\lefteqn{
\sum_{\ell=0}^n\sum_{k=1}^{m-1}
 \sum_{i=0}^{m-k}
\sum_{\sigma\in \mathfrak{S}_{l, n-l}}\epsilon(\sigma)(-1)^{k(i+1)}(-1)^m
U_{l}(\alpha_{\sigma(1)},\dots,\alpha_{\sigma(l)})
}
\nonumber \\ &&
(F_0\otimes\cdots\otimes F_{i-1}
\otimes\, U_{n-l}(\alpha_{\sigma(l+1)},
\dots,\alpha_{\sigma(n)})(F_{i}\otimes\cdots\otimes F_{i+k})
\otimes F_{i+k+1}\otimes\cdots\otimes F_{m})
\nonumber \\
&=&
\sum_{i<j}\epsilon_{ij} U_{n-1}([\alpha_i,\alpha_j],\alpha_1,\dots,
\widehat\alpha_i,\dots,\widehat\alpha_j,\dots,\alpha_n)
(F_0\otimes\cdots\otimes F_m),~~~~
\label{LinftymorphismfromAKSZ}
\end{eqnarray}
where
\begin{eqnarray*}
&& (i \hbar)^{n+m-1} U_n(\alpha_1, \ldots, \alpha_n)
(F_0 \otimes \cdots \otimes F_m) 
=
\int \calO_x (F_0, \ldots, F_m) 
e^{\frac{i}{\hbar}S_0}
\frac{i}{\hbar} S_{\alpha_1}
\cdots 
\frac{i}{\hbar} S_{\alpha_n}.
\end{eqnarray*}
The map $U$ satisfying equation \eqref{LinftymorphismfromAKSZ} is nothing but
the $L_{\infty}$-morphism used in the proof of Theorem \ref{formalitytheorem}.

\section{Comments and Future Outlook}
\noindent
The AKSZ construction is a clear method for the construction and 
analysis of topological field theories in any dimension.
Although important aspects have been discussed here,
we could not consider all topics related to AKSZ sigma models.
We briefly list 
the subjects related to AKSZ sigma models that have not been discussed here.

The Poisson sigma model on a general Lie algebroid (the Lie algebroid Poisson sigma model) has been analyzed 
\cite{Bonechi:2005sj, Zucchini:2008hn, Velez:2011ed}.
Several versions of TFTs with a generalized geometric 
structure have been constructed  
in two, three and higher dimensions
\cite{Zucchini:2004ta, Zucchini:2005rh, Pestun:2006rj, 
Ikeda:2006pd, Ikeda:2007rn, Cattaneo:2009zx}.
The Rozansky-Witten theory has been 
formulated by the AKSZ construction in three dimensions
\cite{Qiu:2009zv}.
Open $p$-branes with worldvolume boundaries have been analyzed
\cite{Park2000au, Hofman:2002rv}.
A TFT with Dirac structure (a Dirac sigma model) has been formulated in Refs.
\citenum{Kotov:2004wz, Kotov:2010wr}.
A three-dimensional version of the A-model has been proposed \cite{Stojevic:2008qy}
and the relation between the doubled formalism and the AKSZ formalism has been analyzed \cite{Stojevic:2008hc}.
A topological sigma model with a Nambu-Poisson structure (the Nambu-Poisson sigma model) has been 
constructed \cite{Bouwknegt:2011vn, Schupp:2012nq}.
The Poisson (and symplectic) reduction has been discussed in terms of the AKSZ approach
\cite{Cattaneo:2007er, Cattaneo:2009zza, Zucchini:2007ie, 
Zucchini:2008cg, Zucchini:2010uu, BCZ}.
%

Many other geometric structures have been realized in the AKSZ construction
\cite{Zucchini:2005cq, Qiu:2010xi, 
Barnich:2010sw, Kallen:2010ff, Zucchini:2011aa}.

General structures of this formulation and applications to various aspects
of quantum field theories have been analyzed 
\cite{Batalin:2001fc, Edgren:2002xg, Edgren:2003nk, 
Ikeda:2004gp, Kotov:2007nr, Barnich:2009jy, 
Bonechi:2011um, Grigoriev:2012xg}.
The AKSZ construction on a discrete spacetime has been considered 
\cite{Bonechi:2009kx, Alekseev:2010ud}.
Categorical and graded versions of bundles 
related to the AKSZ method, called derived geometry, have been formulated 
\cite{PTVV}.
There are categorical and Chern-Weil formulations of the AKSZ construction \cite{Fiorenza:2011jr}.
The Wilson loop in the Chern-Simons theory has been formulated \cite{Alekseev:2012wc}.
A current algebra theory based on the supergeometric AKSZ formulation
has been constructed \cite{Ikeda:2011ax, Ikeda:2013vga}.
The AKSZ sigma models have been applied to analyze 
T-duality and R-flux in string theory
\cite{Mylonas:2012pg, Bessho:2015tkk}.
The AKSZ formalism has been used in the construction of higher spin theories
\cite{Boulanger:2011dd, Boulanger:2012bj, Bonezzi:2015igv}.

Supergeometry such as QP-manifolds is used to analyze the geometry of
 double field theory
\cite{Deser:2014mxa, Deser:2016qkw, Heller:2016abk}.
There are also recent papers that analyze AKSZ theories 
\cite{Alkalaev:2013hta, Bonavolonta:2013mza, 
Johnson-Freyd:2013oea, Mnev:2012qd}.


The geometric structures of
AKSZ theories have not yet been satisfactory analyzed.
Many geometric structures have been realized by 
AKSZ sigma models, but
there exist some structures for which the topological sigma model formulations have yet to be found.
For example, the Nambu bracket itself, which appears
in membrane theory, has not been constructed as 
a target space structure of an AKSZ type sigma model,
although the Nambu-Poisson tensor has been realized 
by the AKSZ sigma model on a manifold with boundary
\cite{Bouwknegt:2011vn}.
Here, we did not fully 
discuss AKSZ theories on an open manifold, although we note that 
they are important and
 related to higher categories.

Many analyses of the quantization of AKSZ sigma models can be found in the literature
\cite{Hofman:2002jz, Bonechi:2007ar, Bonechi:2009kx, Qiu:2009rf,
Qiu:2010xi, Kallen:2010ff},
but the analysis of the general AKSZ theory has not been completely understood.
The gauge fixing procedure is complicated. It
requires the BV formalism of component fields,
since the gauge fixing is not formulated by superfields.
Moreover, the moduli space of the observables in the path integral
is not clear in more than two dimensions,
and it is difficult to generalize the formality theorem.

Since gauge structures are algebroids,
in general, their structures are highly nonlinear.
Analysis of their structures, 
including their quantizations, is not so easy.
For complete quantizations, we must solve the problem of
globalization of algebroids to groupoids \cite{Cattaneo:2000iw, LiBland}.
Mathematical structures of algebroids and groupoids
in general dimensions should be analyzed.
Other important problems are the analysis of
nonperturbative effects, such as instantons or monopoles.

AKSZ sigma models have not only reformulated topological invariants,
but also led to the proposal of new topological or differential topological 
invariants.
Thus, analysis of these models may solve problems, 
such as the classification of differential topological manifolds.
It will also be important to clarify the relationship between the AKSZ formulation and the mathematical formulation of TFTs. \cite{Atiyah:1989vu}

In TFTs, mathematical and physical arguments are closely connected.
AKSZ sigma models are rich in potential, and they 
lead to a deeper understanding of the relationship between mathematics and physics.

\section*{Acknowledgments}
\noindent
This lecture is partially based on my lecture 
series at Tohoku and Keio Universities, and I
would like to thank them
for their hospitality and discussions.
The author would like to thank T.~Asakawa,
U.~Carow-Watamura,
T.~Bessho, K.~Koizumi, Y.~Maeda, M.~A.~Heller, S.~Sasa, M.~Sato,
K.~Uchino, X.~Xu and S.~Watamura for valuable comments and discussions.
He would like to thank Y. Maeda for encouraging me 
to write this lecture note, and thank 
U.~Carow-Watamura, M.~A.~Heller and S.~Watamura for careful reading of 
this manuscript.
This work was supported by the
Maskawa Institute, Kyoto Sangyo University
and supported by the research promotion program grant 
at Ritsumeikan University.

\appendix

\section{Appendix: Formulas in Graded Differential Calculus}
\noindent
We summarize the formulas of graded symplectic geometry.

\subsection{Basic definitions}
\noindent
Let $z$ be a local coordinate on a graded manifold $\calM$.
A differential on a function is defined as follows.
\begin{eqnarray}
d f(z) &=& d z^a \frac{\ld f}{\partial z^a}.
\end{eqnarray}
A vector field $X$ is expanded using local coordinates, as follows.
\begin{eqnarray}
X = X^a(z) \frac{\ld}{\partial z^a}.
\end{eqnarray}
The interior product is defined using differentiation by
the following graded vector field on $T[1]\calM$,
\begin{eqnarray}
\iota_X &=& (-1)^{|X|}
X^a(z) \frac{\ld}{\partial d z^a},
\end{eqnarray}
where we define
$\frac{\ld}{\partial d z^a} d z^b = \delta^b{}_a$.
For a graded differential form $\alpha$, we denote by $|\alpha|$ the total degree 
(form degree plus degree by grading) of $\alpha$.
Note that $|d| = 1$, $|d z^a| = |z^a| +1$
and
$|\iota_X|= |X| -1$.
%
%
For vector fields, $X = X^a(z) \frac{\ld}{\partial z^a}$
and $Y = Y^a(z) \frac{\ld}{\partial z^a}$,
the graded Lie bracket is
\begin{eqnarray}
[X, Y] 
&=& X^a \frac{\ld Y^b}{\partial z^a} \frac{\ld}{\partial z^b}
- (-1)^{|X||Y|} Y^a \frac{\ld X^b}{\partial z^a} \frac{\ld}{\partial z^b}.
\end{eqnarray}
We obtain the following formula,
\begin{eqnarray}
Xf &=& (-1)^{|X|} \iota_X df
= (-1)^{(|f|+1)|X|} df(X),
\label{differentialoffunction}
\end{eqnarray}
where
\begin{eqnarray}
d z^a \left(\frac{\ld}{\partial z^b} \right) &=& \delta^a{}_b.
\end{eqnarray}
\begin{proof}
We prove \eeqref{differentialoffunction}.
Since $Xf = X^a(z) \frac{\ld f}{\partial z^a}$,
we have
\begin{eqnarray}
(-1)^{|X|} \iota_X df
&=& (-1)^{|X|} (-1)^{|X|} X^a(z) 
\frac{\ld}{\partial d z^a} \left(
d z^a \frac{\ld f}{\partial z^a} \right).
\end{eqnarray}
Therefore,
\begin{eqnarray}
df(X)
&=&
d z^a \frac{\ld f}{\partial z^a}
\left(
X^b(z) \frac{\ld}{\partial z^b}
\right)
\nonumber \\
&=&
(-1)^{(|f|-|z|)|X|} 
\left[
d z^a \left(X^b(z)
\frac{\ld}{\partial z^b}
\right) 
\right]
\frac{\ld f}{\partial z^a}
\nonumber \\
&=&
(-1)^{(|f|-|z|)|X|} 
(-1)^{(|X| - |z|)(|z|+1)} 
X^b(z) 
\left[
d z^a \left(
\frac{\ld}{\partial z^b}
\right)
\right]
\frac{\ld f}{\partial z^a}
\nonumber \\
&=&
(-1)^{(|f|+1)|X|} 
X^a(z) \frac{\ld f}{\partial z^a}.
\end{eqnarray}
\end{proof}

\subsection{Cartan formulas}
\noindent
The Lie derivative is defined by
\begin{eqnarray}
L_X = \iota_X d - (-1)^{(|X|-1) \times 1} d \iota_X
= \iota_X d + (-1)^{|X|} d \iota_X.
\end{eqnarray}
Its degree is $|L_X| = |X|$.

Let $\alpha$ and $\beta$ be graded differential forms.
We can show the following graded Cartan formulas,
\begin{eqnarray}
&& \alpha \wedge \beta = (-)^{|\alpha||\beta|}\beta \wedge \alpha,
\\
&& d (\alpha \wedge \beta) =
d \alpha \wedge \beta + (-1)^{|\alpha|} \alpha \wedge d \beta,
\\
&& \iota_X (\alpha \wedge \beta) =
\iota_X \alpha \wedge \beta + (-1)^{|\alpha|(|X|+1)} 
\alpha \wedge \iota_X \beta,
\\
&& L_X (\alpha \wedge \beta) =
L_X  \alpha \wedge \beta + (-1)^{|\alpha||X|} 
\alpha \wedge L_X \beta,
\\
&& L_X d = (-1)^{|X|}  d L_X,
\\
&& \iota_X \iota_Y - (-1)^{(|X|- 1)(|Y|-1)} \iota_Y \iota_X = 0,
\\
&& L_X \iota_Y - (-1)^{|X|(|Y|- 1)} \iota_Y L_X = \iota_{[X, Y]},
\\
&& L_X L_Y - (-1)^{|X||Y|} L_Y L_X = L_{[X, Y]}.
\end{eqnarray}

\subsection{Differential forms}
\noindent
Let $\alpha = d z^{a_1} \wedge \cdots d z^{a_m} \alpha_{a_1 \cdots a_m}(z)$
be an $m$-form on $\mathcal{M}$.
The contraction of $\alpha(X, -, \cdots, -)$ with a vector field $X$ on $\mathcal{M}$ is
\begin{eqnarray}
\alpha(X, -, \cdots, -)
= (-1)^{|X|(|\alpha|+1)} \iota_X \alpha(-, \cdots, -).
\label{interiorderivativeondform}
\end{eqnarray}

\begin{proof}
\begin{eqnarray}
\alpha(X, -, \cdots, -)
&=&
d z^{a_1} \wedge \cdots 
d z^{a_m} \alpha_{a_1 \cdots a_m}(z)
\left (X^b 
\frac{\ld}{\partial z^b}\right)
\nonumber \\
&=&
(-1)^{|X|(|\alpha|-|z|-1)} 
d z^{a_1} 
\left(X^b
\frac{\ld}{\partial z^b}
\right) 
d z^{a_2} \wedge \cdots 
d z^{a_m} \alpha_{a_1 \cdots a_m}(z)
\nonumber \\
&=&
(-1)^{|X|(|\alpha|-|z|-1)} 
(-1)^{(|X| - |z|)(|z|+1)} 
\nonumber \\ && 
\times X^{a_1} 
d z^{a_2} \wedge \cdots 
d z^{a_m} \alpha_{a_1 \cdots a_m}(z)
\nonumber \\
&=&
(-1)^{|X||\alpha|} 
X^{a_1} 
d z^{a_2} \wedge \cdots 
d z^{a_m} \alpha_{a_1 \cdots a_m}(z)
\nonumber \\
&=&
(-1)^{|X||\alpha|} 
(-1)^{|X|}
\iota_X \alpha.
\end{eqnarray}
\end{proof}

By induction 
using \eeqref{interiorderivativeondform}, 
we obtain the following general formula,
\begin{align}
& \alpha(X_m, X_{m-1}, \cdots, X_1)
= - (-1)^{\sum_{i=1}^m |X_i|(|\alpha|+i)} \iota_{X_m} 
\cdots \iota_{X_1} \alpha,
\\
%
%
&
\alpha(X_m, \cdots, X_j, \cdots, X_i, \cdots X_1)
= - (-1)^{|X_i||X_j|} 
\alpha(X_m, \cdots, X_i, \cdots, X_j, \cdots X_1).
\end{align}
In particular, if $\alpha$ is a $2$-form, we obtain
\begin{eqnarray}
\alpha(X, Y)
&=& - (-1)^{|X||Y|} \alpha(Y, X).
\end{eqnarray}

\subsubsection{Exterior derivatives}
\noindent
Recall the exterior derivative of a function was given by \eeqref{differentialoffunction},i.e.
\begin{eqnarray}
df(X)
= (-1)^{|X|(|f|+1)} Xf.
\end{eqnarray}
Let $\alpha$ be a $1$-form on $\mathcal{M}$.
Then, from the Cartan formulas, we obtain
\begin{align}
d \alpha(X_1, X_2)
&=
(-1)^{|X_1||\alpha|} X_1 \alpha(X_2)
- (-1)^{|X_2||\alpha|} (-1)^{|X_1||X_2|} X_2 \alpha(X_1)
- \alpha([X_1, X_2]).
\end{align}
For a $2$-form $\alpha$, the formula gives
\begin{align}
d \alpha(X_1, X_2, X_3)
&=
(-1)^{|X_1|(|\alpha| +1)} X_1 \alpha(X_2, X_3)
- (-1)^{|X_2|(|\alpha| +1)} (-1)^{|X_1||X_2|} X_2 \alpha(X_1, X_3)
\nonumber \\ &
\quad
+ (-1)^{|X_3|(|\alpha| +1)} (-1)^{(|X_1|+|X_2|)|X_3|} X_3 \alpha(X_1, X_2)
- \alpha([X_1, X_2], X_3)
\nonumber \\ &
\quad
+ (-1)^{|X_2||X_3|} \alpha([X_1, X_3], X_2)
- (-1)^{|X_1|(|X_2| + |X_3|)} \alpha([X_2, X_3], X_1).
\end{align}
Let $\alpha = d z^{a_1} \wedge \cdots d z^{a_m} \alpha_{a_1 \cdots a_m}(z)$
be an $m$-form on $\mathcal{M}$. 
Then, we can prove the following formula by induction,
\begin{eqnarray}
d \alpha(X_1, X_2, \cdots, X_m)
&=& 
\sum_{i=1}^m (-1)^{i-1} (-1)^{|X_i|(|\alpha|+m)} 
(-1)^{\sum_{k=1}^{i-1} |X_i||X_k|}
X_i \alpha(X_1, \cdots, \hat{X_i}, \cdots, X_m)
\nonumber \\
&& 
+
\sum_{i < j}
(-1)^{i+j} 
(-1)^{\sum_{k=1}^{i-1} |X_i||X_k|
+ \sum_{l=1, l \neq j}^{j-1} |X_j||X_l|}
\nonumber \\ && \times 
\alpha([X_i, X_j], 
\cdots, \hat{X_i}, \cdots, \hat{X_j}, \cdots, X_m).
\end{eqnarray}

\subsection{Graded symplectic form and Poisson bracket}
\noindent
Let $\omega$ be a symplectic form of degree $n$.
Since $\omega$ is a $2$-form, its total degree is $|\omega| = n+2$.
Let $z = (q^a, p_a)$ be
Darboux coordinates such that 
$|q| + |p| =n$. Then, we obtain
\begin{eqnarray}
\omega &=& (-1)^{|q|(|p|+1)} d q^a \wedge d p_a 
= (-1)^{n|q|} d q^a \wedge d p_a 
\nonumber \\
&=& (-1)^{n|q|}(-1)^{(|q|+1)(|p|+1)} d p_a \wedge d q^a
= (-1)^{|p|+1} d p_a \wedge d q^a.
\end{eqnarray}
The Liouville $1$-form $\omega = - d \vartheta$ is then given by
\begin{eqnarray}
\vartheta 
&=& 
(-1)^{|p|} p_a d q^a
= - (-1)^{n +1 - |q|} p_a d q^a
= (-1)^{|q||p|} d q^a p_a
\\
&=& - (-1)^{|q|(|p|+1)} q^a d p_a
= - d p_a q^a.
\end{eqnarray}

The Hamiltonian vector field $X_f$ of a function $f$ is defined by
\begin{eqnarray}
\iota_{X_f} \omega &=& - d f.
\end{eqnarray}
Its total degree is $|X_f| = |f|-n$.
In order to obtain the Darboux coordinate expression of $X_f$,
we take a local coordinate expression $X = X_a \frac{\ld }{\partial p_a}
+
Y^a \frac{\ld }{\partial q^a}$.
Then we obtain
\begin{eqnarray}
\iota_{X_f} \omega &=& 
 \left(
(-1)^{|X| + p} X_a \frac{\ld }{\partial d p_a}
+ 
(-1)^{|X| + q} 
Y^a \frac{\ld }{\partial d q^a}
\right)
\cdot \left( (-1)^{n|q|} d q^a \wedge d p_a \right)
\nonumber \\
&=& - dq^a \frac{\ld f}{\partial q^a}
- dp_a \frac{\ld f}{\partial p_a}.
\end{eqnarray}
By solving this equation, we finally obtain
\begin{eqnarray}
X_f &=& 
\frac{f \rd}{\partial q^a}\frac{\ld }{\partial p_a}
- (-1)^{|q||p|}
\frac{f \rd}{\partial p_a}\frac{\ld }{\partial q^a}.
\end{eqnarray}
Here, $\frac{f \rd}{\partial q^a} = 
(-1)^{(|f|-q)q} \frac{\ld f}{\partial q^a}$
is the right derivative.

The graded Poisson bracket is defined by
\begin{eqnarray}
\sbv{f}{g} &=& X_f g
= (-1)^{|f|+n} \iota_{X_f} dg
= (-1)^{|f|+n + 1} \iota_{X_f} \iota_{X_g} \omega.
\end{eqnarray}
It satisfies
\begin{eqnarray*}
\sbv{f}{g}&=&-(-1)^{(|f| - n)(|g| - n)} \sbv{g}{f},\\
\sbv{f}{g  h}&=&\sbv{f}{g} h
+ (-1)^{(|f| - n)|g|} g \sbv{f}{h},\\
\{f,\{g,h\}\}&=&\{\{f,g\},h\}+(-1)^{(|f|-n)(|g|-n)}\{g,\{f,h\}\}.
\end{eqnarray*}
For the Darboux coordinates, we get the relations 
\begin{eqnarray}
\sbv{q^a}{p_b} &=& \delta^a{}_b,
\qquad
\sbv{p_b}{q^a} = - (-1)^{|q||p|}\delta^a{}_b.
\end{eqnarray}
For the functions $f = f(q, p)$ and $g = g(q, p)$,
the graded Poisson bracket is given by
\begin{eqnarray}
\sbv{f}{g} &=& 
\frac{f \rd}{\partial q^a}\frac{\ld g}{\partial p_a}
- (-1)^{|q||p|}
\frac{f \rd}{\partial p_a}\frac{\ld g}{\partial q^a}.
\end{eqnarray}


$X$ is called a \textit{symplectic vector field} if $L_X \omega = 0$, 
i.e., $d \iota_X \omega =0$.
Let $X, Y$ be symplectic vector fields. 
Then, $[X,Y]$ is the Hamiltonian vector field for 
$ - (-1)^{|X|}\iota_X \iota_Y \omega$.
\begin{proof}
\begin{eqnarray}
\iota_{[X, Y]} \omega &=& 
(L_X \iota_Y - (-1)^{|X|(|Y|-1)} \iota_Y L_X ) \omega
= (-1)^{|X|} d \iota_X \iota_Y \omega
\nonumber \\
&=& - d [- (-1)^{|X|} \iota_X \iota_Y \omega].
\end{eqnarray}
\end{proof}
If $X = X_f$, $Y = X_g$ are Hamiltonian vector fields, then the following equation holds,
\begin{eqnarray}
\iota_{[X_f, X_g]} \omega &=& 
(-1)^{|f| + n} d \iota_{X_f} \iota_{X_g} \omega.
\end{eqnarray}
Therefore, we get
\begin{eqnarray}
X_{\sbv{f}{g}}&=&  - [X_f, X_g].
\end{eqnarray}
Since $\iota_{X_f} \iota_{X_g} \omega
= - (-1)^{|f|n + |g|(n+1)} \omega(X_g, X_f)$, we easily obtain
\begin{eqnarray}
\sbv{f}{g} &=& (-1)^{|f| +n + 1} \iota_{X_f} \iota_{X_g} \omega
\nonumber \\
&=&
(-1)^{(|f|+ |g|)(n+1)} \omega(X_g, X_f)
\nonumber \\ 
&=&
(-1)^{|f||g| + n+1} \omega(X_f, X_g).
\end{eqnarray}

We consider the AKSZ construction on $\Map(\calX, \calM)$.
Let $D$ be a differential on $\calX$. It can be locally expressed as
$D = \theta^{\mu}\frac{\partial}{\partial \sigma^{\mu}}$.
We denote by $\hat{D}$ the vector field on $\Map(\calX, \calM)$
of degree $1$ which is induced by $D$.
The following equation holds,
\begin{eqnarray}
\sbv{\iota_{\hat{D}} \mu_* \ev^* \vartheta}{\mu_* \ev^* f} 
&=& - \iota_{\hat{D}} \mu_* \ev^* d f
\nonumber \\ 
&=& \int d^{n+1}\sigma d^{n+1}\theta \bbd f(\sigma, \theta),
\end{eqnarray}
for $f \in C^{\infty}(\calM)$.
\begin{proof}
$S_0 = \iota_{\hat{D}} \mu_* \ev^* \vartheta$ is a Hamiltonian 
for the vector field $\hat{D}$, i.e.,
$X_{S_0} = \hat{D}$. Therefore, we have
\begin{eqnarray}
\sbv{\iota_{\hat{D}} \mu_* \ev^* \vartheta}{\mu_* \ev^* f} 
&=& 
\sbv{S_0}{\mu_* \ev^* f} 
\nonumber \\ &=&
(-1)^{|S_0|} \iota_{\hat{D}} \iota_{X_{\mu_* \ev^* f}} \bomega
\nonumber \\ &=&
- \iota_{\hat{D}} \mu_* \ev^* d f.
\end{eqnarray}

\end{proof}

\newcommand{\bibit}{\sl}



\begin{thebibliography}{99}
\bibitem{Uchino}
The author would like to thank to K~.Uchino for discussions.
\end{thebibliography}
\end{document}